\documentclass[11pt,journal, onecolumn]{IEEEtran}
\usepackage[margin=1in]{geometry}
\usepackage{amsmath,amssymb,mathtools,amsfonts}
\usepackage{setspace}
\usepackage[algo2e,ruled,vlined,linesnumbered]{algorithm2e}
\usepackage{hyperref}
\usepackage{booktabs}
\usepackage{thm-restate}
\usepackage{cleveref}
\usepackage{graphicx,psfrag}
\usepackage{pstool}
\usepackage{calc}
\usepackage{comment}
\usepackage{times}
\usepackage{hyperref,color,cleveref}
\usepackage{multicol}
\graphicspath{ {./figs} }
\usepackage{bbm}

\usepackage[numbers]{natbib} 
\usepackage{enumitem}

\usepackage{amsthm}

\usepackage{thmtools}
\makeatletter
\newtheorem*{rep@theorem}{\rep@title}
\newcommand{\newreptheorem}[2]{%
\newenvironment{rep#1}[1]{%
 \def\rep@title{#2 \ref{##1}}%
 \begin{rep@theorem}}%
 {\end{rep@theorem}}}
\makeatother

\newreptheorem{theorem}{Theorem}
\newtheorem{lemma}{Lemma}
\newreptheorem{lemma}{Lemma}

\newreptheorem{fact}{Fact}
\newtheorem{corollary}{Corollary}

\theoremstyle{definition}

\newtheorem{remark}{Remark}

\newcommand{\eps}{\epsilon}

\def\erdos{Erd\H os}
\def\renyi{R\'enyi }
\def\E{{\mathsf E}}
\def\P{{\mathsf P}}
\def\cg{{\mathcal{G}}}

\def\dist{{\mathrm{dist}}}
\def\ca{{\mathcal{A}}}
\def\cb{{\mathcal{B}}}
\def\cc{{\mathcal{C}}}
\def\cd{{\mathcal{D}}}
\def\ce{{\mathcal{E}}}
\def\cf{{\mathcal{F}}}
\def\ci{{\mathcal{I}}}
\def\cj{{\mathcal{J}}}

\def\ch{{\mathcal{H}}}
\def\cl{{\mathcal{L}}}
\def\cn{{\mathcal{N}}}

\def\ct{{\mathcal{T}}}
\def\cx{{\mathcal{X}}}
\def\co{{\mathcal{O}}}
\def\cm{{\mathcal{M}}}

\def\rmic{{\mathrm{IC}}}
\def\bl{{\bar{\lambda}}}

\def\gw{{\mathrm{GW}}}
\def\indi{{\mathbbm{1}}}
\def\poi{{\mathrm{Poi}}}
\def\sto{{\stackrel{\mathrm{sto.}}{\le}}}
\def\stoge{{\stackrel{\mathrm{sto.}}{\ge}}}

\def\bin{{\mathrm{Binom}}}

\def\dtv{{d_\mathrm{TV}}}
\def\dmax{{d_{\max}}}

\def\dist{\mathrm{dist}}
\def\emp{{\emptyset}}
\def\tdt{{\td{\ct}}}
\def\tdh{{\td{\ch}}}

\def\barn{{\bar{\cn}}}
\def\tic{{\mathcal{T}^{\mathrm{IC}}}}
\def\tdct{\tilde{\mathcal{T}}}
\def\tdcg{\tilde{\mathcal{G}}}

\def\Nup{N^\uparrow}
\def\tdNup{\td{N}^\uparrow}
\def\rmin{\mathrm{IN}}
\def\pass{\mathrm{PASS}}
\def\no{\mathrm{NO}}
\newcommand{\td}[1]{\tilde{#1}}
\newcommand{\cond}{\mathchoice{\,\vert\,}{\mspace{2mu}\vert\mspace{2mu}}{\vert}{\vert}}

\allowdisplaybreaks
\title{Diffusion-Network Alignment: An Efficient Algorithm and Explicit Probability Bounds}

\author{Ziao Wang, Lei Ying
\thanks{Ziao Wang is with the Department of Electrical and Computer Engineering, University of Michigan, Ann Arbor, MI 48109, USA (email: ziaow@umich.edu).}
\thanks{Lei Ying is with the Department of Electrical and Computer Engineering, University of Michigan, Ann Arbor, MI 48109, USA (email: leiying@umich.edu).}
}

\begin{document}

\maketitle

\begin{abstract}
This paper studies a variation of the classic network alignment problem, named diffusion–network alignment. The goal is to align the vertices of a rooted diffusion tree to the vertices of a network, where the diffusion tree could be from a communication trace or contact tracing, and the network could be an online or offline social network. Different from the classic network alignment where both networks are fully observed, this model captures the information asymmetry of two networks. To solve this problem, this paper presents an efficient algorithm based on tree correlation tests to extract alignment information from local neighborhoods. We analyze the performance of the algorithm in the sparse graph regime and show that with high probability, all matched pairs are correct.  Furthermore, for each vertex on the diffusion tree, this paper establishes an explicit lower bound on the probability that the vertex is correctly matched.  These lower bounds are depth-dependent and increase as vertices get closer to the root. 
\end{abstract}

\section{Introduction}
\emph{\textbf{Network alignment.}}
The network alignment (or graph matching) problem refers to the problem of recovering the hidden correspondence between the vertices of two correlated graphs. 
A primary motivation for this problem comes from social network deanonymization~\cite{Nar-Shm-deanonymizing2009,Kor-Lat-Reconciliation2014}, where one seeks to match the users of an anonymized social network (e.g. Twitter) to a correlated reference network in which user identities are known (e.g. Facebook).
Related alignment problems also arise in areas such as bioinformatics~\cite{singh2007pairwise} and natural language processing~\cite{Hag-Ng-robust2005}.

A classical mathematical model for network alignment is the correlated \erdos--\renyi (ER) graph pair model~\cite{Ped-Gro-privacy2011}. In this model, a base ER graph is first sampled on a latent vertex set, and two observed graphs are then generated by independently subsampling the edges of this base graph. As a result, each observed graph is an ER graph, while edges of the two observed graphs are correlated through the shared base graph and the latent vertex correspondence. This goal is to find the latent vertex correspondence (vertex alignment) based on the topology of the two observed graphs.
A sequence of works has characterized the information-theoretic limits of this problem across different graph density regimes~\cite{Cul-Kiy-improved2016,Cul-Kiy-exact2017,settling-TIT,ding2023matching}.
The other sequence of works studies the problem from an algorithmic perspective, and identifies a computational threshold for the sparse graph regime. In particular, when the edge correlation between the two graphs is above $\sqrt{\alpha}$, where $\alpha\approx0.338$ is the Otter's constant~\cite{otter1948}, polynomial-time algorithms are known to achieve reliable recovery~\cite{mao2026,ganassali2024,ganassali2024statistical}. Below this threshold,~\cite{li2025algorithmic} provides evidence for the non-existence of such polynomial-time algorithms.

\emph{\textbf{Diffusion-network alignment.}}
In the aforementioned line of works, a common assumption is that both graphs are \emph{fully observable}. This paper studies a variation of the problem where one of the networks is a rooted diffusion tree. This rooted diffusion tree could be a communication trace, a trace from epidemic contact tracing, blockchain transition history. 
This variation captures the information asymmetry of the network alignment problem.  This paper aims at answering the following question: can one align the vertices of an observed diffusion tree  with those of a correlated reference network?

Motivated by the above question, we consider the following \emph{diffusion-network alignment} model. We consider two networks generated from the correlated \erdos--\renyi graph pair model. To model the information diffusion process on the first network, we employ the Independent Cascade (IC) Model, a classic model first proposed in~\cite{kempe2003} for information diffusion on networks. In this model, information propagates from a source node through repeated, probabilistic activations of neighboring nodes. The activated vertices, together with the edges through which they are activated, form a tree-structured diffusion cascade. We observe only this diffusion tree from the first network, while the second network is fully observed. The goal is to recover the correspondence between the vertices of the diffusion tree and those of the second network. In this work, we study diffusion–network alignment from an algorithmic perspective, and focus on studying efficient algorithms for this problem.

\emph{\textbf{Relation to classic network alignment.}} Diffusion–network alignment is closely related to classical network alignment, since the diffusion tree is sampled from the first network and inherits its structural information, which is correlated with the second network we observe. Therefore, it is natural to consider whether existing algorithmic approaches for network alignment can be adapted to diffusion-network alignment. In the following, we focus on discussing two algorithmic approaches that are known to achieve reliable recovery at the computational threshold $\sqrt{\alpha}$, while deferring a broader review of the literature to Appendix~\ref{appd:related}.

The first approach is \emph{subgraph counting}~\cite{mao2026}, which constructs vertex feature vectors by aggregating weighted counts of specific subgraph structures in the network and matches the vertices based on the similarity between their feature vectors. However, this approach assumes access to the entire network to perform the weighted counts and is therefore incompatible with diffusion network alignment, where the observation of the first network is restricted to a small subset of vertices.

The second approach is known as \emph{tree correlation testing}~\cite{ganassali2024,ganassali2024statistical,maier2025}, which reduces the task of matching two vertices to a hypothesis test of whether their tree-like local neighborhoods are generated independently or correlatedly. In diffusion–network alignment, the diffusion tree retains a tree-structured subgraph of the neighborhood of the diffusion source. This makes tree correlation testing a natural choice for the diffusion setting, as it operates directly on tree-like observations and does not require access to the entire graph. 
In particular, the tree correlation testing approach developed in~\cite{maier2025} can be adapted to match the root of the diffusion tree.
However, this direct application does not extend to non-root vertices in the diffusion tree.

This motivates our proposed algorithm, which exploits the fact that the diffusion cascade provides a natural structure for information propagation on the diffusion tree. The algorithm proceeds in two passes. In the upward pass, it applies tree correlation tests to collect local matching information starting from the bottom of the diffusion tree, and propagates this information upward toward the root. In the downward pass, the algorithm starts from the root and propagates the accumulated matching information toward the leaves. This two-pass procedure allows us to go beyond the direct application of existing tree correlation tests. It provides guarantees for matching non-root vertices in the diffusion tree, and also yields a stronger guarantee for matching the root.

\emph{\textbf{Our results.}}
In this work, we introduce an efficient algorithm for the diffusion-network alignment problem, building on the approach motivated above. We analyze the proposed algorithm in the constant-degree graph regime and provide performance guarantees from two perspectives.
\begin{itemize}
\item \emph{Depth-dependent probability bounds for correct matchings:}
Under the assumption that the effective correlation between the diffusion tree and the second network is above $\sqrt{\alpha}$, we derive depth-dependent lower bounds on the probability of correct matching for each individual vertex in the diffusion tree. These depth-dependent lower bounds are constants bounded away from zero and become larger for vertices closer to the root.
    \item \emph{Global correctness property:} 
    We also show that, with high probability, the proposed algorithm is globally correct in the sense that for every vertex in the diffusion tree, it either outputs the true corresponding vertex or outputs no match at all.
\end{itemize}

\emph{\textbf{Comparison to theoretical guarantees for constant-degree network alignment.}}
In the literature for network alignment in the constant-degree regime, algorithms are often evaluated in terms of achieving \emph{partial recovery}, i.e., a non-vanishing fraction of vertices should be correctly aligned.
However, partial recovery may not always be possible for the diffusion-network alignment problem. For example, the diffusion process may produce a cascade containing only the root vertex, which occurs with a non-vanishing probability. Therefore, high-probability guarantees for recovering a positive fraction of the vertices are not possible. For this reason, this paper considers vertex-wise guarantees: we provide {\em explicit} constant lower bounds on the success probability for matching each individual vertex, conditioned on its depth in the diffusion tree, and we show that the algorithm produces no incorrect matchings with high probability. This notion provides a performance guarantee that is comparable in spirit to partial recovery, while being adapted to the diffusion–network alignment problem.

\section{Model}\label{sec:model}
In this section, we introduce the proposed diffusion-network alignment model.

\emph{\textbf{Correlated \erdos--\renyi Graph Pair Model.}}
Let $n$ be a positive integer, $\lambda$ and $s$ be two positive real numbers with $s\in(0,1]$. Throughout this paper, we assume that both $\lambda$ and $s$ are constant independent of $n$. The correlated \erdos--\renyi graph pair model $\mathrm{CER}(n,\lambda,s)$ is defined as follows: We first sample a base graph $\cg\sim\mathrm{ER}(n,\frac{\lambda}{ns})$, where $\cg$ has $n$ vertices with indices from $[n]:=\{1,\ldots,n\}$, and an edge is generated between each pair of vertices independently with probability $\frac{\lambda}{ns}$.
Given the base graph $\cg$, the graph $\cg_1$ is obtained by subsampling each edge in $\cg$ independently with probability $s$. The same subsampling process is then independently repeated to generate $\cg_2$.
Then a permutation $\pi^*:[n]\rightarrow[n]$ is chosen uniformly at random and used to permute the vertex labels in $\cg_2$ to obtain the graph $\td{\cg}_2$. The model finally outputs the two correlated graphs $\cg_1$ and $\td{\cg}_2$ both with marginal distribution $\mathrm{ER}(n,\frac{\lambda}{n})$, and we denote $(\cg_1,\td{\cg}_2)\sim \mathrm{CER}(n,\lambda,s)$.

\emph{\textbf{Independent Cascade (IC) Model.}}
Given the graph pair $(\cg_1,\td{\cg}_2)\sim \mathrm{CER}(n,\lambda,s)$,
we use the IC model to simulate a diffusion process that starts at the vertex $1$ in the graph $\cg_1$.
At time step $0$, the vertex $1$ in $\cg_1$ becomes active (marked red in Figure~\ref{fig:IC}), indicating that it is the source of diffusion, while all other vertices are inactive (marked black in Figure~\ref{fig:IC}). The model proceeds in discrete time steps.
At each time step $k\ge 1$, for each edge $(u,v)$ such that $u$ is activated at time $k-1$ and $v$ is inactive, a single activation attempt is made along $(u,v)$, which succeeds independently with probability $q$. If at least one such attempt succeeds, vertex $v$ becomes activated at time $k$. Here $q\in (0,1]$ is a model parameter, and it is assumed to be a constant independent of $n$.
In Figure~\ref{fig:IC}, the edges with successful activation attempts are represented by blue arrows.
The process ends at the time step in which no new vertex is activated.
\begin{figure}[htbp]
    \centering
    \includegraphics[width=1\linewidth]{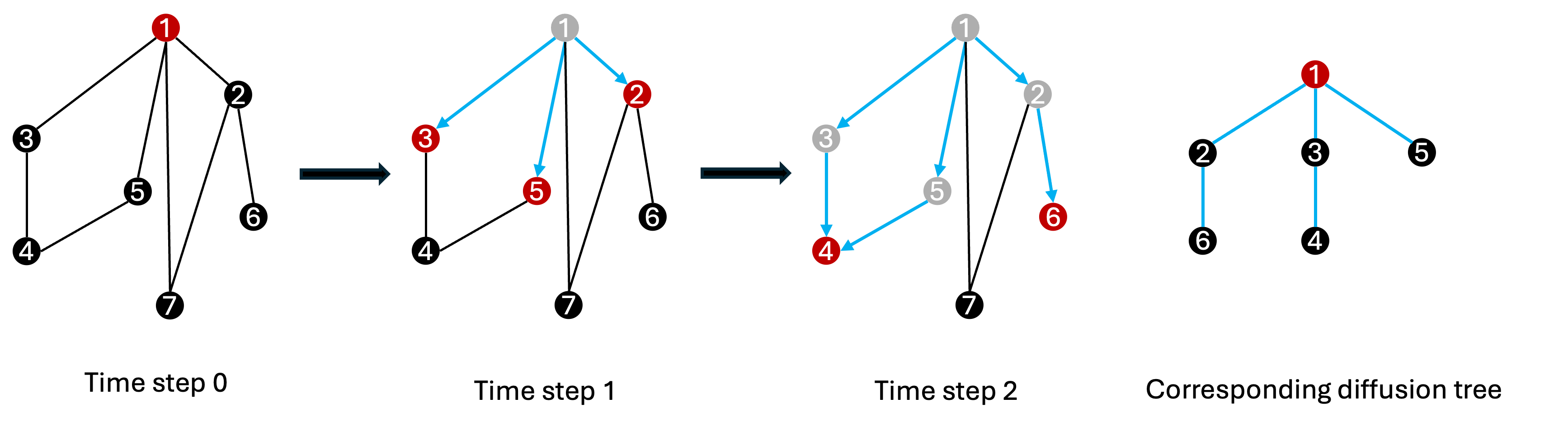} 
    \caption{An example of the IC model and its corresponding diffusion tree.}
    \label{fig:IC}
\end{figure}

\emph{\textbf{Diffusion Tree.}}
The IC model defined above is captured by the \emph{diffusion tree}, denoted $\tic$. The diffusion tree $\tic$ is rooted at vertex $1$, and for each $k\ge 1$, the set of vertices at depth $k$ is exactly the set of vertices that are activated at the time step $k$ of the IC process. For each non-root vertex in $\tic$, its parent is set to be the vertex that activated it in the IC process. However, it is possible that a vertex has multiple active neighbors, and get activated by multiple of them in the same time step. For example, in Figure~\ref{fig:IC}, vertex $4$ is activated by vertices $3$ and $5$ in time step $2$. To ensure the tree structure of $\tic$, we introduce a canonical order among the vertices at the same depth to select a unique parent for each vertex. The order is defined iteratively.
For each depth $k\ge 1$, a vertex at depth $k$ has a higher order if its parent has a higher order among the vertices at depth $k-1$, and for multiple vertices with the same parent, the vertex with a smaller index has a higher order. In Figure~\ref{fig:IC}, vertices that appear further to the left in the diffusion tree correspond to higher canonical order. When a vertex is activated by multiple neighbors in the same time step, its unique parent is chosen to the one with the highest order among those vertices which activated it.

\emph{\textbf{Diffusion-network Alignment.}}
In the diffusion-network alignment problem, we observe the diffusion tree $\tic$ and the graph $\td{\cg}_2$, and have access to the model parameters $\lambda$, $s$ and $q$. The goal is to find the corresponding vertex $\pi^*(u)$ in $\td{\cg}_2$ for each vertex $u$ in the observed diffusion tree $\tic$.

\emph{\textbf{Notation.}}
Recall that the vertices of both graphs $\cg_1$ and $\td{\cg}_2$ are indexed by $[n]$. For each index $i\in[n]$, we use $i$ to denote the vertex with index $i$ in $\cg_1$ (and also $\tic$), and use $\td{i}$ to denote the vertex with index $i$ in $\td{\cg}_2$. Then following this notation, the vertex $i$ in $\cg_1$ corresponds to the vertex $\td{j}$ in $\td{\cg_2}$ if $\pi^*(i)=j$. We use $V$ (resp. $\td{V}$ and $V^\rmic$) to denote the set of vertices $\cg_1$ (resp. $\td{\cg}_2$ and $\tic$). For each $d\ge 0$, we use $V_d$ to denote the set of vertices at depth $d$ of $\tic$, and define $V_{\le d}=\cup_{d'=0}^dV_{d'}$.

\section{Preliminary: Tree Correlation Tests}\label{sec:preliminary}
Before presenting our algorithm for diffusion-network alignment, we first introduce the preliminaries of the tree correlation testing problem, which is known as a classic tool for network alignment in the sparse regime.

\emph{\textbf{Poisson Galton--Watson trees.}}
Let $\lambda\in\mathbb{R}^+$. The Galton--Watson branching tree with offspring distribution $\poi(\lambda)$, denoted $\gw_{\lambda}$, is a distribution over unlabeled rooted trees. The tree has a distinguished vertex as its root.
Then iteratively for each vertex $v$ in the tree, an independent random variable $X_v\sim\poi(\lambda)$ is sampled, representing the number of children attached to $v$. 
For each $d\ge 0$, we use $\gw_{d,\lambda}$ to denote the depth-$d$ truncation of $\gw_\lambda$, which is the distribution over the rooted trees obtained by generating a tree from $\gw_\lambda$ and remove all the vertices with distance greater than $d$ from the root. We define the extinction probability $p^{\mathrm{ext}}_{d,\lambda}$ as the probability that there exist no vertices at distance $d$ from the root under the law $\gw_\lambda$, and define $p^{\mathrm{ext}}_\lambda:=\lim_{d\rightarrow\infty}p^{\mathrm{ext}}_{d,\lambda}$.

\emph{\textbf{Independent and correlated Galton--Watson tree distributions.}}
Let $d\in\mathbb{N}$, $\lambda\in\mathbb{R}^+$ and $s,\td{s}\in(0,1]$.
With the definition of $\gw_{d,\lambda}$, we are now ready to define the independent tree pair distribution $\mathbb{Q}^{(\lambda,s,\td{s})}_d$ and the correlated tree pair distribution $\mathbb{P}^{(\lambda,s,\td{s})}_d$ considered in the tree correlation test problem.
\begin{itemize}
    \item  Independent tree pair distribution $\mathbb{Q}^{(\lambda,s,\td{s})}_d$: Under this distribution, two unlabeled rooted trees $\ct$ and $\td{\ct}$ are sampled independently from distributions $\gw_{d,\lambda s}$ and $\gw_{d,\lambda\td{s}}$ respectively. 
    \item Correlated tree pair distribution $\mathbb{P}^{(\lambda,s,\td{s})}_d$:
    Under this distribution, two correlated trees $\ct$ and $\td{\ct}$ with marginal distributions $\gw_{d,\lambda s}$ and $\gw_{d,\lambda\td{s}}$ are generated.
    To generate the two correlated trees, an \emph{intersection tree} $\ct^*$ is first sampled from distribution $\gw_{d,\lambda s\td{s}}$. Then for each $k\in\{0,\ldots,d-1\}$, we independently sample a random variable $X^+_v\sim\poi(\lambda s(1-\td{s}))$ for each vertex $v$ at depth $k$ in $\ct^*$, and add $X_v$ new children to $v$. For each newly added vertex $u$ at depth $k$, a rooted tree $t_u$ is independently sampled from the distribution $\gw_{d-k,\lambda s}$ and attached to $u$ as a subtree. After attaching $t_u$ for each new vertex $u$, the resulting rooted tree is $\ct$. To obtain tree $\td{\ct}$, we again start with the intersection tree $\ct^*$, and independently repeat the process of adding new children and attaching new subtrees, but instead with $X_v^+\sim\poi(\lambda \td s(1-s))$ and $t_u\sim\gw_{d-k,\lambda \td{s}}$. It can be readily seen from Poisson superposition that $\ct$ and $\td{\ct}$ have marginal distributions $\gw_{d,\lambda s}$ and $\gw_{d,\lambda\td{s}}$ respectively.
\end{itemize}

\emph{\textbf{Likelihood ratio test.}}
In the tree correlation test problem, we are given two rooted trees $\ct$ and $\tdct$, and the goal is to distinguish the two hypotheses $H_0: (\ct,\tdct)\sim \mathbb{Q}^{(\lambda,s,\td{s})}_d$ and $H_1: (\ct,\tdct)\sim \mathbb{P}^{(\lambda,s,\td{s})}_d$.
The Neyman--Pearson Lemma~\citep{neyman1933ix} states that the \emph{likelihood ratio test} is the most powerful test for hypothesis testing~\citep{lehmann2005testing}. For the two trees $\ct$ and $\tdct$, their likelihood ratio is given by
\begin{equation}
\label{eq:LR}
    L_{d}^{(\lambda,s,\td{s})}(\ct,\tdct):=\frac{\mathbb{P}^{(\lambda,s,\td{s})}_d(\ct,\tdct)}{\mathbb{Q}^{(\lambda,s,\td{s})}_d(\ct,\tdct)}.
\end{equation}
The likelihood ratio for tree correlation tests has been studied in~\cite{maier2025}, who derive a recursive formula that enables efficient computation of~\eqref{eq:LR}. Moreover, they show that, with an appropriately chosen threshold, the likelihood ratio test achieves a \emph{vanishing Type~I error}, while having a \emph{non-vanishing power}. We state their recursive formula for the likelihood ratio in Appendix~\ref{appd:LR_recursive} and summarize the associated performance guarantees in Theorem~\ref{thm:tree_test} of Appendix~\ref{appd:LR_feasibility}.

\section{Proposed Algorithm and Main Results}
In this section, we present the proposed algorithm for diffusion-network alignment, and summarize the theoretical performance guarantees of the algorithm.

\subsection{Proposed algorithm for diffusion-network alignment}
In sparse \erdos--\renyi graphs, the local neighborhood distribution converges to a Poisson Galton–Watson tree as $n\rightarrow\infty$. In the IC model, each vertex activates its neighbors independently with probability $q$, so the subtree rooted at a vertex $u$ in $\tic$ can be regarded as a random subsample of its neighborhood in $\cg_1$. Consequently, this subtree retains the correlation with the local neighborhood of $\pi^*(u)$ in $\td{\cg}_2$.
Therefore, a natural approach to verify whether a vertex $u$ should be matched to a vertex $\td{v}$ is to apply the likelihood ratio test~\eqref{eq:LR} to determine whether their subtree and neighborhood are correlated or independent. However, such local tests alone are insufficient for certifying the matchings, because neighborhoods of non-corresponding nodes are not necessarily independent, as they may partially overlap. In our approach, we therefore use likelihood ratio tests only as local evidence, and propagate this evidence upward and downward along the diffusion tree. Aggregating evidence through the propagation process allows the algorithm to better identify correct matches.

We now turn to formally describe the proposed algorithm. To this end, we first introduce a notation for the subtrees in $\tic$ and define the \emph{breadth-first search (BFS) trees} in $\td{\cg}_2$. For a vertex $u$ in $\tic$ and $d\ge 0$, we use $\ct^\rmic_{u,d}$ to denote the subtree rooted at $u$ with depth up to $d$ in $\tic$. In other words, $\ct^\rmic_{u,d}$ is the rooted tree obtained by removing all the vertices that has distance greater than $d$ from $u$ in the subtree rooted at $u$ in $\tic$. For example, in Figure~\ref{fig:IC}, the subtree $\ct^{\mathrm{IC}}_{2,1}$ is the tree rooted at vertex $2$, with a single edge connecting to vertex $6$.
For two distinct vertices $\td{v}$ and $\td{w}$ in $\td{\cg}_2$, define the BFS tree $\td{\ct}^{\setminus\td{w}}_{\td{v},d}$ as follows: The tree $\td{\ct}^{\setminus\td{w}}_{\td{v},d}$ is rooted at vertex $\td{v}$. For each $k\in [d]$, the set of vertices at depth $k$ in $\td{\ct}^{\setminus\td{w}}_{\td{v},d}$ is the set of vertices at distance $k$ from $\td{v}$ in the graph $\td{\cg_2}\setminus\td{w}$. For each non-root vertex $\td{i}$ at depth $k$, the unique parent of $\td{i}$ is set to be the vertex at depth $k-1$ that is connected with $\td{i}$ in $\td{\cg_2}\setminus \td{w}$ that has the highest canonical order. Here, the canonical order within each depth is the same as in the definition of the diffusion tree $\tic$. 
The proposed algorithm collects local evidence for matching by applying tree correlation tests to pairs of subtree $\ct^\rmic_{u,d}$ and BFS tree $\td{\ct}^{\setminus\td{w}}_{\td{v},d}$.

The proposed algorithm takes the model parameters $\lambda,s,q$, the diffusion tree $\tic$, the observed graph $\td{\cg}_2$, and two integers $l$ and $\dmax$ as input. 
Recall that $V_{\le\dmax +1}$ denote the set of vertices within depth $\dmax+1$ of $\tic$. This is the set of vertices that the algorithm attempts to match. 
For each vertex $u\in V_{\le d_{\max}+1}$, the algorithm maintains a matching  candidate set $\mathcal M_u$, which is initialized to be empty. The candidate sets are updated sequentially as the four matching criteria stated below are applied. Whenever a pair $(u,\tilde v)$ is certified by one of the criteria below, the algorithm adds $\tilde v$ to $\mathcal M_u$. Thus, the criteria below are interpreted with respect to the current candidate sets at the time they are evaluated.
In stating these criteria, we omit the superscript from the likelihood ratio $L_d^{(\lambda/s,qs,s)}$ as is clear from the context.
\begin{figure}[thbp]
    \centering
    \includegraphics[width=1\linewidth]{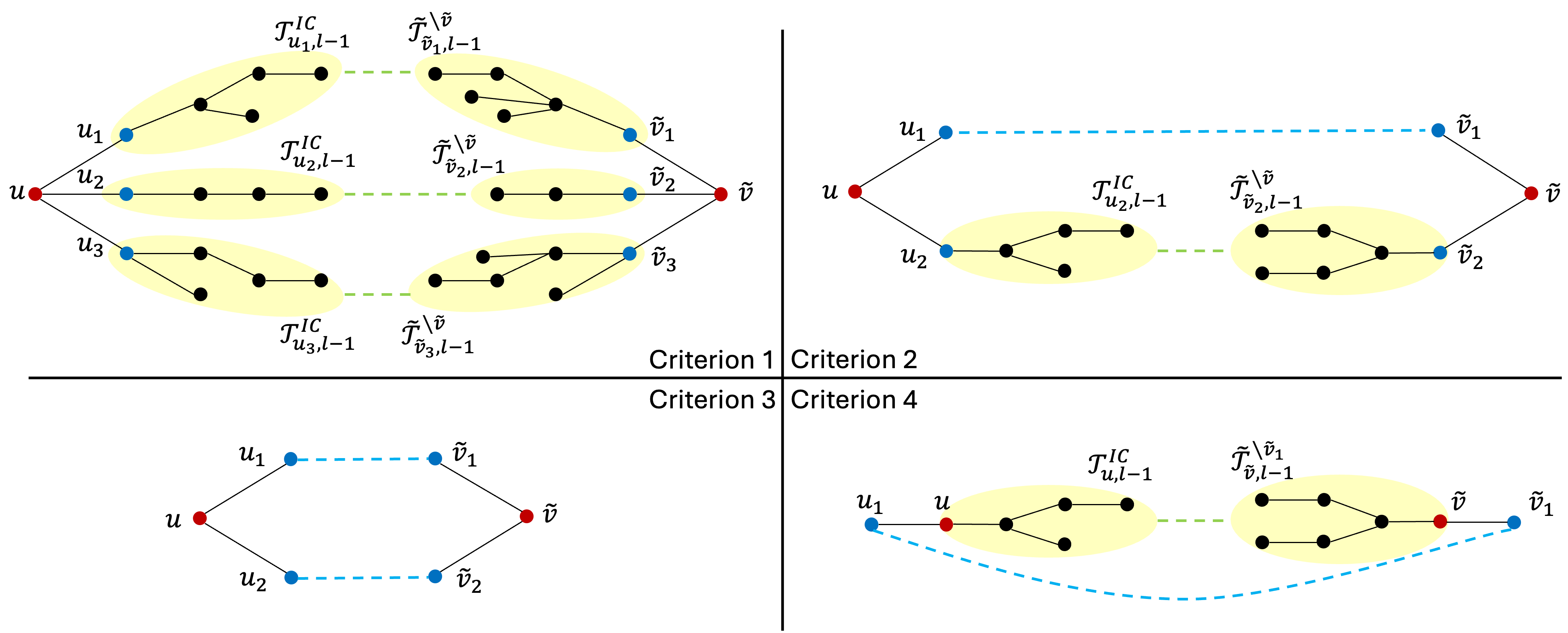} 
    \caption{Illustration of the four matching criteria.}
    \label{fig:4criteria}
\end{figure}

\begin{enumerate}
    \item Two vertices $u\in V^\rmic$ and $\td{v}\in \td{V}$ satisfy Criterion 1 if there exist three children $u_1,u_2,u_3$ of $u$ in $\tic$ and three neighbors $\td{v}_1,\td{v}_2,\td{v}_3$ of $\td{v}$ in $\td{\cg}_2$ such that\footnote{In the original definition of the likelihood ratio~\eqref{eq:LR}, the inputs are two \emph{unlabeled trees}. Here we extend the definition to allow labeled trees as inputs simply by removing all the labels.}
    \[
    L_{l-1}(\ct^\mathrm{IC}_{u_k,l-1},\td{\ct}^{\setminus\td{v}}_{\td{v}_k,l-1})>\exp\left(\frac{(\lambda s q)^{l-1}}{\log n}\right), \text{ for } k=1,2,3.
    \]
    \item Two vertices $u\in V^\rmic$ and $\td{v}\in \td{V}$ satisfy Criterion 2 if there exist two children $u_1,u_2$ of $u$ in $\tic$ and two neighbors $\td{v}_1,\td{v}_2$ of $\td{v}$ in $\td{\cg}_2$ such that 
    \[
    \td{v}_1\in \cm_{u_1}\text{ and } L_{l-1}(\ct^\mathrm{IC}_{u_2,l-1},\td{\ct}^{\setminus\td{v}}_{\td{v}_2,l-1})>\exp\left(\frac{(\lambda s q)^{l-1}}{\log n}\right).
    \]
    \item Two vertices $u\in V^\rmic$ and $\td{v}\in \td{V}$ satisfy Criterion 3 if there exist two children $u_1,u_2$ of $u$ in $\tic$ and two neighbors $\td{v}_1,\td{v}_2$ of $\td{v}$ in $\td{\cg}_2$ such that 
    \[
    \td{v}_k\in \cm_{u_k}, \text{ for } k=1,2.
    \]
    \item Two vertices $u\in V^\rmic$ and $\td{v}\in \td{V}$ satisfy Criterion 4 if there exists a neighbor $\td{v}_1$ of $\td{v}$ in $\td{\cg}_2$ such that $\td{v}_1\in\cm_{u_1}$, where $u_1$ is the parent of $u$ in $\tic$, and 
    \[
    L_{l-1}(\ct^\mathrm{IC}_{u,l-1},\td{\ct}^{\setminus\td{v}_1}_{\td{v},l-1})>\exp\left(\frac{(\lambda s q)^{l-1}}{\log n}\right).
    \]
\end{enumerate}
Figure~\ref{fig:4criteria} illustrates the four matching criteria introduced above. Each panel shows two local views. The left side is the view from the diffusion tree $\tic$, containing the vertex $u$, its children, and the corresponding subtrees. The right side is the view from the network $\td\cg_2$, containing the candidate vertex $\td v$, its neighboring vertices, and the BFS trees rooted at those neighbors. Green dashed lines connect pairs of rooted trees whose likelihood ratio exceeds the threshold $(\lambda s)^{l-1}/\log n$, while blue dashed lines indicate vertex pairs that have already been matched in previous iterations.

The proposed algorithm is summarized in Algorithm~\ref{alg:mpmatch}. The algorithm can be divided into two phases: an \emph{upward pass} phase and a \emph{downward pass} phase. 

\textbf{\underline{Upward pass:}} The upward pass phase starts from the vertices at depth $\dmax$ of $\tic$ and moves towards the root, with iterations over empty depths in $\tic$ skipped. Equivalently, if $\tic$ has depth less than $\dmax$, the first non-vacuous iteration starts from the leaf depth of $\tic$.
At the starting depth, the algorithm relies solely on Criterion 1 to match the vertices, since all matching candidate sets $\cm_u$ are initialized as empty. A vertex $\td{v}\in \td{V}$ is added to the matching candidate set $\cm_u$ of $u\in V_\dmax$ if $u$ and $\td{v}$ satisfy Criterion 1 (lines 2--6). We note that this criterion is also known as the three-dangling tree test in the literature~\citep{ganassali2024}. The algorithm then iteratively moves one level upward. At each subsequent depth $d$, the non-empty candidate sets found at depth $d+1$ enable the use of Criteria 2 and 3. A vertex $\td{v}\in \td{V}$ is added to the matching candidate set $\cm_u$ of $u\in V_d$ if $u$ and $\td{v}$ satisfy at least one of Criteria 1--3 (lines 2--6). This process continues until the root of the diffusion tree $\tic$ is reached. In this manner, the matching information obtained at lower levels and the local evidence from the tree correlation tests are progressively aggregated and propagated upward, boosting the matching performance at higher levels of $\tic$.

\textbf{\underline{Downward pass:}} The downward pass phase starts at depth $1$ of $\tic$ and moves towards the leaves of $\tic$. At each depth $d$, a vertex $\td{v}\in\td{V}$ is added to the matching candidate set $\cm_u$ of $u\in V_d$ if $u$ and $\td{v}$ satisfy Criterion 4 (lines 7--11). Recall that Criterion 4 is based on the matching information of the parental vertex in $\tic$. Therefore, this phase propagates the matching information from the upper levels of $\tic$ towards the lower levels. 

After the two passes, the algorithm outputs all the matching candidate sets $\cm_u$ (line 16).

In the analysis of the proposed algorithm, we choose parameters $$l=\lfloor\sqrt{\log n}\rfloor\,\text{ and } \,\dmax=\left\lfloor\frac{1-\epsilon}{\log(\lambda q)}\log n\right\rfloor,$$
where $\epsilon>0$ is an arbitrarily small constant.
\begin{remark}[Computational complexity of Algorithm~\ref{alg:mpmatch}]
    Under the correlated \erdos--\renyi graph pair model, the maximum degree in both graphs $\cg_1$ and $\tdcg_2$ is at most $\frac{(1+o(1))\log n}{\log \log n}$ with high probability~\citep{frieze2015introduction}. Since $(\frac{(1+o(1))\log n}{\log\log n})!=n^{1+o(1)}$, with high probability, the time complexity of each likelihood ratio computation in Algorithm~\ref{alg:mpmatch} is $n^{2+o(1)}$ (see Remark~\ref{rem:LR_complexity} in Appendix~\ref{appd:tree_test} for the detailed reasoning). 
    Algorithm~\ref{alg:mpmatch} performs at most $n^{2+o(1)}$ likelihood ratio tests, and hence runs in polynomial time.
\end{remark}

\begin{algorithm2e}[t]
\caption{Proposed algorithm for diffusion-network alignment} 
\label{alg:mpmatch}
\DontPrintSemicolon
\SetAlgoNoEnd
\SetKwInOut{Input}{Input}
\SetKwInOut{Output}{Output}
\Input{Diffusion tree $\tic$, fully observed graph $\td{\cg}_2$, model parameters $\lambda,s,q$, integers $d_{\max},l$}
\Output{A matching candidate set $\cm_v$ for each $v$ within depth $d_{\max}+1$ of $\tic$ }
 $\cm_u \gets \emptyset$ for each $u\in V_{\le \dmax+1}$\\
\tcp{Upward pass phase; empty depths are skipped}
\For{$d=d_{\max}:-1:0$}{
\For{$u\in V_{d}$}{
\For{$\td{v}\in \td{V}$}{
\If{$u$ and $\td{v}$ satisfy at least one of Criteria 1--3}{
 $\cm_u\gets \cm_u\cup\{\td{v}\}$
}
}
}
}
\tcp{Downward pass phase; empty depths are skipped}

\For{$d=1:+1:d_{\max}+1$}{
\For{$u\in V_{d}$}{
\For{$\td{v}\in \td{V}$}{
\If{$u$ and $\td{v}$ satisfy Criterion 4}{
 $\cm_u\gets \cm_u\cup\{\td{v}\}$
}
}
}
}

 \Return{The matching candidate set $\cm_u$ for each $u\in V_{\le\dmax+1}$}
\end{algorithm2e}

\subsection{Main results}
\label{sec:main_results}
In this section, we present the theoretical guarantee of Algorithm~\ref{alg:mpmatch}. To state the results, we first define a sequence of probabilities based on Poisson distributions. The sequence will later be used to bound the probability for the vertices at each depth of $\tic$ to be correctly matched by Algorithm~\ref{alg:mpmatch}.

Let $\eps'$ denote an arbitrarily small constant.
We recursively define the sequence $(p_d)_{d\ge 0}$ as follows: First define \begin{equation}
    p_0=\P(X\ge 3),
\end{equation}
where $X\sim\poi(\lambda sq(1-p^\mathrm{ext}_{\lambda sq}-\eps'))$. Then for each $d\ge 1$, define
\begin{equation}
p_{d}=\P(2Y+Z\ge 3),
\end{equation}
where $Y\sim\poi(\lambda sq p_{d-1})$, $Z\sim\poi(\lambda sq (1-p^\mathrm{ext}_{\lambda sq}-\eps'-p_{d-1}))$, and $Y$ and $Z$ are independent.

    To understand the definition of the sequence $(p_d)_{d\ge 0}$, consider running the upward pass phase of Algorithm~\ref{alg:mpmatch} on a pair of correlated trees $(\ct,\tdct)\sim \mathbb{P}_{d+l}^{(\lambda/s,qs,s)}$ starting at depth $d$.
    
    At the bottom depth $d$, the algorithm relies only on Criterion 1. For a vertex at depth $d$ in the intersection tree $\ct^*$, let $X$ denote the number of its children $c$ in $\ct^*$ such that $c$'s associated subtrees in $\ct$ and $\tdct$ can pass the tree correlation test.
    By definition, the number of children in $\ct^*$ is distributed as $\poi(\lambda sq)$. Moreover, we will see by Theorem~\ref{thm:tree_test} that, for each of such children, the associated subtrees in $\ct$ and $\tdct$ pass the test with probability at least $1-p^\mathrm{ext}_{\lambda sq}-\eps'$. By Poisson thinning, the distribution of $X$ stochastically dominates $\poi(\lambda sq(1-p^\mathrm{ext}_{\lambda sq}-\eps'))$, which is the distribution in the definition of $p_0$.
    A match at this level requires $X\ge 3$, so the probability of correct matching at the bottom layer is at least $p_0$.  
    
    When the algorithm moves to depth $d-1$, it applies Criteria 1--3 jointly. In this step, a pair of previously matched children one depth below can be viewed as a strong signal of value $2$, while a pair of subtrees that pass the correlation test but are not matched provides a weak signal of value $1$. Declaring a match whenever at least one of Criteria 1--3 is satisfied is then equivalent to requiring the total signal strength to be at least $3$. Let $Y$ denote the number of strong signals and $Z$ denote the number of weak signals.
    Using the lower bound $p_0$ on the probability of correct matching at the bottom level together with Theorem~\ref{thm:tree_test}, a similar Poisson thinning argument shows that one may take
    $Y\sim\poi(\lambda sqp_0)$ and $Z\sim\poi(\lambda sq(1-p^\mathrm{ext}_{\lambda sq}-\eps'-p_0))$ with $Y$ and $Z$ independent. These are precisely the random variables appearing in the definition of $p_1$, and $p_1$ thus characterizes the probability of correct matching one level above. This reasoning can be repeated inductively as the algorithm propagates upward: each move up by one level advances the recursion by one, so that the lower bound on the probability of a correct match increases from $p_j$ to $p_{j+1}$.
    In particular, $p_d$ is the probability that algorithm maps the roots of the two trees.
    This argument is formally presented as Lemma~\ref{prop:auxiliary} in Appendix~\ref{appd:pf_fraction}.

    It is straightforward to verify that the sequence $(p_d)_{d\ge 0}$ is strictly increasing and converges to a finite limit $p\in (0,1]$. Figure~\ref{fig:pd_evolution} illustrates an example of the evolution of the sequence.
\begin{figure}[htbp]
    \centering
    \includegraphics[width=0.5\linewidth]{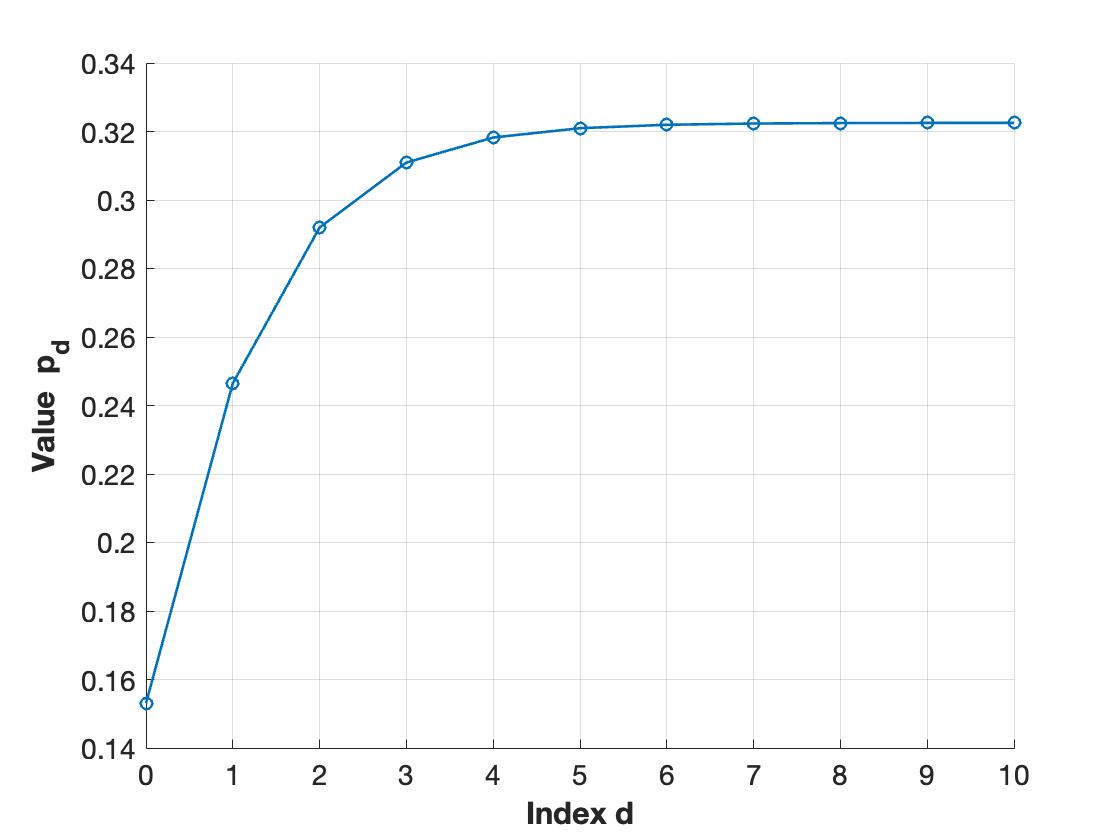}
    \caption{An example of the sequence $(p_d)_{d\ge 0}$. The parameters are set to $\lambda=4$, $s=q=0.7$ and $\eps'=0.1$ in the example.}
    \label{fig:pd_evolution}
\end{figure}

With the sequence $(p_d)_{d\ge 0}$ in hand, we are now ready to state the performance guarantee of the proposed algorithm.
Recall that $\cm_u$ is the set of matching candidates for vertex $u$ in $\tic$ that Algorithm~\ref{alg:mpmatch} outputs.
Accordingly, for an individual vertex $v$, the event $\cm_v=\{\pi^*(v)\}$ means that $v$ is correctly matched by Algorithm~\ref{alg:mpmatch}. Meanwhile, the event $\cm_u\subseteq\{\pi^*(u)\}$ for all $u\in V_{\le\dmax+1}$ captures the \emph{global correctness property}, meaning that, for every vertex, the algorithm outputs either an empty candidate set or the singleton set containing the true corresponding vertex.
The following theorem provides a guarantee for Algorithm~\ref{alg:mpmatch} in terms of the probability of these desired events.

\begin{restatable}{theorem}{guarantee}
    \label{thm:guarantee}
    Let $\alpha\approx0.338$ denote Otter's tree-counting constant.
    Suppose $qs^2>\alpha$, $\lambda sq>1$, and $\lambda>\lambda_0$, where $\lambda_0$ is a constant independent of $n$. Define $l^+=\lfloor(\log n)^{3/4}\rfloor$, and $d'=\min(l^+,\dmax-d)$ for each $d\in[\dmax]$. Then, for a vertex $v\in\{2,\ldots,n\}$ in $\cg_1$ and a depth $d\in[\dmax]$,
    \begin{equation}
    \label{eq:cond_frac}
        \P(\cm_v=\{\pi^*(v)\}\cond v\in V_{d})\ge p_{d'}-O(n^{-\eps/3}).
    \end{equation}
    For the root vertex $1$ of the diffusion tree,
    \begin{equation}
    \label{eq:root_frac}
        \P(\cm_1=\{\pi^*(1)\})\ge p_{l^+}-O(n^{-\eps/3}).
    \end{equation}
    Moreover, 
    \begin{equation}
        \label{eq:noerror_thm}
        \P(\cm_u\subseteq\{\pi^*(u)\},\forall u\in V_{\le\dmax+1})\ge 1-O(n^{-\eps/3}).
    \end{equation}
\end{restatable}

Recall that $V_d$ is the set of vertices at depth $d$ of $\tic$. For a non-root vertex $v\in\{2,\ldots,n\}$, both its appearance in the
diffusion tree and its depth are random. Hence, for each non-zero depth $d\in[\dmax]$, equation~\eqref{eq:cond_frac} provides a lower bound on the probability that an individual non-root vertex is correctly matched, conditional on the event that it appears at depth $d$ of $\tic$.
On the other hand, the root vertex $1$ is the deterministic source of the diffusion tree and always appears at depth $0$. Therefore, equation~\eqref{eq:root_frac} provides an unconditional lower bound on the probability that the root vertex is correctly matched.

By the definition of $d'=\min\{l^+,\dmax-d\}$, the lower bound depends on the depth $d$ through the index $d'$. Moreover, $d'$ is non-increasing as a function of $d$. As $d$ decreases from $\dmax$ to $\dmax-l^+$, the index $d'$ increases from $0$ to $l^+$, and the corresponding lower bound improves along the sequence $p_0,\ldots,p_{l^+}$. For vertices at smaller depths $d<\dmax-l^+$, this improvement saturates, and the bound remains $p_{l^+}$. Since $l^+\to\infty$ as $n\to\infty$, the value of $p_{l^+}$ approaches the sequence limit $p$. Therefore, for depths $d\le \dmax-l^+$, our lower bound is equivalently $p-o(1)$.

\begin{remark}
    In our model, the quantity $\sqrt{qs^2}$ captures the edge correlation coefficient between the diffusion tree $\tic$ and graph $\cg_2$. Therefore, condition $qs^2>\alpha$ in Theorem~\ref{thm:guarantee} is analogous to the assumption that the edge correlation coefficient lies above the threshold $\sqrt{\alpha}$ in the literature on efficient algorithms for sparse network alignment~\citep{mao2026,ganassali2024,ganassali2024statistical}.
\end{remark}
\begin{remark}
    The use of tree correlation tests in sparse network alignment has been studied in~\cite{maier2025}, where the central idea is to apply Criterion 1, also known as the three-dangling-tree test, to verify vertex matchings between two fully observed graphs. One can adapt this approach to the diffusion–network alignment setting and show that the root vertex $1$ can be correctly matched with probability at least $p_0-o(1)$. 
    In contrast, our algorithm leverages four complementary matching criteria and propagates matching evidence across the diffusion tree. This enables us to provide vertex-wise performance guarantees for non-root vertices and to obtain 
    a strictly improved lower bound $p_{l^+}-o(1)$ for the root.
\end{remark}

\begin{remark}
    When specialized to the case $q=1$, our model is closely related to the classical sparse network alignment problem, since the diffusion tree contains all vertices in the connected component of the root. It is therefore interesting to compare our guarantee with the information-theoretic optimum for partial recovery characterized in~\cite{du2025optimal}. For example, when $\lambda s=2$, the information-theoretic characterization in~\cite{du2025optimal} gives an optimal aligned fraction of approximately $0.7149$,\footnote{The optimal matching fraction in~\cite{du2025optimal} does not admit a closed-form expression. The value reported here is obtained by numerically evaluating the characterization provided in that work.} whereas our lower bound on the success probability of the root is approximately $0.4125$. Closing this gap, or determining whether the information-theoretic optimum can be achieved efficiently in this setting, remains a future research problem.
\end{remark}

\section{Proof Sketch of the Main Results}
In this section, we sketch the proof of Theorem~\ref{thm:guarantee}, while deferring the detailed proof to Appendix~\ref{appd:pf_main}. The proof relies on two key ingredients, each of which is formalized by a proposition.

The first ingredient is showing that the global correctness property holds with high probability. Importantly, this property remains valid even after conditioning on the event that a given vertex $v$ appears at a specific  depth $d$ of $\tic$. This statement is made precise in Proposition~\ref{prop:noerror}.

The second ingredient is a per-vertex inclusion guarantee. Conditioning on the event that vertex $v$ is at depth $d$ of $\tic$, 
we show that Algorithm~\ref{alg:mpmatch} includes $\pi^*(v)$ in the candidate seet $\cm_v$ with probability at least ${p_{d'}}-o(1)$, where $d'=\min\{l^+,\dmax-d\}$. This guarantee is formalized in Proposition~\ref{prop:inclusion}.

Together, these two ingredients yield the theorem as follows. The global correctness property ensures that, with high probability, each candidate set $\cm_u$ contains at most the true correspondence $\pi^*(u)$. On this event, the inclusion guarantee implies that $\cm_v=\{\pi^*(v)\}$, meaning that vertex $v$ is correctly matched. Applying this argument vertex-wise and combining the conditional and unconditional bounds gives the stated performance guarantees in Theorem~\ref{thm:guarantee}. 

To prove the two propositions, it is important to introduce the notion of the  \emph{live-edge graph}, a standard representation of the IC model~\citep{kempe2003}.
The live-edge graph $\cg_1'$ is obtained by subsampling each edge in $\cg_1$ independently with probability $q$, where the randomness in the subsampling is coupled with the IC process (see Appendix~\ref{appd:live-edge} for the formal definition).
The live-edge graph fully captures the randomness in both $\cg_1$ and the IC process. The proof of the two propositions is based on analyzing the graph pair $(\cg_1',\tdcg_2)$. We sketch the proof of the two propositions in the following two subsections, while deferring the detailed proofs to Appendices~\ref{appd:pf_correctness} and~\ref{appd:pf_fraction} respectively.

\subsection{Global correctness property}

\emph{\textbf{Sufficient condition for global correctness.}}
We first establish a sufficient condition for the desired global correctness of Algorithm~\ref{alg:mpmatch}. Let $\bar{\cg}$ denote the union graph of $\cg_1'$ and $\cg_2$.
We define the following two events.
\begin{itemize}
\setlength\itemsep{0em}
    \item $\ce_1$: for every $u$ within depth $\dmax$ of $\tic$, the $2l$-neighborhood of $u$ in $\bar{\cg}$ is a tree.
    \item $\ce_2$: among all pairs of  subtree $\ct^{\mathrm{IC}}_{u,l-1}$ and BFS tree $\tdct_{\td{v},l-1}^{\setminus\td{w}}$ tested in Algorithm~\ref{alg:mpmatch}, there is no pair that simultaneously passes the tree correlation test and contains no common vertices.
\end{itemize}

We claim that the event $\ce_1\cap\ce_2$ implies the global correctness property. 
The claim is proven by induction.
Consider any iteration of the algorithm and suppose that all previously identified matches are correct. If two non-corresponding vertices $u$ and $\td{v}$ were to satisfy any of the four matching criteria at this iteration, then event $\ce_2$ can be shown to imply the existence of multiple distinct paths between $u$ and $\td{v}$ in $\bar{\cg}$. This further implies that the $2l$-neighborhood of $u$ in $\bar{\cg}$ contains a cycle, contradicting $\ce_1$.
A formal statement of this implication is given in Lemma~\ref{lem:sufficient} in Appendix~\ref{appd:pf_correctness}. Given this sufficient condition, it is left to show that both $\ce_1$ and $\ce_2$ occur with high probability, \emph{both unconditionally and conditionally} on a given vertex $v$ at depth $d$ of $\tic$. 

\emph{\textbf{Event $\ce_1$.}} 
We begin by considering the unconditional probability of event $\ce_1$. In a sparse \erdos--\renyi graph, the probability that the $2l$-neighborhood of a fixed vertex contains no cycles is $1-O(n^{-1+o(1)})$, because $l=o(\log n)$. Moreover, by our choice of $\dmax$, the total number of vertices within depth $\dmax$ of $\tic$ is at most $\td{O}(n^{1-\eps})$ with high probability. A naive union bound would therefore suggest that $\ce_1$ holds with probability $1-o(1)$.

However, this argument cannot directly be applied, since the event that a vertex lies within depth $\dmax$ of $\tic$ already implies the existence of a path connecting it to the root, and thus enforces local structure in $\bar{\cg}$. To address this issue, we show that the $1-O(n^{-1+o(1)})$ bound on the probability of having no cycles still holds after conditioning on the existence of such a path. This allows us to apply a union-bound-type argument to establish the desired unconditional probability bound for $\ce_1$.

We then turn to the conditional bound. In addition to the conditioning considered above, namely, that a vertex lies within depth $\dmax$ of $\tic$, we now further condition on the event that a specific vertex $v$ lies at depth $d$ of $\tic$. This introduces a second, explicitly conditioned path from $v$ to the root. The main technical challenge is therefore to control the possible overlap patterns between these two conditioned paths. By carefully analyzing such overlaps, we show that the same $O(n^{-1+o(1)})$ bound continues to hold, yielding the desired conditional high-probability bound for $\ce_1$.

\emph{\textbf{Event $\ce_2$.}} 
We first consider the unconditional probability of event $\ce_2$.
From definition~\eqref{eq:LR}, the likelihood ratio $L_{l-1}^{(\lambda/s,qs,s)}$ has expectation $1$ under the independent tree pair distribution $\mathbb{Q}_{l-1}^{(\lambda/s,qs,s)}$. By our choice $l=\lfloor\sqrt{\log n}\rfloor$, the likelihood ratio threshold in Algorithm~\ref{alg:mpmatch} is $n^{\omega(1)}$. Hence, by Markov's inequality, $L_{l-1}^{(\lambda/s,qs,s)}$ exceeds the threshold with probability $n^{-\omega(1)}$. 
Moreover, for any pair of trees tested in Algorithm~\ref{alg:mpmatch} that are vertex-disjoint, we show that their joint law is within a constant factor from $\mathbb{Q}_{l-1}^{(\lambda/s,qs,s)}$, under high probability conditions on the tree sizes.
From here, we can apply a union bound to all possible tested tree pairs to show that $\ce_2$ holds with high probability.

The treatment for the conditional probability of $\ce_2$ is rather straightforward. We show that the conditioned event $v\in V_d$ occurs with probability $\Omega(n^{-1})$. The chain rule implies that the multiplicative gap between the conditioned and the unconditioned probability is $O(n)$, which yields the desired bound.

\subsection{Per-vertex inclusion guarantee}
For illustration purposes, we focus on showing the lower bound $p_{l^+}-o(1)$ at depth $d=\dmax-l^+$ in this proof sketch. Recall that we argued in Section~\ref{sec:main_results} that $p_{l^+}$ is a lower bound for the probability of matching the root node when applying the upward pass phase of Algorithm~\ref{alg:mpmatch} to a tree pair $(\ct,\td{\ct})\sim \mathbb{P}^{(\lambda/s,qs,s)}_{l+l^+}$. Therefore, it is desirable to couple the subtree rooted at $v$ up to depth $(l+l^+)$ in $\tic$ and the $(l+l^+)$-neighborhood of $\pi^*(v)$ in $\td{\cg}_2$ with the Galton--Watson tree pair $(\ct,\td{\ct})\sim \mathbb{P}^{(\lambda/s,qs,s)}_{l+l^+}$. 

As a key step, we show that when $S\in [n]$ is a vertex subset with sublinear size, the joint law between the $(l+l^+)$-neighborhood of $v$ in $\cg_1'\setminus S$ and the $(l+l^+)$-neighborhood of $\pi^*(v)$ in $\td{\cg}_2\setminus \pi^*(S)$ converges in total variation distance to $\mathbb{P}^{(\lambda/s,qs,s)}_{l+l^+}$ (see Proposition~\ref{prop:tree_tv} in Appendix~\ref{appd:pf_fraction} for the formal statement).
While the coupling is guided by the Poisson approximation for the neighborhood sizes in \erdos--\renyi graphs, the argument requires additional control beyond matching these marginals. In particular, under the correlated \erdos--\renyi graph pair model, the explored $(l+l^+)$-neighborhoods are not a priori trees: one must rule out edges within the same layer, and multiple edges from a vertex to previously discovered layers. Furthermore, joint exploration across the two graphs introduces additional overlap patterns that do not arise in the Galton–Watson model. For example, a vertex may be reached in the neighborhood of $v$ via an exclusive edge in $\cg_1'$ and simultaneously be reached in the neighborhood of $\pi^*(v)$ via an exclusive edge in $\td{\cg}_2$, creating cross-neighborhood intersections that are not tree-like. A key part of the proof is therefore to show that such non-tree effects and undesired overlaps occur with vanishing probability, so that the neighborhood growth can be coupled to the correlated Galton–Watson process.

With this convergence in hand, we apply it with $S=V_{\le d}\setminus\{v\}$, namely the set of vertices within depth $d$ of $\tic$ excluding $v$. Under this choice of $S$, we show that the $(l+l^+)$-neighborhood of $v$ in $\cg_1'\setminus S$ coincides with the subtree of $v$ in $\tic$, and that the corresponding  $(l+l^+)$-neighborhood of $\pi^*(v)$ in $\tdcg_2\setminus \pi^*(S)$ does not intersect with the neighborhoods of vertices in $\pi^*(S)$. Consequently, the local neighborhoods used by the algorithm around $v$ and $\pi^*(v)$ behave as if they were isolated and distributed according to the correlated Galton–Watson model. This allows the upward pass phase of Algorithm~\ref{alg:mpmatch} to include $\pi^*(v)$ in $\cm_v$ with probability at least $p_{l^+}-o(1)$.

As illustrated in this proof sketch, our
probability lower bounds are established solely from the upward pass. The downward pass, on the other hand, can match additional vertices that cannot be matched by the upward pass alone. Figure~\ref{fig:4criteria} illustrates such an example. Note that Criteria~1--3 require the vertex $u$ to have at least two children in the diffusion tree in order to be matched. By contrast, the example for Criterion~4 shows that a vertex $u$ with only one child can still be matched using information propagated from its parent $u_1$. However, the improvement brought by the downward pass is difficult to quantify because of the complicated correlation between its decisions and those made during the upward pass. We therefore include the downward pass as a practical refinement that can strictly improve the matching performance beyond what is captured by our theoretical lower bounds. It is an interesting future research problem to theoretically quantify the gain of the downward pass.

\section{Conclusion}
In this work, we introduce the diffusion–network alignment problem, motivated by the information asymmetry in real-world networks. We formalize this problem under a probabilistic model that combines the correlated \erdos--\renyi graph pair model with an Independent Cascade diffusion process, capturing both structural correlation across networks and the randomness introduced by diffusion.

We propose a polynomial-time algorithm for diffusion–network alignment, and analyze its performance in the sparse regime.  We provide performance guarantees  for the algorithm from two perspectives. First, we established depth-dependent lower bounds on the probability of correctly matching each individual vertex in the diffusion tree. Second, we prove a global correctness guarantee, showing that with high probability all the vertex pairs matched by the algorithm are correct.

Several directions for future work on diffusion–network alignment remain open. One natural question is to characterize the optimal vertex-wise matching probabilities that can be achieved by polynomial-time algorithms under the global correctness requirement. Beyond efficient algorithms, it is also of interest to understand the information-theoretic limits of diffusion–network alignment.

\section*{Acknowledgments}
This work was supported in part by U.S. National Science Foundation (NSF) under grants 2207548 and 2324769; and AFOSR grant FA9550-24-1-0002. The authors are grateful to Jiale Cheng for many helpful and inspiring discussions. The authors also thank the anonymous reviewers for their constructive comments and suggestions.

\bibliographystyle{IEEEtranN}
\bibliography{bibliography}

\appendices

\section{Related Work}
\label{appd:related}
Diffusion-network alignment is closely related to the classic network alignment problem.
The theoretical study of the network alignment problem has focused mainly on the correlated \erdos--\renyi graph pair model. One line of the work studies the information theoretic limit. The goal is to identify the range of graph parameters such that successful alignment can be achieved under various criteria including exact recovery~\citep{Cul-Kiy-improved2016,Cul-Kiy-exact2017,settling-TIT}, almost exact recovery~\citep{Cullina2019,settling-TIT} and partial recovery~\citep{ganassali21a,settling-TIT,hall2023partial,ding2023matching,du2025optimal}.
The other line of work focuses on constructing efficient solutions for the problem. Several polynomial-time approaches have been proposed, including spectral methods~\citep{fan-20a}, node-degree–based methods~\citep{Dai-Cullina-Kiyavash-Grossglauser2019,mao2023exact}, iterative algorithms~\citep{ding2025polynomial}, subgraph counting techniques~\citep{mao2026}, and tree correlation tests~\citep{ganassali2024,maier2025}.

Following the line of work on the correlated \erdos--\renyi graph pair model, many works have studied network alignment under more general settings, including alignment in community-structured random graph models~\citep{gaudio2022exact,chai2024efficient}, graphs with node attributes~\citep{zhang2024TIT,wang2023feasible,wang2025efficient,yang2024exact}, partially correlated graph pairs~\citep{huang2024information}, and the simultaneous alignment of multiple networks~\citep{ameen2024aligning}. Despite these generalizations, a common assumption in this literature is that the input graphs are fully observed. In contrast, in the diffusion–network alignment problem we consider, only a subgraph induced by the diffusion process is observed from the first network, while the second network is fully observed.

Beyond classic network alignment, the diffusion-network alignment problem is also related to the literature on identifying subgraphs in large networks. One classical line of research is the subgraph isomorphism problem, which asks whether a given query graph appears as a subgraph of a larger host graph~\citep{shang2008taming,He2008}. From a probabilistic perspective,~\cite{shiu2025information} study this problem under random graph models, where the objective is to recover the vertex labels of the subgraph induced by a uniformly sampled vertex set from an \erdos--\renyi graph.
Another closely related direction is the planted clique problem, where a clique of a given size is planted on an \erdos--\renyi random graph and the goal is to detect or recover the planted structure~\citep{bollobas1976cliques,alon1998finding}. Recent studies go beyond cliques and consider planted subgraphs of more general structures~\citep{mossel2025sharp,lee2025}. 

Despite the connections, these settings differ fundamentally from diffusion–network alignment in two respects. Firstly, although the diffusion tree is a subgraph of the underlying network $\cg_1$, the reference network $\tdcg_2$ we observe is only positively correlated with $\cg_1$ through a latent correspondence, rather than containing the diffusion tree as a subgraph. In contrast, both subgraph isomorphism and planted subgraph problems require the target subgraph structure to be strictly contained in the observed network. Second, the diffusion tree in diffusion-network alignment is  generated by a random diffusion process on the latent network through the IC model. This is different from the setting where the queried or planted structure is deterministic or sampled as an induced subgraph.

\section{Notation Used in the Appendices}
The following notations are used throughout the appendices. 

We use standard notation $\mathbb{N},\mathbb{N}^+,\mathbb{R}, \mathbb{R}^+$ to denote the sets of natural numbers, positive natural numbers, real numbers, and positive real numbers, respectively. For $n\in \mathbb{N}^+$, we use $[n]$ to denote the set $\{1,\ldots,n\}$. When $n=0$, we set $[n]=\emptyset$. For a finite set $S$, we use $|S|$ to denote its cardinality. 

For a simple undirected graph $\cg=(V_\cg,E_\cg)$, we use $V_\cg$ and $E_\cg$ to denote the set of vertices and the set of edge in $\cg$ respectively. For $u,v\in V_\cg$, we use $u\stackrel{\cg}{\sim}v$ to denote that there exists an edge between $u$ and $v$ in $\cg$, and use $\dist_\cg(u,v)$ to denote the shortest path distance between $u$ and $v$.
For a vertex set $S$ and a vertex $u$, we define the distance $\dist(S,u)=\min_{v\in S}\dist(u,v)$.
For $u\in V_\cg$, we define $S_\cg(u,d):=\{v\in V_\cg:\dist_\cg(u,v)=d\}$ to denote the set of vertices at distance $d$ from $u$, and $N_\cg(u,d):=\cup_{k=0}^dS_\cg(u,d)$ to denote the set of vertices within distance $d$ from $u$. 

A rooted tree $\ct=(V_\ct,E_\ct)$ is an acyclic simple graph with one distinguished vertex $r_\ct\in V_\ct$ designated as the \emph{root}. The \emph{depth} of a vertex is defined as the distance between it and the root. The depth of $\ct$ is defined as the maximum depth among all of its vertices. For $d\in\mathbb{N}$, we introduce the short-hand notation  $V_{\ct,d}:=S_\ct(u,d)$ for the set of vertices at depth $d$ of $\ct$, and $V_{\ct,\le d}:=\cn_\ct(u,d)$ for the set of vertices within depth $d$ of $\ct$. When $\ct$ is the diffusion tree $\tic$ from the IC process, we omit $\tic$ from the subscript, and simply write $V_d$ and $V_{\le d}$. For $u\in V_{\ct,d}$, the children of $u$ are the vertices in $V_{\ct,d+1}$ that are adjacent to $u$. We say $u$ is the parent of $v$ if $v$ is a child of $u$. For $u\in V_\ct$ and $d\in\mathbb{N}$, we use $\ct_u$ to denote the subtree of $\ct$ rooted at $u$ and $\ct_{u,d}$ to denote the rooted tree obtained by removing all the vertices in $\ct_u$ that are at distance larger than $d$ from the root $u$. We say two rooted trees $\ct$ and $\ct'$ are isomorphic (denoted $\ct\cong\ct'$), if there exists a bijective mapping $\sigma:V_\ct\rightarrow V_{\ct'}$ such that $\sigma(r_\ct)=r_{\ct'}$ and $u\stackrel{\ct}{\sim}v  \Leftrightarrow \sigma(u)\stackrel{\ct'}{\sim}\sigma(v)$.

For a given graph $\cg$ and a vertex $u\in V_\cg$, we define $\ct_{u}^{\cg}\in$ as the \emph{breadth-first search (BFS) tree} rooted at $u$: The root vertex $r_{\ct_{u}^{\cg}}$ of the tree is set to be vertex $u$. For each $d\in\mathbb{N}^+$, we have $V_{\ct_{u}^{\cg},d}=S_\cg(u,d)$, i.e., the set of vertices at depth $d$ in $\ct_{u}^{\cg}$ are the vertices at distance $d$ from $u$ in $\cg$. For a vertex $v\in V_{\ct_{u}^{\cg},d}$, the parent vertex of $v$ is the vertex in $ V_{\ct_{u}^{\cg},d-1}$ that is adjacent to $v$ in $\cg$ with the highest order. Here, the ordering in each $V_{\ct_{u}^{\cg},d}$ is defined in the same way as in the definition of the diffusion tree $\tic$.

We use notations $\P(\cdot)$ and $\E[\cdot]$ to denote the probability of certain event and the expectation of certain random variable. For a sequence of random variables $(X_1,\ldots,X_n)$, we use $\cl(X_1,\ldots,X_n)$ to denote their joint distribution. We use $\poi(\mu)$ to denote the Poisson distribution with mean $\mu$, and use $\bin(n,p)$ to denote the binomial distribution with parameters $n$ and $p$. We say a random vector $(X_1,\ldots,X_m)$ follows the multinomial distribution with parameters $(n,p_1,\ldots,p_m)$, denoted $(X_1,\ldots,X_m)\sim\mathrm{Multi}(n,p_1,\ldots,p_m)$ if $\P(X_1=x_1,\ldots,X_m=x_m)=\frac{n!}{x_1!\cdots x_m!}\prod_{i=1}^mp_i^{x_i}$ for all non-negative integers $x_1,\ldots,x_m$ satisfying $\sum_{i=1}^mx_i=n$. For two distributions $\mu$ and $\nu$, their total variation distance is defined as $\dtv(\mu,\nu)=\inf\P(X\neq Y)$, where the infimum is taken over all couplings $(X,Y)$ with marginals $\mu$ and $\nu$.

We follow the standard order notation: $f(n) = O(g(n))$ if $\lim_{n \to \infty} \frac{|f(n)|}{g(n)} < \infty$; $f(n) = \Omega(g(n))$ if $\lim_{n \to \infty}\frac{f(n)}{g(n)} >0$; $f(n) = \Theta(g(n))$ if $f(n) = O(g(n))$ and $f(n) = \Omega(g(n))$; $f(n) = o(g(n))$ if $\lim_{n \to \infty}\frac{f(n)}{g(n)} = 0$; $f(n) = \omega(g(n))$ if $\lim_{n \to \infty}\frac{|f(n)|}{|g(n)|} = \infty$. Additionally, we use $\td{O}(\cdot)$ notation to indicate that logarithmic factors are omitted.

\section{Proof of the Main Theorem}
\label{appd:pf_main}
In this section, we prove Theorem~\ref{thm:guarantee}, which is restated in the following.
\guarantee*

Throughout the proof of this theorem, we assume without loss of generality that the latent permutation $\pi^*$ sampled in the correlated \erdos--\renyi graph pair model is the identity permutation. Under this assumption, the two graphs $\cg_2$ and $\tdcg_2$ are identical, and for each $u\in [n]$, vertex $u$ in $\cg_1$ corresponds to vertex $\td{u}$ in $\tdcg_2$.

With these assumption, Theorem~\ref{thm:guarantee} follows readily from the following two propositions.

\begin{restatable}{proposition}{noerror}
    \label{prop:noerror}
    Suppose $\lambda sq>1$. Then 
    \begin{equation}
        \label{eq:noerror}
        \P(\cm_u\subseteq\{\td{u}\},\forall u\in V_{\le\dmax+1})\ge 1-O(n^{-\eps/3}).
    \end{equation}
    Moreover, for any vertex $v\in \{2,\ldots,n\}$ and depth $d\in [\dmax]$,
    \begin{equation}
        \label{eq:cond_noerror}
        \P(\cm_u\subseteq\{\td{u}\},\forall u\in V_{\le\dmax+1}\cond v\in V_d)\ge 1-O(n^{-\eps/3}).
    \end{equation}
\end{restatable}
Recall that $\cm_u\subseteq\{\td{u}\},\forall u\in V_{\le\dmax+1}$ is the desired global correctness property of Algorithm~\ref{alg:mpmatch}, and equation~\eqref{eq:noerror} states that this property holds with high probability.
Additionally, equation~\eqref{eq:cond_noerror} states this property still holds with high probability after conditioning on the event that a specific vertex $v$ appears at certain depth $d$ in $\tic$.
\begin{restatable}{proposition}{fraction}
\label{prop:inclusion}
    Suppose the assumptions in Theorem~\ref{thm:guarantee} hold. Then for each vertex $v\in\{2,\ldots,n\}$ and depth $d\in [\dmax]$, 
    \begin{equation}
        \label{eq:inclusion}
        \P(\td{v}\in\cm_v\cond v\in V_{d})\ge p_{d'}-O(n^{-\eps/3}).
    \end{equation}
    For the root vertex $1$ of $\tic$, 
    \begin{equation}
        \label{eq:inclusion_root_prop}
        \P(\td{1}\in\cm_1)\ge p_{l^+}-O(n^{-\eps/3}).
    \end{equation}
\end{restatable}

\begin{proof}[Proof of Theorem~\ref{thm:guarantee}]
Since $\pi^*$ is the identity permutation, we have $\pi^*(u)=\td{u}$ for each $u\in [n]$. 
Therefore, statement~\eqref{eq:noerror_thm} in Theorem~\ref{thm:guarantee} is exactly statement~\eqref{eq:noerror} in Proposition~\ref{prop:noerror}. Moreover, the event $\cm_u\subseteq\{\td{u}\},\forall u\in V_{\le\dmax+1}$ trivially implies that $\cm_1\subseteq\{\pi^*(1)\}$. The two events $\cm_1\subseteq\{\pi^*(1)\}$ and $\pi^*(1)\in \cm_1$ together imply $\cm_1=\{\td{1}\}$. Therefore, by the union bound, statements~\eqref{eq:noerror} and~\eqref{eq:inclusion_root_prop} together imply statement~\eqref{eq:root_frac} in Theorem~\ref{thm:guarantee}. Finally, for any vertex $v\in \{2,\ldots,n\}$, we have $\cm_v\subseteq\{\pi^*(v)\}$ on the event $\cm_u\subseteq\{\td{u}\},\forall u\in V_{\le\dmax+1}$. Therefore, statements~\eqref{eq:cond_noerror} and~\eqref{eq:inclusion} together imply statement~\eqref{eq:cond_frac} in Theorem~\ref{thm:guarantee}.
\end{proof}

To prove Propositions~\ref{prop:noerror} and~\ref{prop:inclusion}, we first introduce the notion of the live-edge graph for the IC model in Appendix~\ref{appd:live-edge}. The detailed proofs of the two propositions are then presented in Appendices~\ref{appd:pf_correctness} and~\ref{appd:pf_fraction} respectively.

\section{Live-edge Graph}
\label{appd:live-edge}
In this section, we define the \emph{live-edge graph} $\cg_1'$ corresponding to the IC model we introduced in Section~\ref{sec:model}. 
graph $\cg_1'$ has the same vertex set as $\cg_1$, and its edges are generated by keeping each edge in $\cg_1$ independently with probability $q$. Furthermore, the randomness in edge subsampling is coupled with the randomness in the IC process: In each time step of the IC process, an edge between an active vertex and an inactive is selected to propagate the message if and only if the edge presents in $\cg_1'$. Recall the correlated \erdos--\renyi graph pair model defined in Section~\ref{sec:model}.
From the subsampling process of generating $\cg_1'$ from $\cg_1$, we know that
for each pair of indices $(i,j)\in [n]\times[n]$, the pair of edge indicators $(\cg_1'(i,j) ,\td{\cg}_2(\td{i},\td{j}))$ has distribution
\begin{equation}
    \label{eq:asym_ER_pair}
    (\cg_1'(i,j),\tdcg_2(\td{i},\td{j}))=\begin{cases}
        (1,1)\text{\;\;\;w.p.\;\;\;}\frac{\lambda sq}{n}\\
        (0,1)\text{\;\;\;w.p.\;\;\;}\frac{\lambda(1-sq)}{n}\\
        (1,0)\text{\;\;\;w.p.\;\;\;}\frac{\lambda q(1-s)}{n}\\
        (0,0)\text{\;\;\;w.p.\;\;\;}1-\frac{\lambda+\lambda q-\lambda sq}{n},
    \end{cases}
\end{equation}
and these edge indicator pairs are mutually independent among different index pairs. Let $\bar{\cg}$ denote the union graph of $\cg_1'$ and $\tdcg_2$. It follows by~\eqref{eq:asym_ER_pair} that $\bar{\cg}$ has distribution $\mathrm{ER}(n,\bar{\lambda}/n)$, where $\bar{\lambda}:=\lambda+\lambda q-\lambda sq$. The following lemma states the relationship between the diffusion tree $\tic$ we observe from the model and the live-edge graph $\cg_1'$.
\begin{lemma}
    \label{lem:TIC}
    The diffusion tree $\tic$ is exactly the BFS tree $\ct_{1}^{\cg_1'}$ rooted at vertex $1$ in $\cg_1'$.
\end{lemma}
\begin{proof}[Proof of Lemma~\ref{lem:TIC}]
    At each time step of the IC model, an inactive vertex get activated if and only if it has a neighbor in $\cg_1'$ that is activated in the previous time step. This proves that the vertices that are activated at the $d$-th time step are the vertices at distance $d$ from $1$ in $\cg_1'$. This shows that for each $d\ge 0$, the set of vertices at depth $d$ of $\tic$ is the same as the set of vertices at depth $d$ of $\ct_{1}^{\cg_1'}$. Moreover, the canonical order of vertices at each depth in the two trees are defined in the same way. Therefore, $\tic$ and $\ct_{1}^{\cg_1'}$ choose the same unique parent for each vertex in the tree. This shows that the two tree are equal.
\end{proof}

Lemma~\ref{lem:TIC} shows that the live-edge graph $\cg_1'$ captures the randomness in both $\cg_1$ and the IC model.
Consequently, throughout the proof of Propositions~\ref{prop:noerror} and~\ref{prop:inclusion}, we work exclusively with the graph pair $(\cg_1',\tdcg_2)$, and $\cg_1$ no longer appears explicitly in the analysis. 
We continue to use tilde sign to differentiate the notation related to $\tdcg_2$, and further introduce the bar notation to differentiate the notation related to the union graph $\bar{\cg}$. In particular, for each $i\in [n]$, we use $i$ (resp. $\td{i}$ and $\bar{i}$) to denote the vertex with index $i$ in $\cg_1'$ (resp. $\tdcg_2$ and $\bar{\cg}$).
Since we assume $\pi^*$ is the identity permutation, $i,\td{i}$ and $\bar{i}$ are three corresponding vertices in the three graphs $\cg_1',\tdcg_2$ and $\bar{\cg}$. With a slight abuse of notation, for sets $S\subseteq\{1,\ldots,n\}$ and $\td{S}\subseteq\{\td{1},\ldots,\td{n}\}$, we define $$S\cap\td{S}:=\{i\in [n]:i\in S\text{ and }\td{i}\in\td{S}\}$$ and $$S\cup\td{S}:=\{i\in [n]:i\in S\text{ or }\td{i}\in\td{S}\}.$$ 
We follow the same convention when one of the subsets is replaced by $\bar{S}\subseteq\{\bar{1},\ldots,\bar{n}\}$.

\section{Proof of Proposition~\ref{prop:noerror}}\label{appd:pf_correctness}
In this section, we prove Proposition~\ref{prop:noerror}, which asserts that Algorithm~\ref{alg:mpmatch} outputs no false matches with high probability, both unconditionally and conditionally on a given vertex $u$ lying at depth $d$ of the diffusion tree $\tic$.

This proposition is proven in three steps. We first formulate a sufficient condition under which Algorithm~\ref{alg:mpmatch} outputs no false matches, by defining two good events based on the two graphs $\cg_1'$ and $\tdcg_2$. In the second and third steps, we show that these two good events occur with high probability, both unconditionally and conditionally on a given vertex $u$ at depth $d$ of the diffusion tree $\tic$.

\subsection{A sufficient condition for Algorithm~\ref{alg:mpmatch} to output no false matches}
To establish the sufficient condition, we first define two good events. Define
\begin{equation}
    \label{eq:non_exists_cycle}
    \ce_1:=\{\text{the $2l$-neighborhood of $\bar{i}$ in $\bar{\cg}$  is a tree}, \forall i\in V_{\le\dmax}\}.
\end{equation}
In other words, $\ce_1$ is the event that for every vertex $i$ within depth $\dmax$ of $\tic$, the $2l$-neighborhood of the corresponding vertex $\bar{i}$ in $\bar{\cg}$ is a tree.

Recall that Algorithm~\ref{alg:mpmatch} uses subtrees of $\tic$, denoted $\ct^\rmic_{i,d}$, and BFS trees in $\tdcg_2$, denoted $\tdct^{\setminus \td{j}}_{\td{k},d}$, in tree correlation tests. 
We define $\ce_2$ as the event that there exist no triplets $(i,j,k)\in [n]^3$ with $j\neq k$ that satisfy all of the following properties:
\begin{enumerate}[itemsep=6pt]
    \item $i \in V_{\le\dmax+1}$;
    \item 
    $L_{l-1}(\mathcal{T}^\mathrm{IC}_{i,l-1},\td{\ct}^{\setminus\td{j}}_{\td{k},l-1})> \exp\left(\frac{(\lambda s q)^{l-1}}{\log n}\right)$;
    \item $V_{\mathcal{T}^\mathrm{IC}_{i,l-1}}\cap V_{\td{\ct}^{\setminus\td{j}}_{\td{k},l-1}}=\emp$.
    \end{enumerate}
Here $V_{\mathcal{T}^\mathrm{IC}_{i,l-1}}$ (resp. $V_{\td{\ct}^{\setminus\td{j}}_{\td{k},l-1}}$) is the set of vertices in $\mathcal{T}^\mathrm{IC}_{i,l-1}$ (resp. $\td{\ct}^{\setminus\td{j}}_{\td{k},l-1}$).
In the definition of $\ce_2$, if we further require that vertices $\td{j}$ and $\td{k}$ are adjacent in $\tdcg_2$, the event then essentially means that no pairs of trees tested in Algorithm~\ref{alg:mpmatch} can pass the likelihood ratio test, while having no common vertices between them. We deliberately omit this adjacency requirement. Doing so only makes the event $\ce_2$ more restrictive, and hence still sufficient for establishing the desired correctness property. At the same time, this relaxation simplifies the probabilistic analysis, as it allows us to bound the probability of $\ce_2$ without conditioning on the presence of specific edges in $\tdcg_2$.
With the two events defined, we have the following lemma that establishes the sufficient condition.

\begin{lemma}
\label{lem:sufficient}
    Event $\ce_1\cap\ce_2$ implies that for all $u\in V_{\tic,\le\dmax+1}$, the output matching candidate set $\cm_u$ satisfies
    \begin{equation}
        \label{eq:correctness}
        \cm_u\subseteq\{\td{u}\}.
    \end{equation}
\end{lemma}
\begin{proof}[Proof of Lemma~\ref{lem:sufficient}]
     To prove this lemma, we essentially need to show that $\ce_1\cap\ce_2$ implies the following four claims.
     \begin{itemize}[itemsep=6pt]
        \item Claim 1: For a pair of vertices $i\in V_{\le\dmax}$ and $\td{j}\in V_{\tdcg_2}$, Criterion 1 is satisfied only if $i=j$.
        \item Claim 2: Suppose the set $\cm_u$ satisfies~\eqref{eq:correctness} for all $u\in V_{\le\dmax+1}$. Then for a pair of vertices $i\in V_{\le\dmax}$ and $\td{j}\in V_{\td{\cg}_2}$, Criterion 2 is satisfied only if $i=j$. 
        \item Claim 3: Suppose the set $\cm_u$ satisfies~\eqref{eq:correctness} for all $u\in V_{\le\dmax+1}$. Then for a pair of vertices $i\in V_{\le\dmax}$ and $\td{j}\in V_{\td{\cg}_2}$, Criterion 3 is satisfied only if $i=j$.
        \item Claim 4: Suppose the set $\cm_u$ satisfies~\eqref{eq:correctness} for all $u\in V_{\le\dmax+1}$. Then for a pair of vertices $i\in V_{\le\dmax+1}$ and $\td{j}\in V_{\td{\cg}_2}$, Criterion 4 is satisfied only if $i=j$.
    \end{itemize}
    
Notice that at the start of Algorithm~\ref{alg:mpmatch}, all sets $\cm_u$ are initialized to be empty, so~\eqref{eq:correctness} holds for all these matching candidate sets.
The four claims then imply that in any iteration of Algorithm~\ref{alg:mpmatch}, a vertex $\td{j}$ is added to set $\cm_i$ only if $i=j$. Therefore, the desired property~\eqref{eq:correctness} holds for all $\cm_u$ throughout the process of Algorithm~\ref{alg:mpmatch}. We are now left to prove the four claims.

 \textbf{Proof of Claim 1.\footnote{We comment that the proof of this claim partially appeared in the proof of Theorem~7 in~\cite{ganassali2024} in the context of network alignment. Our proof adapts this argument to the diffusion–network alignment setting and accounts for the partial, tree-structured observation induced by the diffusion process.}}
    We prove this claim by contradiction. Suppose $i\neq j$, and the vertex pair $i$ and $\td{j}$ satisfies the Criterion 1. By the definition of the criterion, we know that there exists three children $i_1,i_2,i_3$ of $i$ in $\tic$, and three neighbors $\td{j}_1,\td{j}_2,\td{j}_3$ of $\td{j}$ in $\td{\cg}_2$ such that 
    \[
    L_{l-1}(\ct^\mathrm{IC}_{i_k,l-1},\td{\ct}^{\setminus\td{j}}_{\td{j}_k,l-1})>\exp\left(\frac{(\lambda s q)^{l-1}}{\log n}\right), \forall j\in [3].
    \]
    From event $\ce_2$, we know that 
    \[
    V_{\mathcal{T}^\mathrm{IC}_{i_k,l-1}}\cap V_{\td{\ct}^{\setminus\td{j}}_{\td{j}_k,l-1}}\neq\emptyset,\forall k\in[3].
    \]
    We will show that this implies that there is a cycle in the $2l$-neighborhood of $\bar{i}$ in $\bar{\cg}$, which contradicts the property that the $2l$-neighborhood of $\bar{i}$ in $\bar{\cg}$ is a tree required by $\ce_1$.

    First, consider the two trees $\mathcal{T}^\mathrm{IC}_{i_1,l-1}$ and $\td{\ct}^{\setminus\td{j}}_{\td{j}_1,l-1}$.
    Let $x\in [n]$ be such that $x\in V_{\mathcal{T}^\mathrm{IC}_{i_1,l-1}}$ and $\td{x}\in V_{\td{\ct}^{\setminus\td{j}}_{\td{j}_1,l-1}}$. Since $x\in V_{\mathcal{T}^\mathrm{IC}_{i_1,l-1}}$, we know that there exists a path $(i,y_1,y_2,\ldots,y_t)$ of length $t\in [l]$ with $y_1=i_1$ and $y_t=x$ in graph $\cg_1'$. Similarly, because $\td{x}\in V_{\td{\ct}^{\setminus\td{j}}_{\td{j}_1,l-1}}$, there exists a path $(\td{j},\td{z}_1,\td{z}_2,\ldots,\td{z}_{\td{t}})$ of length $\td{t}\in [l]$ with $z_1=j_1$ and $z_{\td{t}}=x$ in graph $\td{\cg}_2$.  From this we are going to argue that there exists a path $(\bar{p}_0,\bar{p}_1,\ldots,\bar{p}_{\bar{t}})$ of length $\bar{t}\in [2l]$ in $\bar{\cg}$ that satisfies the following two properties
    \begin{enumerate}
        \item $p_0=i$ \emph{and} $p_{\bar{t}}=j$;
        \item $p_1=i_1$ \emph{or} $p_{\bar{t}-1}=j_1$.
    \end{enumerate}
    To show this, we consider three different cases.

    \underline{\emph{Case (i):}} $j=y_\tau$ for some $\tau\in [t]$ (top figure in Figure~\ref{fig:cases}). Because the vertices on the path $(i,y_1,y_2,\ldots,y_t)$ are all distinct, we have the desired path $(\bar{p}_0,\bar{p}_1,\ldots,\bar{p}_{\bar{t}})$ of length at most $l$
    by setting $p_0=i, p_1=y_1=i_1,p_2=y_2,\ldots,p_{\bar{t}-1}=y_{\tau-1},p_{\bar{t}}=y_\tau=j$.

    \underline{\emph{Case (ii):}} $i=z_\tau$ for some $\tau\in [\td{t}]$ (middle figure in Figure~\ref{fig:cases}). Because the vertices of the path $(\td{j},\td{z}_1,\td{z}_2,\ldots,\td{z}_{\td{t}})$ are all distinct, we have the desired path $(\bar{p}_0,\bar{p}_1,\ldots,\bar{p}_{\bar{t}})$ of length at most $l$ by setting $p_{0}=z_{\tau}=i,p_{1}=z_{\tau-1},p_{2}=z_{\tau-2},\ldots,p_{\bar{t}-1}=z_1=j_1,p_{\bar{t}}=j$.

    \underline{\emph{Case (iii):}} $j\neq y_\tau$ for all $\tau\in [t]$ and $i\neq z_\tau$ for all $\tau\in [\td{t}]$ (bottom figure in Figure~\ref{fig:cases}). By concatenating the two paths, we know that 
    \[
    (\bar{i},\bar{y}_1,\bar{y}_2,\ldots,\bar{y}_{t-1},\bar{y}_t=\bar{z}_{\td{t}}=\bar{x},\bar{z}_{\td{t}-1},\ldots,\bar{z}_{2},\bar{z}_{1},\bar{j})
    \]
    is a walk of length at most $2l$ between $\bar{i}$ and $\bar{j}$ in $\bar{\cg}$. We can then erase the cycles in this walk to obtain a path $(\bar{p}_0,\bar{p}_1,\ldots,\bar{p}_{\bar{t}})$ of length at most $2l$ with $p_0=i$ and $p_{\bar{t}}=j$. Moreover, because both vertices $\bar{i}$ and $\bar{j}$ only appear once in this walk, the resulting path must satisfy both $p_1=i_1$ and $p_{\bar{t}-1}=j_1$.

    \begin{figure}[htbp]
    \centering
    \includegraphics[width=0.5\linewidth]{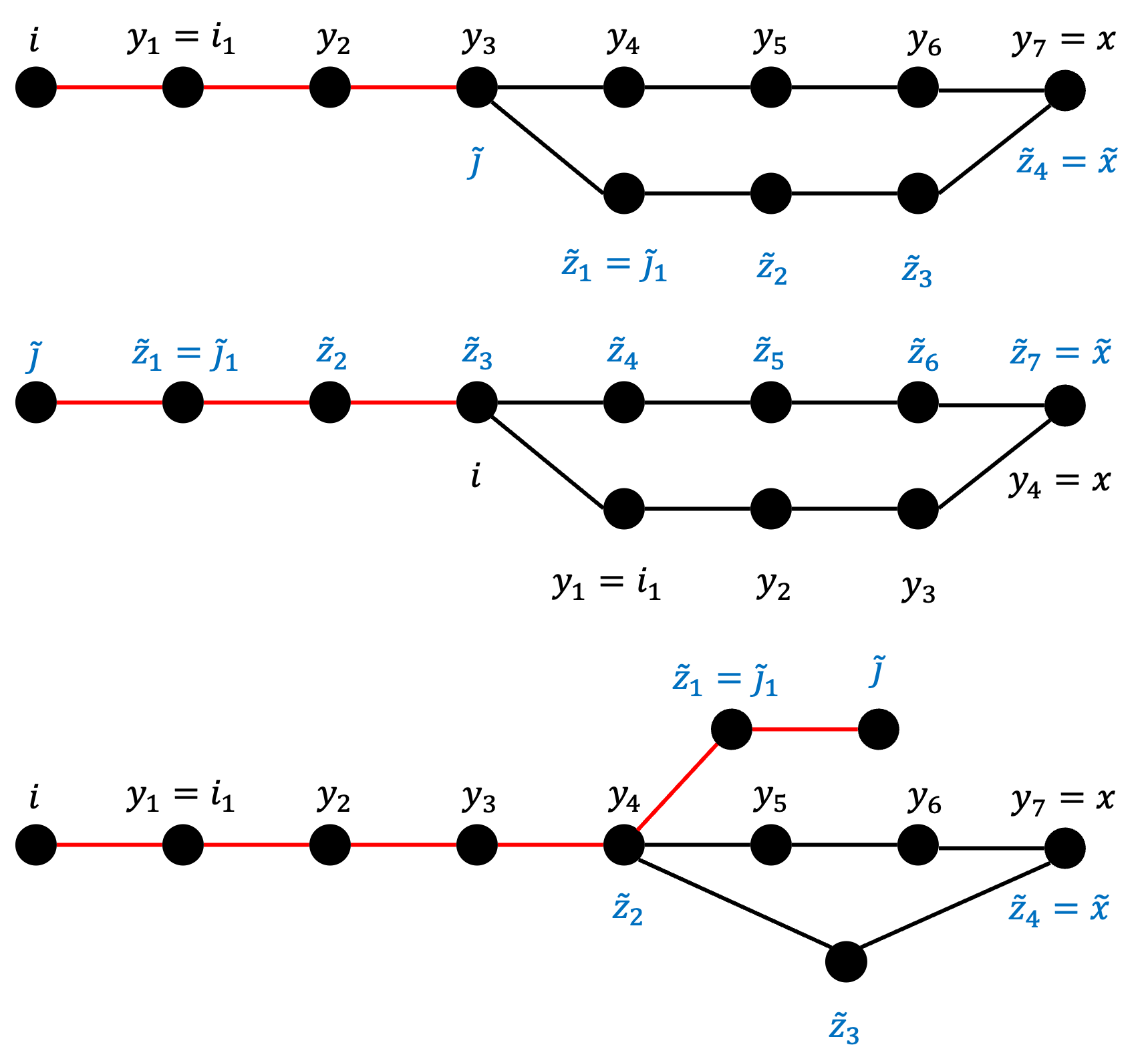}
    \caption{Intersection patterns of the two paths corresponding to cases (i), (ii) and (iii). The desired path $(\bar{p}_0,\bar{p}_1,\ldots,\bar{p}_{\bar{t}})$ is highlighted by the edges in red in the three figures respectively.}
    \label{fig:cases}
\end{figure}

    By the same argument, it follows by   $V_{\mathcal{T}^\mathrm{IC}_{i_2,l-1}}\cap V_{\td{\ct}^{\setminus\td{j}}_{\td{j}_2,l-1}}\neq\emptyset$ that there exists a path $(\bar{q}_0,\bar{q}_1,\ldots,\bar{q}_{\bar{t'}})$ of length $\bar{t'}\in [2l]$ that satisfies both properties
    \begin{enumerate}
        \item $q_0=i$ \emph{and} $q_{\bar{t'}}=j$;
        \item $q_1=i_2$ \emph{or} $q_{\bar{t'}-1}=j_2$.
    \end{enumerate}
    Because $V_{\mathcal{T}^\mathrm{IC}_{i_3,l-1}}\cap V_{\td{\ct}^{\setminus\td{j}}_{\td{j}_3,l-1}}\neq\emptyset$, there exists a path $(\bar{r}_0,\bar{r}_1,\ldots,\bar{r}_{\bar{t''}})$  of length $\bar{t''}\in [2l]$ that satisfies both properties
    \begin{enumerate}
        \item $r_0=i$ \emph{and} $r_{\bar{t''}}=j$;
        \item $r_1=i_3$ \emph{or} $r_{\bar{t''}-1}=j_3$.
    \end{enumerate}

    However, because $i_1\neq i_2\neq i_3$ and $j_1\neq j_2\neq j_3$, at least two of these three paths in $\bar{\cg}$ are non-identical paths from $\bar{i}$ to $\bar{j}$ in $\bar{\cg}$. This implies that there exists a cycle in the $2l$-neighborhood of $\bar{i}$ in $\bar{\cg}$, which contradicts to event $\ce_1$.
    
    \textbf{Proof of Claim 2.} Assume that~\eqref{eq:correctness} holds for all $v\in V_{\le \dmax+1}$. Suppose $i\neq j$, and $i$ and $\td{j}$ satisfy Criterion 2. By definition, we know that $i$ has two children $i_1,i_2$ in $\tic$ and $\td{j}$ has two neighbors $\td{j}_1,\td{j}_2$ in $\td{\cg}_2$ such that $\td{j}_1\in \cm_{i_1}$ and $$L_{l-1}(\ct^\mathrm{IC}_{i_2,l-1},\td{\ct}^{\setminus\td{j}}_{\td{j}_2,l-1})>\exp\left(\frac{(\lambda s q)^{l-1}}{\log n}\right).$$
    From the argument in the proof of Claim 1, we know that there exists a path $(\bar{p}_0,\bar{p}_1,\ldots,\bar{p}_{\bar{t}})$ of length at most $2l$ in $\bar{\cg}$ that satisfies
    \begin{enumerate}
        \item $p_0=i$ and $p_{\bar{t}}=j$;
        \item $p_1=i_2$ or $p_{\bar{t}-1}=j_2$;
    \end{enumerate}

    Moreover, by~\eqref{eq:correctness}, we know that $j_1=i_1$, and there is another path $(\bar{i},\bar{i}_1,\bar{j})$ between $\bar{i}$ and $\bar{j}$ in $\bar{\cg}$. The two path we found are non-identical, and this implies that there exists a cycle in the $2l$-neighborhood of $\bar{i}$ in $\bar{\cg}$, and hence leads to a contradiction with $\ce_1$.

     \textbf{Proof of Claim 3.} Assume that~\eqref{eq:correctness} holds for all $v\in V_{\dmax+1}$. Suppose $i\neq j$, and $i$ and $\td{j}$ satisfy Criterion 3. By definition, we know that $i$ has two children $i_1,i_2$ in $\tic$ and $\td{j}$ has two neighbors $\td{j}_1,\td{j}_2$ in $\td{\cg}_2$ such that $\td{k}\in\cm_{i_k},\forall k\in[2]$. Property~\eqref{eq:correctness} immediately implies that there exists a cycle $(\bar{i},\bar{i}_1,\bar{j},\bar{i}_2,\bar{i})$ in $\bar{\cg}$, which contradicts $\ce_1$.

      \textbf{Proof of Claim 4.} Assume that~\eqref{eq:correctness} holds for all $v\in V_{\le\dmax+1}$. Suppose $i\neq j$, and $i$ and $\td{j}$ satisfy Criterion 4. By definition, we know that there exists a neighbor $\td{j}_1$ of $\td{j}$ in $\td{\cg}_2$ such that $\td{j}_1\in\cm_{i_1}$, where $i_1$ is the parent of $i$ in $\tic$, and 
    \[
    L_{l-1}(\ct^\mathrm{IC}_{i,l-1},\td{\ct}^{\setminus\td{j}_1}_{\td{j},l-1})>\exp\left(\frac{(\lambda s q)^{l-1}}{\log n}\right).
    \]
    Notice that~\eqref{eq:correctness} implies $j_1=i_1$.
    By event $\ce_2$, we know that there exists $x\in[n]$ such that $x\in V_{\mathcal{T}^\mathrm{IC}_{i,l-1}}$ and $\td{x}\in V_{\td{\ct}^{\setminus\td{j}_1}_{\td{j},l-1}}$. Therefore, there exists a path $(j_1,q_1,\ldots,q_t,x)$ in $\cg_1$ of length at most $l$ that satisfies $q_1=i$, and a path $(\td{j_1},\td{r}_1,\ldots,\td{r}_{\td{t}},\td{x})$ in $\td{\cg}_2$ of length at most $l$ that satisfies $\td{q}_1=j$. Since $i\neq j$, there exist two non-identical paths between $\bar{j}_1$ and $\bar{x}$ in the union graph $\bar{\cg}$. This contradicts the event $\ce_1$.
\end{proof}
With Lemma~\ref{lem:sufficient} in hand, Proposition~\ref{prop:noerror} easily follows from the following two lemmas. 
\begin{lemma}
    \label{prop:E1}
    Suppose $\lambda sq>1$. We have 
    \begin{equation}
    \label{eq:e1_uncond}
        \P(\ce_1^c)=O(n^{-\epsilon/3}).
    \end{equation}
    Moreover, for any vertex $v\in \{2,\ldots,d\}$ and depth $d\in [\dmax]$,
    \begin{equation}
    \label{eq:e1_cond}
        \P(\ce_1^c\cond v\in V_d)=O(n^{-\epsilon/3}).
    \end{equation}
\end{lemma}

\begin{lemma}
    \label{prop:E2}
    Suppose $\lambda sq>1$. We have 
    \begin{equation}
    \label{eq:e2_uncond}
        \P(\ce_2^c)=\td{O}(n^{-4/3}).
    \end{equation}
    Moreover, for any vertex $v\in \{2,\ldots,d\}$ and depth $d\in [\dmax]$,
    \begin{equation}
    \label{eq:e2_cond}
        \P(\ce_2^c\cond v\in V_d)=\td{O}(n^{-1/3}).
    \end{equation}
\end{lemma}
We prove these two propositions in the next two subsections respectively.

\subsection{Proof of Lemma~\ref{prop:E1}}
We first establish the conditional probability bound in~\eqref{eq:e1_cond}. The unconditional probability bound in~\eqref{eq:e1_uncond} then follows directly from intermediate steps of this proof.

\subsubsection{Proof of the conditional probability bound~\eqref{eq:e1_cond}}

Recall that we defined 
\[
\ce_1=\{\nexists i\in V_{\le \dmax}:\text{the $2l$-neighborhood of $\bar{i}$ in $\bar{\cg}$  contains a cycle}\}.
\]
By Lemma~\ref{lem:TIC}, we know that $V_{\le d}=N_{\cg_1'}(1,d)$ for each $d\ge 0$. So we can equivalently write $\ce_1$ as 
\[
\ce_1=\{\nexists i\in N_{\cg_1'}(1,\dmax):\text{the $2l$-neighborhood of $\bar{i}$ in $\bar{\cg}$  contains a cycle}\}.
\]
We want to show that
\[
\P(\ce_1^c\cond v\in V_d)=O(n^{-\eps/3}).
\]

To prove this probability bound, 
we first define a sequence of auxiliary events that bound the neighborhood sizes of vertex $1$ in $\cg_1'$. For each $k\in\{0,\ldots,\dmax\}$, define
\[
\ca_k=\{|N_{\cg_1'}(1,k)|\le K(\lambda q)^k\log n\}
\]
where
$K:=\frac{\lambda q}{\lambda q-1}\prod_{h=0}^\infty(1+(\lambda q)^{-h/2})(1+(\lambda q)^{-h}).$ By the assumption that $\lambda sq>1$, we have
\begin{align*}
    &\prod_{h=0}^\infty(1+(\lambda q)^{-h/2})(1+(\lambda q)^{-h})\le \exp\left(\sum_{h=0}^{\infty}(\lambda q)^{-h/2}+(\lambda q)^{-h}\right)<\infty.
\end{align*}
This shows that $K$ is a constant.  It follows by Lemma~\ref{lem:neighborhood_size_cond} in Appendix~\ref{appd:properties} that these neighborhood size bounds hold with high probability:
    \[
    \P(\cap_{k=0}^{\dmax}\ca_k\cond u\in V_d)\ge 1-O(n^{-\eps/3}).
    \]
We also decompose event $\ce_1$ as $\ce_1=\cap_{k=0}^\dmax\ce_{1,k}$, where 
    \[
    \ce_{1,k}=\{\nexists i\in [n]:i\in S_{\cg_1'}(1,k)\text{ and }\text{the $2l$-neighborhood of $\bar{i}$ in $\bar{\cg}$  contains a cycle}\}.
    \]

We then have
\begin{align}
        \P(\ce_1^c\cond v\in V_d)&=\P(\cup_{k=0}^\dmax\ce_{1,k}^c\cond v\in V_d)\nonumber\\
        &\le \P((\cup_{k=0}^\dmax\ce_{1,k}^c)\cap (\cap_{k=0}^{\dmax}\ca_k)\cond v\in V_d)+\P((\cap_{k=0}^{\dmax}\ca_k)^c\cond v\in V_d)\nonumber\\
        &\le \P(\cup_{k=0}^{\dmax}(\ce_{1,k}^c\cap(\cap_{h=0}^{\dmax}\ca_h))\cond v\in V_d)+O(n^{-1/3})\nonumber\\
        &\le O(n^{-1/3})+\sum_{k=0}^{\dmax}\P(\ce_{1,k}^c\cap(\cap_{h=0}^{\dmax}\ca_h)\cond v\in V_d)\nonumber\\
        &\le O(n^{-1/3})+\sum_{k=0}^{\dmax}\P(\ce_{1,k}^c\cap\ca_k\cond v\in V_d)\nonumber\\
        &\le O(n^{-1/3})+\sum_{k=0}^{\dmax}\P(\ce_{1,k}^c\cond\ca_k,v\in V_d),\label{eq:e1_chain}
    \end{align}
and we focus on bounding each term in the summation.
In particular, we show that 
\begin{equation}
    \label{eq:e1_k}
    \P(\ce_{1,k}^c\cond\ca_k,v\in V_d)\le K(\log n)(\lambda q)^k\cdot O( n^{-1+\eps/4}).
\end{equation}
for each $k\in \{0,\ldots,\dmax\}$. Substituting~\eqref{eq:e1_k} into~\eqref{eq:e1_chain} yields
\begin{align*}
    \P(\ce_1^c\cond v\in V_d)&\le O(n^{-1/3})+\sum_{k=0}^{\dmax}K(\log n)
    (\lambda q)^k\cdot O( n^{-1+\eps/4})\\
    &\le O(n^{-1/3})+O( n^{-1+\eps/4})\sum_{k=0}^{\dmax}K(\log n)(\lambda q)^k\\
    &=O(n^{-\eps/3}),
\end{align*}
where the last equality follows because $\sum_{k=0}^{\dmax}K(\log n)(\lambda q)^k=O((\lambda q)^{\dmax}\log n)$ and the fact that $\dmax=\lfloor\frac{(1-\eps)\log n}{\log(\lambda q)}\rfloor$. Therefore, it suffices to prove~\eqref{eq:e1_k}. We separately consider two cases: $k\neq 0$ and $k=0$.

\noindent\underline{\textbf{Case 1. $k\neq 0$:}} For each vertex $i\in[n]$, define events
    \[
     \cb_{i,k}=\{i\in S_{\cg_1'}(1,k)\},
     \]
     and 
     \[
     \cc_i=\{\text{the $2l$-neighborhood of $\bar{i}$ in $\bar{\cg}$  contains a cycle}\}.
     \]
 By the union bound, we have
    \begin{align}
         \P(\ce_{1,k}\cond\ca_k,v\in V_{d})&\le \sum_{i=2}^{n}\P(\cb_{i,k}\cap\cc_i|\ca_k,v\in V_{d}).\label{eq:bikci}
    \end{align}

To bound this summation, we further consider two subcases $i\neq v$ and $i= v$.

\noindent\underline{\textbf{Case 1.1. $i\neq v$:}} We can apply the chain rule to get
    \begin{align}
        &\P(\cb_{i,k}\cap\cc_i|\ca_k,v\in V_d)\nonumber\\
        &\le \P(\cb_{i,k}|\ca_k,v\in V_d)\cdot\P(\cc_i\cond \ca_k,v\in V_d,i\in V_k)\nonumber\\
        &\leq \frac{K(\lambda q)^k\log n}{n-2}\cdot\P(\cc_i\cond \ca_k,v\in V_d,i\in V_k),\label{eq:ineqv}
    \end{align}
     where the last inequality follows by symmetry and the fact that $$|S_{\cg_1'}(1,k)|\le |N_{\cg_1'}(1,k)|\le K(\lambda q )^k\log n$$ under event $\ca_k$. To proceed, we express the two conditioned events $i\in V_k$ and $v\in V_d$ as path-existence events in the underlying graph. Conditioning on these path events allows us to bound the probability that the neighborhood of $\bar{i}$ in $\bar{\cg}$ being a tree more directly. 
     
     Notice that the event $i\in V_k$ is equivalent as the event that there exists a path of length $k$ between $1$ and $i$ in $\cg_1'$, while there exists no path of length less than $k$ between them. The event $v\in V_d$ can be rewritten in a similar way. 
     We define $\mathcal{P}^d_{1,v}$ as the collection of all possible paths of length $d$ from $1$ to $v$ in $\cg_1'$. The total order of the elements in $\mathcal{P}^d_{1,v}$ is defined as follows: For two distinct paths $p=(v_0,v_1,\ldots,v_{d})$ and $p'=(v_0',v_1',\ldots,v_{d}')$ with $v_0=v_0'=1$ and $v_{d}=v_{d}'=v$. We say $p$ has a higher order than $p'$, denoted $p\prec p'$, if there exists an index $d'\in[d-1]$ such that $v_{d'}<v_{d'}'$ and  $v_t=v_t',\forall t< d'$.
     We also define $\mathcal{P}^k_{1,i}$ as the collection of all possible paths of length $k$ from $1$ to $i$ in $\cg_1'$. The total order of the elements in $\mathcal{P}^k_{1,i}$ is analogously defined. With the definitions of these two collections and their corresponding total ordering, we have the following lemma.

     \begin{lemma}
         \label{lem:path_overlapp}
         Suppose $v\in V_d$ and $i\in V_k$.
         Assume $p\in \mathcal{P}^d_{1,v}$ is the highest order $d$-path between $1$ and $v$ in $\cg_1'$ and $p'\in \mathcal{P}^k_{1,i}$ is the highest order $k$-path between $1$ and $i$ in $\cg_1'$. Let $\ch_{p,p'}$ denote the union graph of the two path $p$ and $p'$. Then $\ch_{p,p'}$ must satisfy the following two properties:
         \begin{enumerate}
             \item $\ch_{p,p'}$ is a tree
             \item For each $h\ge 0$, the number of vertices in $\ch_{p,p'}$ that are at distance $h$ from $i$ is at most $2$.
         \end{enumerate}
     \end{lemma}
     \begin{proof}[Proof of Lemma~\ref{lem:path_overlapp}]
         To prove the lemma, we consider two separate case: $k=d$ and $k\neq d$.

         \underline{\textit{Case $k=d$:}} We denote $p=(v_0,v_1,\ldots,v_d)$ and $p'=(v_0',v_1',\ldots,v_d')$. Here $v_0=v_0'=1$, $v_d=v$ and $v_d'=i$. We claim the following property that $p$ and $p'$ must satisfy: Let $d^*=\arg\min_{t\in[d] }\{v_{t}\neq v_{t}'\}$. Then we have 
    \begin{equation}
    \label{eq:structure_h}
    \{v_{d^*},v_{d^*+1}\ldots,v_d\}\cap\{v_{d^*}',v_{d^*+1}'\ldots,v_d'\}=\emp.
    \end{equation}
    To see the claim, suppose that we have some $t>d^*$ such that $v_t=v_t'$. Let us assume without generality that the index of $v_{d^*}$ is larger than the index of $v_{d^*}'$. Then we know that $$p=(v_0,\ldots,v_{d^*-1},v_{d^*},v_{d^*+1},\ldots,v_{t-1},v_t,\ldots,v_d)$$ and 
    \[
    p^*=(v_0,\ldots,v_{d^*-1},v_{d^*}',v_{d^*+1}',\ldots,v_{t-1}',v_t,\ldots,v_d)
    \]
    are two distinct path between $1$ and $v$ in $\cg_1$. However, since the index of $v_{d^*}$ is larger than the index of $v_{d^*}'$, we know that $p^*$ has a higher order than $p$, contradicting the fact that $p$ is the highest order path between $1$ and $v$. 
    
    Now suppose that we have two indices $t,t'\ge d^*$ with $t> t'$ such that $v_{t}=v_{t'}'$. Then we know that there exists a path 
    \[
    (v_0,\ldots,v_{d^*-1},v_{d^*}',\ldots,v_{t'}',v_{t+1},v_{t+2},\ldots,v_d)
    \]
    between $1$ and $v$ in $\cg_1'$ with length less than $d$, and this contradicts with the fact that $v\in V_d$. When we have two indices $t,t'\ge d^*$ with $t< t'$, we can use a similar argument to conclude that there exists a path of length less than $d$ between $1$ and $i$, which contradicts with $i\in V_d$. This completes the proof of the claim. 

    The claim implies that $\ch_{p,p'}$ is a tree consists of three disjoint paths, that are from vertex $v_{d^*-1}$ to vertices $1$, $u$ and $i$ respectively. It is possible that $d^*=1$, in which case the path from $v_{d^*-1}$ to $1$ is a trivial path with only a single vertex and no edges.
    In all these scenarios, there are at most $2$ vertices in $\ch_{p,p'}$ that are at distance $h$ from $i$ for each $h\ge 0$.
    Therefore, the two properties in Lemma~\ref{lem:path_overlapp} all hold.
    In Figure~\ref{fig:overlap}, we provide examples of the two possible overlapping patterns of the two paths and the corresponding union graph $\ch_{p,p'}$.
    
    \begin{figure}[htbp]
    \centering
    \includegraphics[width=0.5\linewidth]{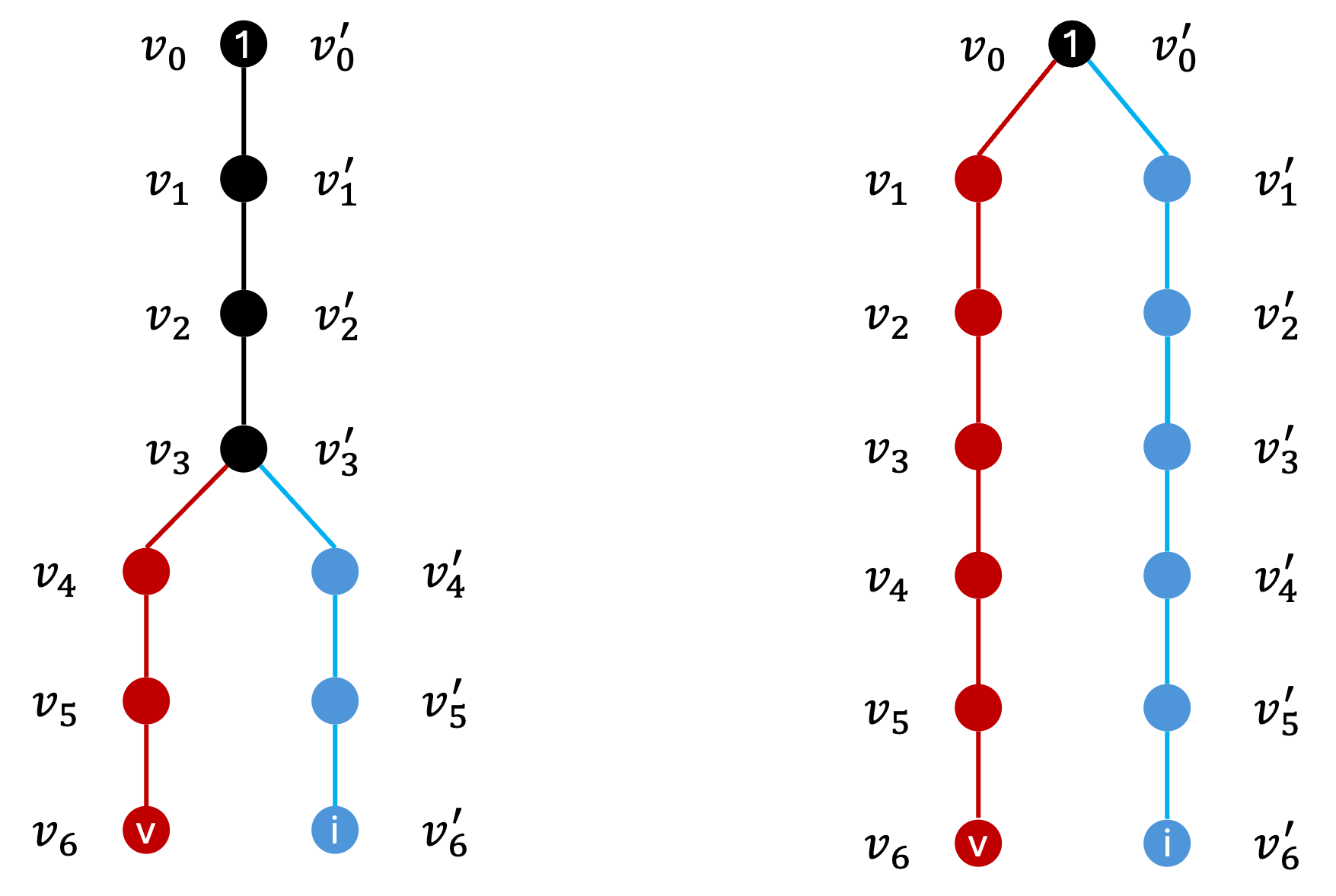} 
    \caption{Two examples of the overlap between $p$ and $p'$. In both examples, we take $k=d=6$. On the left-hand side, we have $d^*=4$. The three black vertices are shared by $p$ and $p'$, and the red (resp. blue) vertices are exclusive to $p$ (resp. $p'$). The union graph $\ch_{p,p'}$ is consist of three disjoint paths from vertex $v_{3}$ to vertices $1$, $v$ and $i$ respectively. On the right-hand side, we have $d^*=1$, and the two paths only share vertex $1$. In this scenario, $\ch_{p.p'}$ reduces to a path from $v$ to $i$.}
    \label{fig:overlap}
\end{figure}

    \underline{\textit{Case $k\neq d$:}} In this proof, we assume that $k<d$. The proof in the case of $k>d$ follows by an analogous argument. We denote $p=(v_0,v_1,\ldots,v_d)$ and $p'=(v_0',v_1',\ldots,v_k')$. First, assume that $v_t=v_t'$ for each $t\le k$, i.e., $p'$ is a subgraph of $p$. Then $\ch_{p,p'}$ is a path from $1$ to $u$, and it clearly satisfies the two properties in the lemma. Now suppose that the previous assumption does not hold, and let $d^*\le k$ be the smallest index such that $v_{d^*}\neq v_{d^*}'$. Then by the same argument as in the case of $k=d$, we know that $\{v_{d^*},v_{d^*+1}\ldots,v_d\}\cap\{v_{d^*}',v_{d^*+1}'\ldots,v_k'\}=\emp$. This again implies that $\ch_{p,p'}$ is a tree consists of three disjoint paths, that are from vertex $v_{d^*-1}$ to vertices $1$, $u$ and $i$ respectively. Therefore, the two properties in the lemma are satisfied.
     \end{proof}
    For two paths $p\in \mathcal{P}^d_{1,v}$ and $p'\in \mathcal{P}^k_{1,i}$, we say the pair $(p,p')$ is \emph{admissible} if they satisfy the property stated in Lemma~\ref{lem:path_overlapp}. We use $\cd^{(1,v,d)}_p$ to denote the event that $p$ is the highest-order $d$-path between $1$ and $v$ in $\cg_1'$, and $\cd^{(1,i,k)}_{p'}$ to denote the event that $p'$ is the highest-order $d$-path between $1$ and $i$ in $\cg_1'$. Then we have
    \begin{align}
        &\P(\cc_i\cond \ca_k,v\in V_d,i\in V_k)\nonumber\\
        &=\sum_{\substack{p\in \mathcal{P}^d_{1,v},p'\in\mathcal{P}^k_{1,i}\nonumber\\(p,p')\text{ admissible}}}\P(\cc_i\cap\cd^{(1,v,d)}_p\cap\cd^{(1,i,k)}_{p'}\cond\ca_k,v\in V_d,i\in V_k)\\
        &\le \max_{\substack{p\in \mathcal{P}^d_{1,v},p'\in\mathcal{P}^k_{1,i}\nonumber\\(p,p')\text{ admissible}}}\P(\cc_i\cond \cd^{(1,v,d)}_p,\cd^{(1,i,k)}_{p'},\ca_k,v\in V_d,i\in V_k)\\
        &=\max_{\substack{p\in \mathcal{P}^d_{1,v},p'\in\mathcal{P}^k_{1,i}\nonumber\\(p,p')\text{ admissible}}}\P(\cc_i\cond \ca_k, \{\text{paths $p,p'$ exist in $\cg_1'$}\},\nonumber\\
        &\quad\quad\quad\quad\quad\quad\quad\quad\quad\{\text{no path in $\mathcal{P}^d_{1,v}$ with order higher than $p$ exists in $\cg_1'$ }\},\nonumber\\
        &\quad\quad\quad\quad\quad\quad\quad\quad\quad\{\text{no path in $\mathcal{P}^k_{1,i}$ with order higher than $p'$ exists in $\cg_1'$ }\},\nonumber\\
        &\quad\quad\quad\quad\quad\quad\quad\quad\quad \{\text{no path between $1$ and $u$ with length less than $d$ exists in $\cg_1'$}\},\nonumber\\
        &\quad\quad\quad\quad\quad\quad\quad\quad\quad \{\text{no path between $1$ and $i$ with length less than $k$ exists in $\cg_1'$}\})\nonumber
    \end{align}
    Notice that $\cc_i$ is an increasing graph property in $\bar{\cg}$ because if $\cc_i$ holds in $\bar{\cg}$, then the property still hold after adding an arbitrary edge into $\bar{\cg}$.  On the other hand events $\ca_k$,  
    $$ \{\text{no path in $\mathcal{P}^d_{1,v}$ with order higher than $p$ exists in $\cg_1'$ }\},$$
    $$\{\text{no path in $\mathcal{P}^k_{1,i}$ with order higher than $p'$ exists in $\cg_1'$ }\},$$
    $$\{\text{no path between $1$ and $u$ with length less than $d$ exists in $\cg_1'$}\}$$
    and
    $$\{\text{no path between $1$ and $i$ with length less than $k$ exists in $\cg_1'$}\}$$
    are decreasing graph properties in $\cg_1'$ because if they hold in $\cg_1'$, removing any edge from $\cg_1'$ does not break the properties.
    Since $\bar{\cg}$ is the union graph of $\cg_1'$ and $\cg_2$, we can use Harris inequality~\citep{harris1960lower} to obtain
    \begin{align*}
        &\P(\cc_i\cond \ca_k,v\in V_d,i\in V_k)\\
        &\le\max_{\substack{p\in \mathcal{P}^d_{1,v},p'\in\mathcal{P}^k_{1,i}\nonumber\\(p,p')\text{ admissible}}}\P(\cc_i\cond\{\text{paths $p,p'$ exist in $\cg_1'$}\}).
    \end{align*}
    Recall that $\cc_i$ is the event that the $2l$-neighborhood of vertex $\bar{i}$ in $\bar{\cg}$ contains a cycle, and that $\cg_1'$ is a subgraph of $\bar{\cg}$. Therefore, for any two paths $p,p'$, we have
    \begin{equation*}
        \P(\cc_i\cond\{\text{paths $p,p'$ exist in $\cg_1'$}\})=\P(\cc_i\cond\{\text{paths $p,p'$ exist in $\bar{\cg}$}\}).
    \end{equation*}
    It then follows that
    \begin{equation*}
        \P(\cc_i\cond \ca_k,v\in V_d,i\in V_k)\le \max_{\substack{p\in \mathcal{P}^d_{1,v},p'\in\mathcal{P}^k_{1,i}\nonumber\\(p,p')\text{ admissible}}}\P(\cc_i\cond\{\text{paths $p,p'$ exist in $\bar{\cg}$}\}).
    \end{equation*}
    
    In the sparse \erdos--\renyi graph $\bar{\cg}\sim\mathrm{ER}(n,\bar{\lambda}/n)$, the $2l$-neighborhood of a fixed vertex $\bar{i}$ contains a cycle with probability at most $n^{-1+o(1)}$. In our setting, we need to show that this property continues to hold even after conditioning on the presence of two paths $p$ and $p'$. The following lemma provides an upper bound on the probability that the neighborhood of a vertex contains a cycle give the existence of a well-behaved structure in the graph. We defer the proof of this lemma to Appendix~\ref{appd:properties}.
    \begin{restatable}{lemma}{lemcyclefree}
    \label{lem:cyclefree_given_structure}
        Consider $\cg\sim\mathrm{ER}(n,\frac{\mu}{n})$, where $\mu>1$ is a constant independent of $n$. Let $i$ be an arbitrary but fixed vertex in $\cg$, and $\ch$ be a tree labeled with indices in $[n]$ that satisfies:
        \begin{enumerate}
            \item $\ch$ has at most $K\log n$ vertices for some constant $K>0$;
            \item Vertex $i$ is in $\ch$;
            \item There exists a constant $C>0$ such that the number of vertices in $\ch$ at distance $d$ from $i$ is at most $C$ for every $d\ge 0$.
        \end{enumerate}
        Then we have
        \begin{equation}
            \P(\text{the $2\sqrt{\log n}$-neighborhood of $i$ in $\cg$ contains a cycle}\cond \ch\subset\cg)=O(n^{-1+\gamma})
        \end{equation}
        for any constant $\gamma>0$.
    \end{restatable}
    By Lemma~\ref{lem:path_overlapp} and the fact that $d,k\le \dmax=\lfloor\frac{(1-\eps)\log n}{\log(\lambda q)}\rfloor$, we know that the graph $\ch_{p,p'}$ satisfies the three conditions stated in Lemma~\ref{lem:cyclefree_given_structure}. Therefore, it follows by Lemma~\ref{lem:cyclefree_given_structure} that
    \begin{equation}
        \label{eq:ineqv_ci}
        \P(\cc_i\cond \ca_k,v\in V_d,i\in V_k)=O(n^{-1+\eps/4}).
    \end{equation}
    Substituting~\eqref{eq:ineqv_ci} into~\eqref{eq:ineqv} gives 
    \begin{equation}
        \label{eq:ineqv_final}
        \P(\cb_{i,k}\cap\cc_i|\ca_k,v\in V_d)\le \frac{K(\lambda q)^k\log n}{n-2}\cdot O(n^{-1+\eps/4}).
    \end{equation}

    \noindent\underline{\textbf{Case 1.2. $i= v$:}}
    In this case, we can write
    \begin{align}
        \P(\cb_{i,k}\cap\cc_i\cond \ca_k,i\in V_d)&=\P(\cc_i\cond \ca_k,i\in V_k)\indi\{k=d\}\nonumber\\
        &\le \P(\cc_i\cond \ca_k,i\in V_k).
    \end{align}
    This is because $\cb_{i,k}$ is the event $i\in V_k$. Therefore, we have $\P(\cb_{i,k}\cap\cc_i\cond \ca_k,i\in V_d)=0$ if $k\neq d$ and $\P(\cb_{i,k}\cap\cc_i\cond \ca_k,i\in V_d)=\P(\cc_i\cond \ca_k,i\in V_k)$ if $k=d$. Recall the definition of the set $\mathcal{P}^k_{1,i}$. We have
    \begin{align*}
        &\P(\cc_i\cond \ca_k,i\in V_k)\\
        &=\sum_{\substack{p'\in \mathcal{P}^k_{1,i}}}\P(\cc_i\cap\cd^{(1,i,k)}_{p'}\cond\ca_k,i\in V_k)\\
        &\le \sum_{\substack{p'\in \mathcal{P}^k_{1,i}}}\P(\cc_i\cond \cd^{(1,i,k)}_{p'},\ca_k,i\in V_k)\\
        &=\max_{\substack{p'\in \mathcal{P}^k_{1,i}}}\P(\cc_i\cond \ca_k, \{\text{path $p'$ exists in $\cg_1'$}\},\nonumber\\        
        &\quad\quad\quad\quad\quad\quad\quad\quad\{\text{no path in $\mathcal{P}^k_{1,i}$ with order higher than $p'$ exists in $\cg_1'$ }\},\nonumber\\
        &\quad\quad\quad\quad\quad\quad\quad\quad \{\text{no path between $1$ and $i$ with length less than $k$ exists in $\cg_1'$}\})\nonumber
    \end{align*}
    Similarly to the previous case, we can apply Harris inequality to obtain
    \begin{equation}
        \P(\cc_i\cond \ca_k,i\in V_k)\le \max_{\substack{p'\in \mathcal{P}^k_{1,i}}}\P(\cc_i\cond \text{path $p'$ exists in $\cg_1'$})
    \end{equation}
    Notice that path $p'$ can be viewed as a subgraph of $\bar{\cg}$ that satisfies the requirements in Lemma~\ref{lem:cyclefree_given_structure}. Therefore, Lemma~\ref{lem:cyclefree_given_structure} implies that
    \begin{equation}
    \label{eq:i=v_final}
        \P(\cc_i\cond \ca_k,i\in V_k)=O(n^{-1+\eps/4}).
    \end{equation}
    Finally,~\eqref{eq:ineqv_final} and~\eqref{eq:i=v_final} together imply that
    \begin{align}
        \P(\ce_{1,k}\cond \ca_k,v\in V_d)&\le \P(\cb_{i,k}\cap\cc_i\cond \ca_k,i\in V_d)+\sum_{i\neq v} \P(\cb_{i,k}\cap\cc_i|\ca_k,v\in V_{d})\nonumber\\
        &\le
       K(\lambda q)^k\log n\cdot O(n^{-1+\eps/4}),
    \end{align}
    which completes the proof of~\eqref{eq:e1_k} in the case of $k\neq 0$.

    \noindent\underline{\textbf{Case 2. $k= 0$:}}
    In this case, we want to show that
    \[
    \P(\ce_{1,0}\cond\ca_0,v\in V_d)\le K\log n\cdot O(n^{-1+\eps/4}).
    \]
    However, the set $S_{\cg_1'}(1,0)$ contains the vertex vertex $1$. Therefore, we have
    \[
    \P(\ce_{1,0}\cond\ca_0,v\in V_d)=\P(\cc_1\cond v\in V_d).
    \]
    It follows that
    \begin{align*}
        \P(\cc_1\cond v\in V_d)
        &=\sum_{\substack{p\in \mathcal{P}^d_{1,v}}}\P(\cc_1\cap\cd^{(1,v,d)}_{p}\cond v\in V_d)\\
        &\le \sum_{\substack{p\in \mathcal{P}^d_{1,v}}}\P(\cc_1 \cond \cd^{(1,v,d)}_{p},v\in V_d)\\
        &=\max_{\substack{p\in \mathcal{P}^d_{1,v}}}\P(\cc_1\cond \{\text{path $p$ exists in $\cg_1'$}\},\nonumber\\        
        &\quad\quad\quad\quad\quad\quad\{\text{no path in $\mathcal{P}^d_{1,v}$ with order higher than $p$ exists in $\cg_1'$ }\},\nonumber\\
        &\quad\quad\quad\quad\quad\quad \{\text{no path between $1$ and $v$ with length less than $d$ exists in $\cg_1'$}\})\nonumber\\
        &\le \max_{\substack{p\in \mathcal{P}^d_{1,v}}}\P(\cc_1\cond \{\text{path $p$ exists in $\cg_1'$}\}).
    \end{align*}
    It then follows by Lemma~\ref{lem:cyclefree_given_structure} that
    \begin{equation}
        \P(\ce_{1,0}\cond\ca_0,v\in V_d)=\P(\cc_1\cond v\in V_d)=O(n^{-1+\eps/4}),
    \end{equation}
which proves~\eqref{eq:e1_k} in the case of $k=0$.

\subsubsection{Proof of the unconditional probability bound~\eqref{eq:e1_uncond}}
Recall that for each $k\in \{0,\ldots,\dmax\}$, we define event 
\[
\ca_k=\{|N_{\cg_1'}(1,k)|\le K(\lambda q)^k\log n\}
\]
and
\[
\ce_{1,k}=\{\nexists i\in [n]:i\in S_{\cg_1'}(1,k)\text{ and }\text{the $2l$-neighborhood of $\bar{i}$ in $\bar{\cg}$  contains a cycle}\}.
\]
By Lemma~\ref{lem:neighborhood_size} in Appendix~\ref{appd:properties}, we have
\[
\P(\cap_{k=0}^{\dmax}\ca_k)\ge 1-\td{O}(n^{-1/3}).
\]
It follows by the chain rule that
\begin{align*}
    \P(\ce_1^c)&=\P(\cup_{k=0}^\dmax\ce_{1,k}^c)\\
        &\le \P((\cup_{k=0}^\dmax\ce_{1,k}^c)\cap (\cap_{k=0}^{\dmax}\ca_k))+\P((\cap_{k=0}^{\dmax}\ca_k)^c)\\
        &\le \P((\cap_{k=0}^{\dmax}\ca_k)^c)+\sum_{k=0}^{\dmax}\P(\ce_{1,k}^c|\ca_k)\\
        &=\td{O}(n^{-1/3})+\sum_{k=0}^{\dmax}\P(\ce_{1,k}^c|\ca_k).
\end{align*}
It now suffices to show that for each $k\in\{0,\ldots,\dmax\}$
\begin{equation}
    \label{eq:e1_k_uncond}
    \P(\ce_{1,k}^c|\ca_k)\le K(\log n)(\lambda q)^k\cdot O(n^{-1+\eps/4}).
\end{equation}
This is because~\eqref{eq:e1_k_uncond} implies
\begin{align*}
    \P(\ce_1^c)\le \td{O}(n^{-1/3})+O(n^{-1+\eps/4})\cdot\sum_{k=0}^\dmax K(\log n)(\lambda q)^k=O(n^{-\eps/3}),
\end{align*}
which completes the proof of~\eqref{eq:e1_uncond}. To prove~\eqref{eq:e1_k_uncond}, we again separately consider the cases of $k\neq 0$ and $k=0$.

\noindent\underline{\textbf{Case 1. $k\neq 0$:}} For each vertex $i\in [n]$, recall the definition of events $\cb_{i,k}$ and $\cc_i$ in the previous section. We have
\begin{align*}
    \P(\ce_{1,k}^c|\ca_k)&\le \sum_{i=2}^n \P(\cb_{i,k}\cap\cc_i\cond \ca_k)\\
    &\le \sum_{i=2}^n\P(\cc_i\cond\ca_k,i\in V_k)\cdot\P(\cb_{i,k}\cond \ca_k)\\
    &\le \frac{K(\log n)(\lambda q)^k}{n-1} \sum_{i=2}^n \P(\cc_i\cond\ca_k,i\in V_k),
\end{align*}
where the last inequality follows because we have $|S_{\cg_1'}(1,k)|\le K(\log n)(\lambda q)^k$ on the event $\ca_k$.
    It follows by~\eqref{eq:i=v_final} in the previous subsection that for any $i\in \{2,\ldots,n\}$ and $k\in [\dmax]$, we have
    \[
    \P(\cc_i\cond\ca_k,i\in V_k)=O(n^{-1+\eps/4}).
    \]
    This implies that
    \[
    \P(\ce_{1,k}^c|\ca_k)\le K(\log n)(\lambda q)^k\cdot O(n^{-1+\eps/4}),
    \]
    which proves~\eqref{eq:e1_k_uncond} in the case of $k\neq 0$.

    \noindent\underline{\textbf{Case 2. $k= 0$:}} In this case, we have
    \[
    \P(\ce_{1,0}^c\cond\ca_k)=\P(\cc_1),
    \]
    because the set $S_{\cg_1'}(1,0)$ only contains one single vertex $1$. We can set the graph $\ch$ in Lemma~\ref{lem:cyclefree_given_structure} as the trivial graph that only has one vertex and no edge. This implies that
    \[
    \P(\ce_{1,0}^c\cond\ca_k)=\P(\cc_1)=O(n^{-1+\eps/4}),
    \]
    and completes the proof of~\eqref{eq:e1_k_uncond} in the case of $k=0$.

    \subsection{Proof of Lemma~\ref{prop:E2}}
    We prove Lemma~\ref{prop:E2} in this section. In this proof, we take a different path from the proof in the previous section. We start by proving the unconditioned probability bound~\eqref{eq:e2_uncond}. We then proceed to prove the conditional probability bound~\eqref{eq:e2_cond}, by studying the multiplicative gap between $\P(\ce_2^c)$ and $\P(\ce_2^c\cond v\in V_d)$.

    \subsubsection{Proof of the unconditional probability bound~\eqref{eq:e2_uncond}}
Recall that we defined $\ce_2$ as the event that there exist no triplets $(i,j,k)\in [n]^3$ with $j\neq k$ that satisfy all of the following properties:
\begin{enumerate}[itemsep=6pt]
    \item $i \in V_{\le\dmax+1}$;
    \item 
    $L_{l-1}(\mathcal{T}^\mathrm{IC}_{i,l-1},\td{\ct}^{\setminus\td{j}}_{\td{k},l-1})> \exp\left(\frac{(\lambda s q)^{l-1}}{\log n}\right)$;
    \item $V_{\mathcal{T}^\mathrm{IC}_{i,l-1}}\cap V_{\td{\ct}^{\setminus\td{j}}_{\td{k},l-1}}=\emptyset.$
    \end{enumerate}

To simplify the notation, we use $\theta$ to denote the likelihood threshold $\exp\left(\frac{(\lambda s q)^{l-1}}{\log n}\right)$ in this proof. We also omit the superscript from notations $\mathbb{Q}^{(\lambda/s,qs,s)}_{l-1}$, $\mathbb{P}^{(\lambda/s,qs,s)}_{l-1}$ and $L^{(\lambda/s,qs,s)}_{l-1}$ as is clear from the context.
To bound the probability of $\P(\ce_2^c)$, we define the following events corresponding to the three properties above. For an index $i\in [n]$, define  $\mathrm{IN}_i=\{i\in V_{\le \dmax+1}\}$ as the event that $i$ is in the set $V_{\le\dmax+1}$
For a triplet $(i,j,k)\in [n]^3$, define 
\[
\mathrm{PASS}_{i,j,k}=\left\{L_{l-1}(\mathcal{T}^\mathrm{IC}_{i,l-1},\td{\ct}^{\setminus\td{j}}_{\td{k},l-1})> \theta\right\}
\]
as the event that $\mathcal{T}^\mathrm{IC}_{i,l-1}$ and $\td{\ct}^{\setminus\td{j}}_{\td{k},l-1}$ pass the tree correlation test,
and 
$$\mathrm{NO}_{i,j,k}=\left\{V_{\mathcal{T}^\mathrm{IC}_{i,l-1}}\cap V_{\td{\ct}^{\setminus\td{j}}_{\td{k},l-1}}=\emptyset\right\}$$
as the event that the two trees $\mathcal{T}^\mathrm{IC}_{i,l-1}$ and $\td{\ct}^{\setminus\td{j}}_{\td{k},l-1}$ have no overlapping in vertices.
Here, the subtree $\mathcal{T}^\mathrm{IC}_{i,l-1}$ is well-defined only if vertex $i$ is in $\tic$. When this not the case, we extend the definition by setting $\mathcal{T}^\mathrm{IC}_{i,l-1}$ as the trivial tree without any vertices, and let $L_{l-1}(\mathcal{T}^\mathrm{IC}_{i,l-1},\td{\ct}^{\setminus\td{j}}_{\td{k},l-1})=0$. We also define two auxiliary events
\[
\cn=\{|N_{\cg_1'}(1,\dmax+l)|\le n^{1-\epsilon/2}\}
\]
and 
\[
\bar{\cn}=\{|N_{\bar{\cg}}(\bar{i},l-1)|\le n^{\epsilon/4},\forall i\in [n]\}
\]
to control the neighborhood sizes in $\cg_1'$ and $\bar{\cg}$. It follows by Lemma~\ref{lem:neighborhood_size} that
\[
\P(\cn^c)=\td{O}(n^{-7/3}),
\]
and
\[
\P(\bar{\cn}^c)=\td{O}(n^{-4/3}).
\]
By the union bound, we have
\begin{align}
    \P(\ce_2^c)&\le \P(\ce_2^c\cap\cn\cap\bar{\cn})+\td{O}(n^{-4/3})\nonumber\\
    &\le \td{O}(n^{-4/3})+\sum_{(i,j,k)\in [n]^3:j\neq k}\P(\rmin_i\cap\pass_{i,j,k}\cap\no_{i,j,k}\cap\cn\cap\bar{\cn}).
\end{align}

Fix $i,j,k\in[n]^3$ with $j\neq k$. Let $\cx_{l-1}$ denote the set of all unlabeled rooted trees with depth up to $l-1$. We have
\begin{align}
    &\P(\rmin_i\cap\pass_{i,j,k}\cap\no_{i,j,k}\cap\cn\cap\barn)\nonumber\\
    &\le \sum_{t,\td{t}\in \cx_{l-1}}\P\left(\rmin_i\cap\no_{i,j,k}\cap\{\mathcal{T}^\mathrm{IC}_{i,l-1}\cong t\}\cap\{\td{\ct}^{\setminus\td{j}}_{\td{k},l-1}\cong\td{t}\}\cap\cn\cap\barn\right)\indi\{L_{l-1}(t,\td{t})> \theta\}\nonumber\\
    &\le \sum_{\substack{t,\td{t}\in \cx_{l-1}\\|V_t|\le n^{\epsilon/4}\\|V_{\td{t}}|\le n^{\epsilon/4}}}\P\left(\rmin_i\cap\no_{i,j,k}\cap\{\mathcal{T}^\mathrm{IC}_{i,l-1}\cong t\}\cap\{\td{\ct}^{\setminus\td{j}}_{\td{k},l-1}\cong\td{t}\}\cap\cn\right)\indi\{L_{l-1}(t,\td{t})> \theta\},\label{eq:tree_realizations}
\end{align}
where the last inequality follows because event $\barn$ implies both $\mathcal{T}^\mathrm{IC}_{i,l-1}$ and $\td{\ct}^{\setminus\td{j}}_{\td{k},l-1}$ have at most $n^{\eps/4}$ vertices.

To further bound this summation,
we take a slight detour to consider the probability of passing the likelihood test under the independent Galton--Watson tree distribution $\mathbb{Q}_{l-1}$. 
By the definition of $L_{l-1}$ in~\eqref{eq:LR}, we have 
\[
\E_{(\ct,\tdct )\sim \mathbb{Q}_{l-1}}[L_{l-1}(\ct,\tdct)]=
\sum_{t,\td{t}\in \cx_{l-1}}\mathbb{Q}_{l-1}(t,\td{t})\cdot L_{l-1}(t,\td{t})=1,
\]
and it follow by Markov's inequality that 
\begin{align}
    \sum_{\substack{t,\td{t}\in\mathcal{X}_{l-1}}}\mathbb{Q}_{l-1}(t,\td{t})\cdot\indi(L_{l-1}(t,\td{t})> \theta)
    &=\mathbb{Q}_{l-1}\left(L_{l-1}(\ct,\tdct)>\exp\left(\frac{(\lambda s q)^{l-1}}{\log n}\right)\right)\nonumber\\
    &\le \exp\left(-\frac{(\lambda s q)^{l-1}}{\log n}\right)=n^{-\omega(1)},\label{eq:LR_markov}
\end{align}
where the last inequality follows by our choice $l=\lfloor\sqrt{\log n}\rfloor$.

Given~\eqref{eq:LR_markov}, if we can bound the ratio 
\begin{equation}
\label{eq:const_ratio}
    \frac{\P\left(\rmin_i\cap\no_{i,j,k}\cap\{\mathcal{T}^\mathrm{IC}_{i,l-1}\cong t\}\cap\{\td{\ct}^{\setminus\td{j}}_{\td{k},l-1}\cong\td{t}\}\cap\cn\right)}{\mathbb{Q}_{l-1}(t,\td{t})}=O(1)
\end{equation}
for every pair $(t,\td{t})$ that satisfies $|V_t|\le n^{\eps/4}$ and $|V_{\td{t}}|\le n^{\eps/4}$. Then~\eqref{eq:tree_realizations} and~\eqref{eq:LR_markov} together imply that
\[
\P(\rmin_i\cap\pass_{i,j,k}\cap\no_{i,j,k}\cap\cn\cap\barn)=n^{-\omega(1)},
\]
and we get 
\[
\P(\ce_2^c)\le \td{O}(n^{-4/3})+n^3\cdot n^{-\omega(1)}=\td{O}(n^{-4/3}),
\]
which completes the proof of~\eqref{eq:e2_uncond}.

Therefore, in the rest of this proof, we focus on showing~\eqref{eq:const_ratio}.
As is clear of the context, we omit all the subscripts from notations $\rmin_i$ and $\no_{i,j,k}$. We also use short-hand notation $\ct$ and $\tdct$ for  $\mathcal{T}^\mathrm{IC}_{i,l-1}$ and $\td{\ct}^{\setminus\td{j}}_{\td{k},l-1}$ respectively.
Recall that $\mathbb{Q}_{l-1}$ is the product distribution of two independent Galton--Watson tree distributions, and we can write 
\[
\mathbb{Q}_{l-1}(t,\td{t})=\gw_{l-1,\lambda q}(t)\cdot \gw_{l-1,\lambda}(\td{t}).
\]
On the other hand, we can apply the chain rule to write
\begin{align*}
    &\P\left(\rmin\cap\no\cap\{\ct\cong t\}\cap\{\tdct\cong\td{t}\}\cap\cn\right)\\
    &=\P(\tdct\cong\td{t})\cdot \P\left(\rmin\cap\no\cap\{\ct\cong t\}\cap\cn\;\bigg|\; \{\tdct\cong\td{t}\}\right).
\end{align*}
Therefore, to prove~\eqref{eq:const_ratio}, it suffices to show both
\begin{equation}
    \label{eq:bfs_ratio}
    \frac{\P(\tdct\cong\td{t})}{\gw_{l-1,\lambda}(\td{t})}=O(1),
\end{equation}
and 
\begin{equation}
    \label{eq:subtree_ratio}
    \frac{\P\left(\rmin\cap\no\cap\{\ct\cong t\}\cap\cn\;\bigg|\; \{\tdct\cong\td{t}\}\right)}{\gw_{l-1,\lambda q}(t)}=O(1).
\end{equation}

\emph{\textbf{Proof of~\eqref{eq:bfs_ratio}.}}
The desired equation~\eqref{eq:bfs_ratio} essentially states that the law of a local BFS tree $\tdct$ in graph $\tdcg_2$ is within a constant factor from the Galton--Watson tree distribution $\gw_{l-1,\lambda}(\td{t})$. 
To bound the ratio, we consider a standard \emph{sequential exploration process} of sampling the BFS tree rooted at $\td{k}$ in the graph $\td{\cg}_2\setminus \td{j}$. The process begins by sampling the children of the root vertex $\td{k}$; the number of such children is distributed as $\bin(n-2,\lambda/n)$, since vertices $\td{j}$ and $\td{k}$ are removed from consideration. 
The exploration then proceeds breadth-first, following the canonical order of vertices within each depth. When a vertex $\td{u}$ is explored, the number of its children is distributed as $\bin(n-2-x_{\td{u}},\lambda /n)$, where $x_{\td{u}}$ denotes the number of vertices already uncovered by the exploration up to that point. Equivalently, $x_{\td{u}}$ is the sum of binomial random variables sampled in the previous steps. This sequential sampling process continues layer by layer up to depth $l-1$, and the resulting tree has the same distribution as $\tdct$.

To characterize the probability $\P(\tdct\cong\td{t})$, we introduce some auxiliary notation for the unlabeled tree $\td{t}$. For each vertex $\td{u}$ of $\td{t}$ at depth at most $l-2$, we fix an arbitrary ordering of its children. This induces a total order on the vertices at each depth $d$ of $\td{t}$, defined recursively as follows: a vertex is ranked higher if its parent is ranked higher among vertices at depth $d-1$; among vertices sharing the same parent, the order is determined by the prescribed ordering of siblings. For each $d\le l-2$, let $\td{c}_{d,r}$ denote the number of children of the $r$-th vertex at depth $d$ of $\td{t}$ and define
\[
\td{x}_{d,r}=\sum_{d'=0}^{d-1}\sum_{r'=1}^{|V_{\td{t},d'}|}\td{c}_{d',r'} +\sum_{r'=1}^{r-1}\td{c}_{d,r'}.
\]

With this notation, we have
\begin{align*}
    \P(\tdct\cong\td{t})=\prod_{d=0}^{l-2}\prod_{r=1}^{|V_{\td{t},d}|}\P(\bin(n-2-\td{x}_{d,r},\lambda/n)=\td{c}_{d,r}).
\end{align*}
An upper bound on the multiplicative gap between binomial and Poisson distributions is provided in Lemma~\ref{lem:binom-poi-gap} in Appendix~\ref{appd:technical}.
Applying this lemma yields
\begin{align*}
    \P(\tdct\cong\td{t})&\le \prod_{d=0}^{l-2}\prod_{r=1}^{|V_{\td{t},d}|}\exp\left(\frac{\lambda}{n}(2+\td{x}_{d,r}+\td{c}_{d,r})\right)\P(\poi(\lambda) =\td{c}_{d,r})\\
    &\stackrel{(a)}{\le } \gw_{l-1,\lambda}(\td{t})\cdot  \prod_{d=0}^{l-2}\prod_{r=1}^{|V_{\td{t},d}|}\exp\left(\frac{\lambda}{n}(2+|V_{\td{t}}|)\right)\\
    &\stackrel{(b)}{\le}\gw_{l-1,\lambda}(\td{t})\cdot\exp\left(\frac{\lambda|V_{\td{t}}|}{n}(2+|V_{\td{t}}|)\right)
\end{align*}
where (a) follows because we always have $\td{x}_{d,r}+\td{c}_{d,r}\le |V_{\td{t}}|$ and (b) follows because there are at most $|V_{\td{t}}|$ terms in the product. Under the assumption that $|V_{\td{t}}|\le n^{\eps/4}$, we have
\[
\exp\left(\frac{\lambda|V_{\td{t}}|}{n}(2+|V_{\td{t}}|)\right)=O(1),
\]
which completes the proof of~\eqref{eq:bfs_ratio}.

\emph{\textbf{Proof of~\eqref{eq:subtree_ratio}.}}
By the chain rule, it suffices to show that 
\begin{equation*}
    \frac{\P\left(\no\cap\{\ct\cong t\}\cap\cn\;\bigg|\; \{\td{\ct}\cong\td{t}\}\cap\rmin\right)}{\gw_{l-1,\lambda q}(t)}=O(1).
\end{equation*}

Recall the definition of quantities $\td{c}_{d,r}$ in the context of tree $\td{t}$. For tree $t$, we again fix an arbitrary ordering of the children of each vertices within depth $l-2$, and give a total order of the vertices at each depth. We use $c_{d,r}$ to denote the number of children of the $r$-th vertex at depth $d$ of $t$. We also introduce a few notations regarding the subtree $\ct$ in $\tic$. For each depth $d\in \{0,\ldots,l-2\}$, let $v_{d,r}$ denote the $r$-th vertex at depth $d$ of $\ct$. Here the ranking among the vertices at the same depth in $\ct$ follows the canonical ordering of the vertices in $\tic$. We use $C_{d,r}$ to denote the number of children of $v_{d,r}$, and use $X_{d,r}$ to denote the number of vertices that have already been explored in $\tic$ before generating the children of $v_{d,r}$. In other words, $X_{d,r}$ is the total number of vertices in $\tic$ that with depth smaller or equal to $v_{d,r}$ \emph{plus} the sum of the number of children of the vertices in $\tic$ that are in the same depth as $v_{d,r}$ with a higher canonical order (see Figure~\ref{fig:X} for an illustration). We use $\ca_{d,r}$ to denote the event that $v_{d,r}$ does not have any children that correspond to the vertices in $V_{\td{\ct}}$.

\begin{figure}[htbp]
    \centering
    \includegraphics[width=0.4\linewidth]{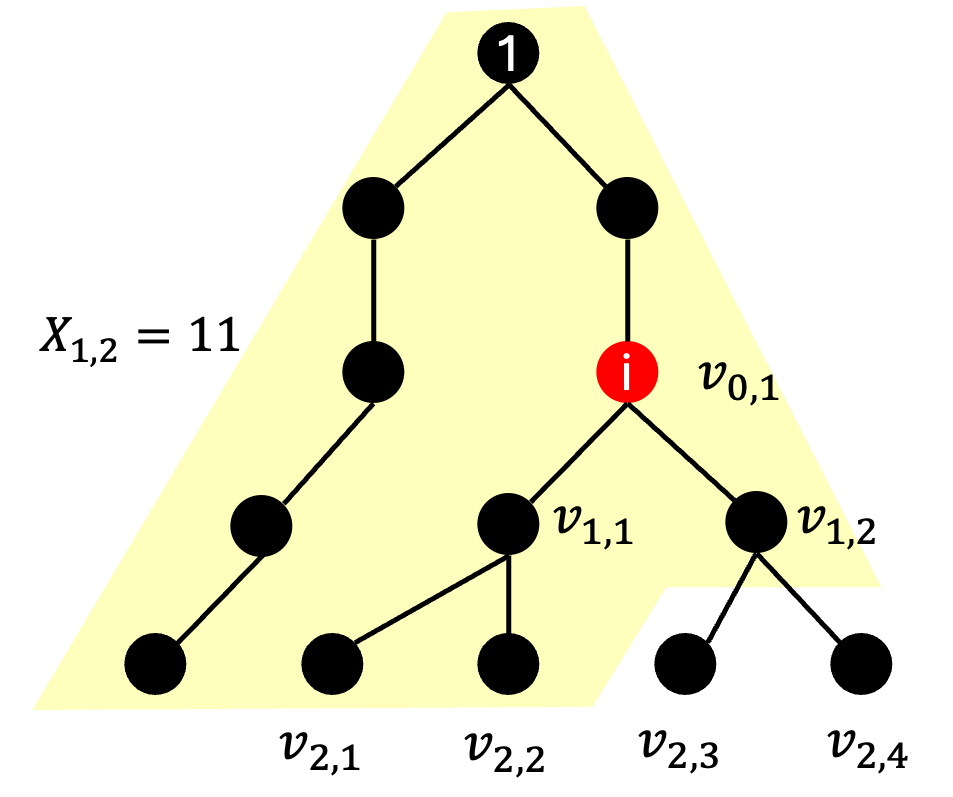}
    \caption{An illustration of the notations defined. The vertex with index $1$ at the top is the root the diffusion tree $\tic$. The red vertex with index $i$ (also vertex $v_{0,1}$ under our notation), is the root of the subtree $\ct$. The eleven vertex in the highlighted area are explored before generating the children of $v_{1,2}$, so we have $X_{1,2}=11$.}
    \label{fig:X}
\end{figure}

By these definitions, we have the following three properties:
\begin{enumerate}
    \item Event $\{\ct=t\}$ implies $\{C_{d,r}=c_{d,r}\}$ for every $d$ and $r$.
    \item Event $\cn$ implies $X_{d,r}\le n^{1-\eps/2}$ for every $d$ and $r$.
    \item Event $\no$ implies $\ca_{d,r}$ for every $d$ and $r$.
\end{enumerate}

By these properties, we have
\begin{align}
    &\P\left(\no\cap\{\ct\cong t\}\cap\cn\;\bigg|\; \{\td{\ct}\cong\td{t}\}\cap\rmin\right)\nonumber\\
    &\le \P\left(\cap_{(d,r)}\left(\ca_{d,r}\cap\{C_{d,r}=c_{d,r}\}\cap\{X_{d,r}\le n^{1-\eps/2}\}\right)\;\bigg|\; \{\td{\ct}\cong\td{t}\}\cap\rmin\right)\nonumber\\
    &= \prod_{(d,r)}\P\bigg(\ca_{d,r}\cap\{C_{d,r}=c_{d,r}\}\cap\{X_{d,r}\le n^{1-\eps/2}\}\nonumber\\
    &\quad\quad\quad\;\bigg|\; \{\td{\ct}\cong\td{t}\}\cap\rmin\cap\left(\cap_{(d',r')\prec(d,r)}\ca_{d',r'}\cap\{C_{d',r'}=c_{d',r'}\}\cap\{X_{d',r'}\le n^{1-\eps/2}\}\right)\bigg),\label{eq:NO}
\end{align}
where the last equality follows by the chain rule. Here, we use the short hand notation $\prod_{(d,r)}$ to denote the product $\prod_{d=0}^{l-2}\prod_{r=1}^{|V_{\ct,d}|}$, and similarly use $\cap_{(d,r)}$ to denote the intersection $\cap_{d=0}^{l-2}\cap_{r=1}^{|V_{\ct,d}|}$. The notation $\cap_{(d',r')\prec(d,r)}$ is a short-hand notation for $\cap_{(d',r'):d'<d \text{ or } (d'=d \text{ and } r'<r)}$. Now we focus on analyzing each term in this product. By the chain rule, we can further move events $\ca_{d,r}\cap \{X_{d,r}\le n^{1-\eps/2}\}$ to the conditioning and get
\begin{align*}
    &\P\bigg(\ca_{d,r}\cap\{C_{d,r}=c_{d,r}\}\cap\{X_{d,r}\le n^{1-\eps/2}\}\\
    &\quad\quad\quad\;\bigg|\; \{\td{\ct}\cong\td{t}\}\cap\rmin\cap\left(\cap_{(d',r')\prec(d,r)}\ca_{d',r'}\cap\{C_{d',r'}=c_{d',r'}\}\cap\{X_{d',r'}\le n^{1-\eps/2}\}\right)\bigg)\\
    &\le \P\bigg(\{C_{d,r}=c_{d,r}\}\;\bigg|\;\{\td{\ct}\cong\td{t}\}\cap\rmin\cap\left(\cap_{(d',r')\prec(d,r)}\ca_{d',r'}\cap\{C_{d',r'}=c_{d',r'}\}\cap\{X_{d',r'}\le n^{1-\eps/2}\}\right)\\
    &\quad\quad\quad\quad\quad\quad\quad\quad\quad\cap\ca_{d,r}\cap\{X_{d,r}\le n^{1-\eps/2}\}\bigg).
\end{align*}
Now consider the distribution of $C_{d,r}$ given all these conditioned events. In the exploration process of generating $\tic$, the already explored vertices (at most $n^{1-\eps/2}$ as we condition on $X_{d,r}\le n^{1-\eps/2}$) cannot be come children of $v_{{d,r}}$. Moreover, the conditioned event $\ca_{d,r}$ implies that none of the vertices corresponding to the set $V_{\tdct}$ can become a children of $v_{d,r}$. Therefore, the total number of excluded candidate for $v_{d,r}$'s children is at most $|V_{\tdct}|+n^{1-\eps/2}$. For the rest of the vertices in $\cg_1'$ that are not excluded, each of them has an edge to $v_{d,r}$ with probability $\lambda q/n$ independently. These facts implies that 
\[
C_{d,r}\sim\bin(n-X,\lambda q/n),
\]
where $X$ is the total number of vertices in $\cg_1'$ that has been explored or corresponds to a vertex in $V_{\tdct}$, and $X$ satisfies $X\le |V_{\tdct}|+n^{1-\eps/2}$ on the conditioned events. By Lemma~\ref{lem:binom-poi-gap}, we get 
\begin{align}
    &\P\bigg(\{C_{d,r}=c_{d,r}\}\;\bigg|\;\{\td{\ct}\cong\td{t}\}\cap\rmin\cap\left(\cap_{(d',r')\prec(d,r)}\ca_{d',r'}\cap\{C_{d',r'}=c_{d',r'}\}\cap\{X_{d',r'}\le n^{1-\eps/2}\}\right)\nonumber\\
    &\quad\quad\quad\quad\quad\quad\quad\cap\ca_{d,r}\cap\{X_{d,r}\le n^{1-\eps/2}\}\bigg)\nonumber\\
    &\le \exp\left(\frac{\lambda q}{n}\cdot(n^{1-\eps/2}+|V_\tdt|)\right)\cdot\P(\poi(\lambda q)=c_{d,r})\nonumber\\
    &\le \exp\left(\frac{\lambda q}{n}\cdot(n^{1-\eps/2}+n^{\eps/4})\right)\cdot \P(\poi(\lambda q)=c_{d,r}),\label{eq:each_binom}
\end{align}
where the last inequality follows because we assume $|V_{\td{t}}|\le n^{\eps/4}$. Finally,~\eqref{eq:NO} and~\eqref{eq:each_binom} imply 
\begin{align*}
    &\P\left(\no\cap\{\ct\cong t\}\cap\cn\;\bigg|\; \{\td{\ct}\cong\td{t}\}\cap\rmin\right)\\
    &\le \prod_{(d,r)}\exp\left(\frac{\lambda q}{n}\cdot(n^{1-\eps/2}+n^{\eps/4})\right)\cdot\P(\poi(\lambda q)=c_{d,r})\\
    &=\gw_{l-1,\lambda q}(t)\prod_{(d,r)}\exp\left(\frac{\lambda q}{n}\cdot(n^{1-\eps/2}+n^{\eps/4})\right)\\
    &\stackrel{(a)}{\le }\gw_{l-1,\lambda q}(t)\cdot\exp\left(\frac{\lambda q n^{\eps/4}}{n}\cdot(n^{1-\eps/2}+n^{\eps/4})\right)\\
    &=\gw_{l-1,\lambda q}(t)\cdot O(1),
\end{align*}
where (a) follows because there are at most $|V_{\td{t}}|$ terms in the product, and we assume $|V_{\td{t}}|\le n^{-\eps/4}$. This completes the proof of~\eqref{eq:subtree_ratio}.

\subsubsection{Proof of the conditional probability bound~\eqref{eq:e2_cond}}
To bound the conditional probability $\P(\ce^c_2\cond v\in V_d)$, we first connect it to a relevant quantity $\P(\ce^c_2\cond V_d\neq\emp)$.  We have
\begin{align*}
    \P(\ce_2^c\cond V_{d}\neq \emptyset)&\ge \P(\ce_2^c\cap\{v\in V_{d}\}\cond V_{d}\neq \emptyset)\\
    &=\P(v\in V_{d}\cond V_{d}\neq \emptyset)\cdot\P(\ce_2^c\cond v\in V_{d})\\
    &\ge \frac{1}{n}\cdot \P(\ce_2^c\cond v\in V_{d}),
\end{align*}
where the last inequality follows by symmetry. Therefore, we have
\[
 \P(\ce_2^c\cond u\in V_{d})\le n\cdot \P(\ce_2^c\cond V_{d}\neq \emptyset).
\]
By the chain rule, we have
\[
\P(\ce_2^c\cond V_{d}\neq \emptyset)\le \frac{\P(\ce^c_2)}{\P(V_d\neq \emp)}.
\]
We have already shown $\P(\ce_2^c)=\td{O}(n^{-4/3})$ in the previous subsection.
Therefore, to complete the proof of~\eqref{eq:e2_cond}, it suffices to show that $\P(V_d\neq \emp)=\Theta(1)$. This is implied by the following technical lemma, which is proven in Appendix~\ref{appd:properties}.
\begin{restatable}{lemma}{nonemptyprob}
        \label{lem:non-empty-neighborhood}
    Consider an \erdos--\renyi random graph $\cg\sim\mathrm{ER}(n,\frac{\mu}{n})$, where $\mu>1$ is a constant independent of $n$. Let $i$ be an arbitrary yet fixed vertex in $\cg$ and $d=c\log n$ for some $c\le\frac{1-\gamma}{\log\mu}$, where $\gamma$ is an arbitrarily small constant. Then for any $k\le d$, we have
    \[
    \P(S_{\cg}(i,k)\neq \emptyset)=\Theta(1).
    \]
    \end{restatable}
Since $V_d=S_{\cg_1'}(1,d)$ and $d\le \dmax\le \frac{1-\eps}{\log (\lambda q)}\log n$, Lemma~\ref{lem:non-empty-neighborhood} implies that $\P(V_d\neq\emp)=\Theta(1)$, which completes the proof.

\section{Proof of Proposition~\ref{prop:inclusion}}\label{appd:pf_fraction}

The proof of Proposition~\ref{prop:inclusion} proceeds in three main steps. In the first step, we analyze the performance of the upward pass phase of Algorithm~\ref{alg:mpmatch} when applied to a pair of trees drawn from the correlated tree distribution $\mathbb{P}_{d+l}^{(\lambda/s,qs,s)}$. We show that the sequence $(p_d)_{d=0}^\infty$, defined in Section~\ref{sec:model}, provides lower bounds on the probability of correctly matching the root under this correlated tree model.

In the second step, we establish Proposition~\ref{prop:inclusion} for vertices at large depths $d\in\{\dmax-l^+\ldots,\dmax\}$ of the diffusion tree $\tic$. We relate the joint distribution of the subtrees of $\tic$ and their corresponding neighborhood in $\tdcg_2$ to the correlated tree distribution $\mathbb{P}_{d}^{(\lambda/s,qs,s)}$.This coupling allows us to transfer the performance guarantees derived under the correlated tree model to the actual execution of our algorithm, showing that $p_d-o(1)$ serves as a lower bound on the probability of correct matching in the original graph setting.

In the final step, we extend the result to vertices at smaller depths $d\in\{0,\ldots,\dmax-l^+-1\}$ of the diffusion tree $\tic$. Intuitively, as the algorithm propagates upward, it aggregates information from an increasing number of descendants, and its performance can only improve. Consequently, the lower bounds established at larger depths continue to hold for all smaller depths. We make this intuition formal via an algorithm truncation argument, which shows that the probability of correct matching is non-decreasing as the algorithm proceeds upward.

\subsection{Step 1: Algorithm performance on correlated Galton--Watson trees}
Consider two rooted trees $\ct$ and $\td{\ct}$ with roots denoted as $r$ and $\td{r}$ respectively. For the purpose of the proof, we introduce an algorithm for matching the corresponding vertices in $\ct$ and $\tdct$ as follows. The algorithm can be viewed as an adaptation of the upward pass phase of Algorithm~\ref{alg:mpmatch} to the case where the inputs are a pair of trees instead of a diffusion tree and a graph. The algorithm takes the two trees $\ct,\tdct$ and a depth parameter $d\ge 0$ as input. For each vertex $i\in V_{\ct,\le d}$, the algorithm maintains a matching candidate set $\cm_i'$, which is initialized to be empty. The algorithm utilizes the following three matching criteria, which are adapted from the first three criteria in Algorithm~\ref{alg:mpmatch} to the case of two input trees.
\begin{enumerate}
    \item Two vertices $i\in V_{\ct}$ and $\td{j}\in V_{\tdct}$ satisfy Criterion $1'$ if there exist three children $i_1,i_2,i_3$ of $i$ in $\ct$ and three children $\td{j}_1,\td{j}_2,\td{j}_3$ of $\td{j}$ in $\tdct$ such that
    \[
    L_{l-1}(\ct_{i_k,l-1},\td{\ct}_{\td{j}_k,l-1})>\exp\left(\frac{(\lambda s q)^{l-1}}{\log n}\right), \forall k\in [3].
    \]
    \item Two vertices $i\in V_{\ct}$ and $\td{j}\in V_{\tdct}$ satisfy Criterion $2'$ if there exist two children $i_1,i_2$ of $i$ in $\ct$ and two children $\td{j}_1,\td{j}_2$ of $\td{j}$ in $\tdct$ such that 
    \[
    \td{j}_1\in \cm_{i_1}\text{ and }  L_{l-1}(\ct_{i_2,l-1},\td{\ct}_{\td{j}_2,l-1})>\exp\left(\frac{(\lambda s q)^{l-1}}{\log n}\right).
    \]
    \item Two vertices $i\in V_{\ct}$ and $\td{j}\in V_{\tdct}$ satisfy Criterion $3'$ if there exist two children $i_1,i_2$ of $i$ in $\ct$ and two neighbors $\td{j}_1,\td{j}_2$ of $\td{j}$ in $\tdct$ such that 
    \[
    \td{j}_k\in \cm_{i_k}, \forall k\in [2].
    \]
\end{enumerate}
This algorithm starts at depth $d$ of both trees. For each $i\in V_{\ct,d}$ and $\td{j}\in V_{\td{\ct},d}$, we add $\td{j}$ to the set $\cm_i'$ if $i$ and $\td{j}$ satisfy Criterion $1'$. Then the algorithm recursively moves upward layer by layer. At depth $d-1$, for each $i\in V_{\ct,d-1}$ and $\td{j}\in V_{\td{\ct},d-1}$, we add $\td{j}$ to the set $\cm_i'$ if $i$ and $\td{j}$ satisfy either one of Criteria $1'$-$3'$. This process is then repeated for each $d=d-1,\ldots,0$ until we reach the root vertex. We provide the pseudocode of this auxiliary algorithm in Algorithm~\ref{alg:auxiliary}.
\begin{algorithm2e}[t]
\caption{Auxiliary algorithm for matching two correlated Galton--Watson trees} 
\label{alg:auxiliary}
\SetAlgoNoEnd
\DontPrintSemicolon
\SetKwInOut{Input}{Input}
\SetKwInOut{Output}{Output}
\Input{Two rooted trees $\ct$ and $\td{\ct}$ with roots $r$ and $\td{r}$, depth parameter $d\ge 0$}
\Output{A matching candidate set $\cm_i'$ for each $i$ within depth $d$ of $\ct$ }
 $\cm_i' \gets \emptyset$ for each $i\in V_{\ct,\le d}$\\
\For{$k=d:-1:0$}{
\For{$i\in V_{\ct,k}$}{
\For{$\td{j}\in V_{\td{\ct},k}$}{
\If{$i$ and $\td{j}$ satisfy at least one of Criteria $1'$-$3'$}{
 $\cm_i'\gets\cm_i'\cup\{\td{j}\}$
}
}
}
}
 \Return{$\cm_i'$, for all $i\in V_{\tic,\le d}$}
\end{algorithm2e}

With the algorithm in hand, we say the two input trees $\ct$ and $\tdct$ satisfy property $P_d$ if executing Algorithm~\ref{alg:auxiliary} with input parameter $d$ on $\ct$ and $\tdct$ outputs $\cm_r=\{\td{r}\}$. In other words, Algorithm~\ref{alg:auxiliary} with input $d$ correctly matches the roots of the two vertices. The following proposition provides a bound on the probability of property $P_d$ under the correlated tree distribution.
\begin{lemma}
\label{prop:auxiliary}
    Suppose parameters $\lambda,s$ and $q$ satisfied the conditions stated in Theorem~\ref{thm:guarantee}.
    Let $d\ge 0$ and $(\ct,\tdct)\sim \mathbb{P}_{d+l}$. We have
    \begin{equation}
        \P(\text{$\ct$ and $\tdct$ satisfy $P_d$})\ge p_d.\label{eq:auxiliary}
    \end{equation}
\end{lemma}
\begin{proof}[Proof of Lemma~\ref{prop:auxiliary}]
    We prove this proposition by induction. First let $d=0$. In this base case, Algorithm~\ref{alg:auxiliary} simply starts at the root level, and verifies whether the roots of the two trees satisfy Criterion $1'$.  Recall that $\ct^*$ denotes the intersection tree in the generation process of $\mathbb{P}_l$. Let $C_r$ denote the number of children of the root in the intersection tree $\ct^*$. By definition of $\mathbb{P}_l$, we have $C_r\sim\poi(\lambda sq)$. We use $x_1,\ldots,x_{C_r}$ to denote the copy of these children in $\ct$ and $\td{x}_1,\ldots,\td{x}_{C_r}$ to denote their counterparts in $\td{\ct}$.
    We then have that the $C_r$ tree pairs $(\ct_{x_1,l-1},\td{\ct}_{\td{x}_1,l-1}),\ldots,(\ct_{x_{C_r},l-1},\td{\ct}_{\td{x}_{C_r},l-1})$ are i.i.d with distribution $\mathbb{P}_{l-1}$.
    For a pair of trees with distribution $\mathbb{P}_{l-1}$, the probability that their likelihood ratio is larger than the threshold $\exp\left(\frac{(\lambda sq)^{l-1}}{\log n}\right)$ is bounded in Corollary~\ref{cor:tree_test_finite} in Appendix~\ref{appd:tree_test}. Setting the constant $\lambda_0$ in Theorem~\ref{thm:guarantee} to be $\lambda'(qs,s)$ in Corollary~\ref{cor:tree_test_finite} yields
    \[
    \P\left(L_{l-1}(\ct_{x_k,l-1},\td{\ct}_{\td{x}_k,l-1})>\exp\left(\frac{(\lambda sq)^{l-1}}{\log n}\right)\right)\ge  1-p^\mathrm{ext}_{\lambda sq}-\eps',
    \]
    for each $k\in [C_r]$, and these $C_r$ events are mutually independent. Let $X_r$ denote the number of these tree pairs satisfying 
    \[
    L_{l-1}(\ct_{x_k,l-1},\td{\ct}_{\td{x}_k,l-1})>\exp\left(\frac{(\lambda sq)^{l-1}}{\log n}\right).
    \]
    We then have 
    \[
    \cl(X_r|C_r=c_r) \stoge\bin(c_r,1-p^\mathrm{ext}_{\lambda sq}-\eps').
    \]
    From Poisson thinning, we get that
    \[
    \cl(X_r)\stoge\poi(\lambda sq(1-p^\mathrm{ext}_{\lambda sq}-\eps')).
    \]
    It then follows that 
    \begin{align*}
        \P(\td{r}\in\cm_r')&\ge\P(X_r\ge 3)\\
        &\ge \P(\poi(\lambda sq(1-p^\mathrm{ext}_{\lambda sq}-\eps'))\ge 3)\\
        &=p_{0},
    \end{align*}
    which completes the proof of the base case.

    For the induction step, suppose the claim in the proposition holds for the case of $d-1$, and let $(\ct,\td{\ct})\sim \mathbb{P}_{d+l}$. We again use $C_r$ to denote the number of children of the root node in the intersection tree $\ct^*$, and we have $C_r\sim\poi(\lambda sq)$. Let us condition on the event that $C_r=c_r$. We use $x_1,\ldots,x_{c_r}$ to denote the copy of these children in $\ct$ and $\td{x}_1,\ldots,\td{x}_{c_r}$ to denote their counterparts in $\td{\ct}$, and we have that the $c_r$ tree pairs $(\ct_{x_1,d-1},\td{\ct}_{\td{x}_1,d-1}),\ldots,(\ct_{x_{c_r},d-1},\td{\ct}_{\td{x}_{c_r},d-1})$ are i.i.d with distribution $\mathbb{P}_{d+l-1}$. For each $k\in [c_r]$, we define a random variable $t_k$ serving as an indicator of the type of vertex pair $(x_k,\td{x}_k)$:
    \[
    t_k=\begin{cases}
        2, \text{ if } \td{x}_k\in\cm_{x_k}\\
        1, \text{ if } L_{l-1}(\ct_{x_k,l-1},\td{\ct}_{\td{x}_k,l-1})>\exp\left(\frac{(\lambda sq)^{l-1}}{\log n}\right)\text{ and }\td{x}_k\notin\cm_{x_k}\\
        0, \text{ if } L_{l-1}(\ct_{x_k,l-1},\td{\ct}_{\td{x}_k,l-1})\le\exp\left(\frac{(\lambda sq)^{l-1}}{\log n}\right)\text{ and }\td{x}_k\notin\cm_{x_k}.
    \end{cases}
    \]
    By the independence between tree pairs $(\ct_{x_1,d-1},\td{\ct}_{\td{x}_1,d-1}),\ldots,(\ct_{x_{c_r},d-1},\td{\ct}_{\td{x}_{c_r},d-1})$, we know that $t_1,\ldots,t_k$ are i.i.d random variables.
    Let $\beta_0:=\P(t_k=0|C_r=c_r)$, $\beta_1:=\P(t_k=1|C_r=c_r)$ and $\beta_2:=\P(t_k=2|C_r=c_r)$. By our induction hypothesis, we know that 
    \begin{equation}
    \label{eq:induction}
        \beta_2\ge p_{d-1}.
    \end{equation}
     By Corollary~\ref{cor:tree_test_finite}, we know that
     \begin{equation}
         \label{eq:by_cor}
         \beta_1+\beta_2\ge 1-p^\mathrm{ext}_{\lambda sq}-\eps' .
     \end{equation}
     Let $Y_r=|\{k\in [c_r]:t_k=2\}|$, $Z_r=|\{k\in [c_r]:t_k=1\}|$ and $W_r=|\{k\in [c_r]:t_k=0\}|$. It then follows that $2Y_r+Z_r=\sum_{k=1}^{c_r}t_k$, and
    \[
    Y_r,Z_r,W_r|C_r=c_r\sim \mathrm{Multi}(c_r,\beta_2,\beta_1,\beta_0).
    \]
    We claim that the event $2Y_r+Z_r\ge 3$ implies that $\td{r}\in\cm_r$. This is because when $2Y_r+Z_r\ge 3$, at least one of the following three must hold true:
    \begin{enumerate}
        \item $Y_r\ge 2$: This implies that $r$ and $\td{r}$ satisfy Criterion $3'$, which implies that $\td{r}\in\cm_r$.
        \item $Y_r\ge 1$ and $Z_r\ge 1$: This implies that $r$ and $\td{r}$ satisfy Criterion $2'$, which implies that $\td{r}\in\cm_r$.
        \item $Z_r\ge 3$: This implies that $r$ and $\td{r}$ satisfy Criterion $1'$, which implies that $\td{r}\in\cm_r$.
    \end{enumerate}
    As a result, we have $\P(\td{r}\in\cm_r|C_r=c_r)\ge \P(2Y_r+Z_r\ge 3|C_r=c_r)$ for any $c_r$, and we will now focus on bounding this conditional probability. Towards this goal, we define a sequence of auxiliary i.i.d random variables $t_1',\ldots,t_{c_r}'$ with distribution
    \[
    \P(t_k'=i)=\begin{cases}
       p_{d-1}\;\;\text{when}\;\; i=2\\
        1-p^\mathrm{ext}_{\lambda sq}-\eps'-p_{d-1}\;\;\text{when}\;\; i=1\\
        p^\mathrm{ext}_{\lambda sq}+\eps'\;\;\text{when}\;\; i=0.
    \end{cases}
    \]
    From~\eqref{eq:induction} and~\eqref{eq:by_cor}, we have $t_k\stoge t_k'$ for each $k\in [c_r]$. Let $Y_r'=|\{k\in [c_r]:t_k'=2\}|$, $Z_r'=|\{k\in [c_r]:t_k'=1\}|$ and $W_r'=|\{k\in [c_r]:t_k'=0\}|$. It then follows that $2Y_r'+Z_r'=\sum_{k=1}^{c_r}t_k'$, and
    \[
    Y_r',Z_r',W_r'|C_r=c_r\sim \mathrm{Multi}(c_r, p_{d-1},1-p^\mathrm{ext}_{\lambda sq}-\eps'-p_{d-1}, p^\mathrm{ext}_{\lambda sq}+\eps').
    \]
    We know that $2Y_r+Z_r\stoge 2Y_r'+Z_r'$. It then follows that
    \[
    \P(\td{r}\in\cm_r|C_r=c_r)\ge \P(2Y_r+Z_r\ge 3|C_r=c_r)\ge \P(2Y_r'+Z_r'\ge 3|C_r=c_r)
    \]
    holds for every $c_r$. So we have
    \[
    \P(\td{r}\in\cm_r')\ge \P(2Y_r'+Z_r'\ge 3).
    \]
    Recall that $C_r\sim\poi(\lambda sq)$. By Poisson thinning, we know that $Y'_r$ and $Z'_r$ are independent Poisson random variables with means $p_{d-1}$ and $1-p^\mathrm{ext}_{\lambda sq}-\eps'-p_{d-1}$ respectively, and we finally get
    \[
    \P(\td{r}\in\cm_r')\ge p_{d},
    \]
    by the definition of $p_{d}$.
\end{proof}

\subsection{Step 2: Proof of Proposition~\ref{prop:inclusion} at large depths}
In this section, we prove~\eqref{eq:inclusion} in Proposition~\ref{prop:inclusion} for depth $d\in\{\dmax-l^+,\ldots,\dmax\}$. In particular, we will show that
\begin{equation}
\label{eq:fraction_large_d}
\P(\td{v}\in \cm_v\cond v\in V_d)\ge p_{\dmax-d}-O(n^{-\eps/3})
\end{equation}
for each $d\in\{\dmax-l^+,\ldots,\dmax\}$. Towards this goal, we will first establish a sufficient condition for the event $\td{v}\in\cm_v$.

Consider a vertex $v$ at depth $d$ of $\tic$. Let $\Nup:=V_{\le d}\setminus\{v\}$ denote the set of other vertices in $\tic$ within depth $d$, and let $\tdNup$ denote the set of vertices in $\tdcg_2$ corresponding to $\Nup$. We define rooted graph 
$\ch$
(resp. $\td{\ch}$) as the induced subgraph on the $(\dmax-d+l)$-neighborhood of $v$ (resp. $\td{v}$) in the graph $\cg_1'\setminus\Nup$ (resp. $\tdcg_2\setminus\tdNup$), with $v$ (resp. $\td{v}$) as its root. We claim that the event $\td{v}\in\cm_v$ holds if the following three conditions are satisfied:
\begin{enumerate}
    \item $\ch$ and $\td{\ch}$ are both trees, and they satisfy the property $P_{\dmax-d}$.
    \item In graph $\tdcg_2$, there exist no edges between the set $\tdNup$ and any non-root vertex of $\td{\ch}$.
    \item $\ch=\ct^\mathrm{IC}_{v,\dmax-d+l}$, where $\ct^{\mathrm{IC}}_{v,\dmax-d+l}$ is the subtree in $\tic$ rooted at $v$ with depth up to $\dmax-d+l$.
\end{enumerate}
\underline{\textbf{Proof of the claim}:} 
To prove the claim, we essentially want to show that if running Algorithm~\ref{alg:auxiliary} with parameter $\dmax-d$ on the two trees $\ch$ and $\td{\ch}$ would output $\td{v}\in\cm_v'$, then Algorithm~\ref{alg:mpmatch} would add vertex $\td{v}$ in to the set $\cm_v$.
Recall that the main difference between these two algorithms is that Criteria $1$ and $2$ in Algorithm~\ref{alg:mpmatch} take subtrees in $\tic$ and the BFS trees in $\tdcg_2$ as input to likelihood ratio tests, while Criteria $1'$ and $2'$ in Algorithm~\ref{alg:auxiliary} take subtrees in the two input trees $\ch$ and $\td{\ch}$ as inputs to likelihood ratio tests.
Consider a vertex $\td{i}$ in $\td{\ch}$ within depth $\dmax-d+1$, and let $\td{j}$ denote its parent. Here the depth and the parent in the graph are well-defined because $\td{\ch}$ is a tree.
By the second assumption, the subtree $\td{\ch}_{\td{i},l-1}$ rooted at $\td{i}$ in $\td{\ch}$ up to depth $l-1$ is exactly the BFS tree $\td{\ct}^{\setminus \td{j}}_{\td{i},l-1}$ in graph $\tdcg_2$. Meanwhile, by the third assumption, we have $\ct_{v,\dmax-d+l}^\rmic=\ch$, and therefore the subtrees in $\ch$ exactly appear as subtrees in $\tic$. Moreover, Algorithm~\ref{alg:auxiliary} starts at depth $\dmax-d$ of the tree $\ch$, which corresponds to the depth $\dmax$ of $\tic$ because the root node $u$ is at depth $d$ in $\tic$. These jusmtify that if Algorithm~\ref{alg:auxiliary} adds a vertex $\td{j}$ in $\td{\ch}$ to $\cm_i'$, then Algorithm~\ref{alg:mpmatch} will also add $\td{j}$ to the set $\cm_i$. Since $\ch$ and $\td{\ch}$ satisfy property $P_{\dmax-d}$, we know that Algorithm~\ref{alg:mpmatch} would include $\td{v}$ in the set $\cm_v$. This completes the proof of the claim.

To prove~\eqref{eq:fraction_large_d}, it now suffices to bound the probability for the three conditions to hold, given the event $u\in V_d$. In particular, we will show that
\begin{equation}
    \label{eq:condition1}
    \P(\text{Condition 1 holds}\cond v\in V_d)\ge p_{\dmax-d}-O(n^{-\eps/3}),
\end{equation}
\begin{equation}
    \label{eq:condition2}
    \P(\text{Condition 2 holds}\cond v\in V_d)\ge 1-O(n^{-\eps/3}),
\end{equation}
and
\begin{equation}
    \label{eq:condition3}
     \P(\text{Condition 3 holds}\cond v\in V_d)\ge 1-O(n^{-\eps/3}),
\end{equation}
which together imply the desired inequality~\eqref{eq:fraction_large_d}.
\subsubsection{Proof of~\eqref{eq:condition1}}
We first provide a high probability upper bound on the size of $\Nup$. Let $\ca=\{|\Nup|\le n^{1-2\eps/3}\}$. It follows by Lemma~\ref{lem:neighborhood_size_cond} that
\begin{equation}
\label{eq:size_given_uvd}
\P(\ca^c\cond v\in V_d)=O(n^{-\eps/3}).
\end{equation}
Now suppose we are given a specific realization of the set $\Nup$. We use $\hat{\mathbb{P}}^{\setminus\Nup}_{\dmax-d+l}$ to denote the joint law of the two trees $\ch$ and $\td{\ch}$. We have the following proposition that bounds the total variation distance between the two distributions $\hat{\mathbb{P}}^{\setminus\Nup}_{\dmax-d+l}$ and $\mathbb{P}_{\dmax-d+l}$.
\begin{restatable}{proposition}{proptreetv}
    \label{prop:tree_tv}
    Suppose $|\Nup|\le n^{1-\gamma}$ for some constant $0<\gamma<1/4$ and $d$ satisfies $\dmax-l^+\le d\le \dmax$. Then there exists a coupling $\sigma$ between $\hat{\mathbb{P}}^{\setminus\Nup}_{\dmax-d+l}$ and $\mathbb{P}_{\dmax-d+l}$ such that 
    \[
    \P_{((\ch,\td{\ch}),(\ct,\td{\ct}))\sim\sigma}(\ch\cong\ct,\td{\ch}\cong\td{\ct} )\ge 1-\td{O}(n^{-3\gamma/4}).
    \]
\end{restatable}
We defer the proof of this proposition to Appendix~\ref{appd:proof_tv}.
By Lemma~\ref{prop:auxiliary} and Proposition~\ref{prop:tree_tv}, we have
\[
\P_{(\ch,\td{\ch})\sim \hat{\mathbb{P}}^{\setminus\Nup}_{\dmax-d+l}}(\ch,\td{\ch}\text{ satisfy property $P_{\dmax-d}$})\ge p_{\dmax-d}-\td{O}(n^{\eps/2}),
\]
for any $\Nup$ with $|\Nup|\le n^{1-2\eps/3}$. 
This implies that 
\begin{equation*}
    \P(\text{Condition 1 holds}\cond v\in V_d,\ca)\ge p_{\dmax-d}-\td{O}(n^{\eps/2}).
\end{equation*}
With this, we finally get 
\begin{align*}
    &\P(\text{Condition 1 holds}\cond v\in V_d)\\
    &\le \P(\text{Condition 1 holds}\cond v\in V_d,\ca)\cdot\P(\ca\cond v\in V_d)\\
    &=\P(\text{Condition 1 holds}\cond v\in V_d,\ca)\cdot(1-\P(\ca^c\cond v\in V_d))\\
    &\le \P(\text{Condition 1 holds}\cond \{v\in V_d\}\cap \ca)-\P(\ca^c\cond v\in V_d)\\
    &=p_{\dmax-d}-O(n^{-\eps/3}),
\end{align*}
which completes the proof of~\eqref{eq:condition1}.

\subsubsection{Proof of~\eqref{eq:condition2}} 
Recall that we defined event $\ca=\{|\Nup|\le n^{1-2\eps/3}\}$, and we have $\P(\ca^c\cond v\in V_d)=O(n^{-\eps/3})$. Notice that $\dmax-d+l=o(\log n)$. We can apply Lemma~\ref{lem:neighborhood_size} to get
\[
\P(|V_{\td{\ch}}|\ge n^{\eps/8}\cond \{v\in V_d\}\cap\ca)=O(n^{-4/3}).
\]
Between a vertex in $\tdNup$ and a non-root vertex in $\td{\ch}$, the probability of having an edge is $O(n^{-1})$. With the high probability bounds on $|V_{\td{\ch}}|$ and $|\Nup|$, we can apply the union bound over all the vertex pairs and get
\[
\P(\text{Condition 2 holds}\cond v\in V_d)\ge 1-O(n^{-\eps/3}+n^{-4/3}+n^{-1}\cdot n^{1-2\eps/3}\cdot n^{\eps/8})=1-O(n^{-\eps/3}).
\]

\subsubsection{Proof of~\eqref{eq:condition3}}
For each vertex $i\in V_d\setminus\{v\}$, let $S_i$ denote the set of vertices within distance $\dmax-d+l$ in the graph $\cg_1'\setminus(V_{\le d}\setminus\{i\})$. Define $S=\cup_{i\in V_d\setminus\{v\}}S_i$. We claim $\ch=\ct^\rmic_{\dmax-d+l}$ if the following two properties hold:
\begin{enumerate}
    \item $\ch$ is a tree;
    \item $V_{\ch}\cap S=\emp$.
\end{enumerate}
To see why these two properties imply the desired event $\ch=\ct^\rmic_{\dmax-d+l}$, notice that the property $V_{\ch}\cap S=\emp$ implies that none of the vertices in $\ch$ can be added into the subtree rooted at any vertex $i\in V_d\setminus\{v\}$. Meanwhile, those vertices are within distance $\dmax-d+l$ from $v$. Therefore, all of the vertices in $\ch$ have to be included in the subtree $\ct^\rmic_{v,\dmax-d+l}$. Because $\ch$ is a tree, all of its edges are also preserved into $\ct^\rmic_{v,\dmax-d+l}$. Therefore, we have $\ch=\ct^\rmic_{v,\dmax-d+l}$. 

It now suffices to bound $\P(\text{$\ch$ is a tree}\cond v\in V_d)$ and $\P(V_{\ch}\cap S=\emp\cond v\in V_d)$.

From~\eqref{eq:size_given_uvd} and Proposition~\ref{prop:tree_tv}, we have
\begin{align}
    \P(\text{$\ch$ is not a tree}\cond v\in V_d)&\le O(n^{-\eps/3})+\P(\text{$\ch$ is not a tree}\cond \{v\in V_d\}\cap \ca)\nonumber\\
    &\le O(n^{-\eps/3})+\td{O}(n^{-\eps/2})\nonumber\\
    &=O(n^{-\eps/3}).\label{eq:not_tree}
\end{align}

Now suppose we are given a specific realization of the set $\Nup$ with $|\Nup|\le n^{1-2\eps/3}$. For each vertices $i\in V_d\setminus\{v\}$, we define event $\cb_i=\{|S_i|\le n^{\eps/8}\}$. It follows by Lemma~\ref{lem:neighborhood_size} that
\[
\P(\cb_i^c\cond v\in V_d,\Nup)=\td{O}(n^{-4/3}).
\]
Since there are at most $ n^{1-2\eps/3}$ vertices in  the set $V_d\setminus\{u\}$, we can apply the union bound to get
\[
\P(\cup_{i\in V_d\setminus\{v\}}\cb_i^c\cond v\in V_d,\Nup)=\td{O}(n^{-1/3}),
\]
which further implies that
\[
\P(|S|\ge n^{1-\eps/2}\cond v\in V_d,\Nup)=\td{O}(n^{-1/3}).
\]
Suppose we are further conditioned on a specific realization of the set $S$ with $|S|\le n^{1-\eps/2}$. For each $k\in [\dmax-d+l]$, we define events
\[
\cc_k=\left\{|S_{\ch}(v,k)|\le 4\log n(\lambda q)^k\prod_{h=0}^k(1+(\lambda q)^{-h/2})\right\}
\]
and 
\[
\cd_k=\{S_{\ch}(v,k)\cap S=\emp\}.
\]
It follows by the definition that $\cap_{k=1}^{\dmax-d+l}\cd_k$ implies that $V_{\ch}\cap S=\emp$. We can write
\begin{align}
    &\P(V_{\ch}\cap S\neq \emp\cond v\in V_d,\Nup,S)\nonumber\\
    &\le \P(\cup_{k=1}^{\dmax-d+l}(\cc_k^c\cup\cd_k^c)\cond v\in V_d,\Nup,S))\nonumber\\
    &\le \sum_{k=1}^{\dmax-d+l}\P(\cc_k^c\cup\cd^c_k\cond v\in V_d,\Nup,S,\cap_{h=1}^{k-1}(\cc_h\cap\cd_h)).\label{eq:sum_nup}
\end{align}
By Bennett's inequality (see Lemma~\ref{lem:bennett} in Appendix~\ref{appd:technical} for the detailed statement), we have
\begin{align}
&\P(\cc_k^c \cond v\in V_d,\Nup,S,\cap_{h=1}^{k-1}(\cc_h\cap\cd_h))\nonumber\\
&\le \P\left(\bin\left(n,\frac{\lambda q}{n}\cdot 4\log n(\lambda q)^{k-1}\prod_{h=0}^{k-1}(1+(\lambda q)^{-h/2})\right)\ge 4\log n(\lambda q)^k\prod_{h=0}^k(1+(\lambda q)^{-h/2})\right)\nonumber\\
&\le \exp\left(-4\log n(\lambda q)^{k}\prod_{h=0}^{k-1}(1+(\lambda q)^{-h/2})\left(1-\frac{4\log n(\lambda q)^k\prod_{h=0}^{k-1}(1+(\lambda q)^{-h/2})}{n}\right)\cdot\phi((\lambda q)^{-k/2})\right),\label{eq:apply_bennett}
\end{align}
where $\phi(x):=(1+x)\log(1+x)-x$. Notice that function $\phi(x)$ satisfies $\phi(x)\ge\frac
{x^2}{3}$ for all $x\in [0,1]$. Moreover, we have $\prod_{h=0}^{k-1}(1+(\lambda q)^{-h/2})\ge 2$ for each $k\ge 1$ and $\prod_{h=0}^{\infty}(1+(\lambda q)^{-h/2})<\infty$ as $\lambda q>1$. These further implies 
\begin{equation*}
    \prod_{h=0}^{k-1}(1+(\lambda q)^{-h/2})\left(1-\frac{4\log n(\lambda q)^k\prod_{h=0}^{k-1}(1+(\lambda q)^{-h/2})}{n}\right)>1
\end{equation*}
for all large enough $n$, since $4\log n (\lambda q)^k=o(n)$. It then follows from~\eqref{eq:apply_bennett} that
\begin{align*}
    &\P(\cc_k^c \cond v\in V_d,\Nup,S,\cap_{h=1}^{k-1}(\cc_h\cap\cd_h))\\
    &\le \exp\left(-4\log n(\lambda q)^{k}\frac{(\lambda q)^k}{3}\right)\\
    &\le n^{-4/3.}
\end{align*}

To bound the probability of $\cd_k^c$, we take a union bound over the vertices in $S_\ch(v,k-1)$ and the vertices in $S$ to obtain
\begin{align*}
    &\P(\cd_k^c\cond v\in V_d,\Nup,S,\cap_{h=1}^{k-1}(\cc_h\cap\cd_h))\\
    &\le \frac{\lambda q}{n}\cdot n^{1-\eps/2}\cdot 4\log n(\lambda q)^{k-1}\prod_{h=0}^{k-1}(1+(\lambda q)^{-h/2})\\
    &=O(n^{-3\eps/8}),
\end{align*}
where the last inequality follows because $k=o(\log n)$ and $\prod_{h=0}^{\infty}(1+(\lambda q)^{-h/2})\le \infty$. By~\eqref{eq:sum_nup}, we have
\[
\P(V_{\ch}\cap S\neq \emp\cond v\in V_d,\Nup,S)\le (\dmax-d+l)\cdot O(n^{-3\eps/8})=O(n^{-\eps/3}).
\]
This further implies that
\begin{align}
    \P(V_{\ch}\cap S\neq \emp\cond v\in V_d)=O(n^{-\eps/3}).\label{eq:vch_nonempty}
\end{align}
Finally, putting~\eqref{eq:not_tree} and~\eqref{eq:vch_nonempty} together yields 
\[
\P(\ch=\ct^\rmic_{v,\dmax-d+l}\cond v\in V_d)=1-O(n^{\eps/3}),
\]
and completes the proof of~\eqref{eq:condition3}.

\subsection{Step 3: Proof of Proposition~\ref{prop:inclusion} at small depths}
In this section, we prove~\eqref{eq:inclusion} in the case of $d\in [\dmax-l^+-1]$ and~\eqref{eq:inclusion_root_prop}. In particular, we want to show
\begin{equation}
\label{eq:include_small}
\P(\td{v}\in\cm_v\cond v\in V_{d})\ge p_{l^+}-O(n^{-\eps/3})
\end{equation}
for each $d\in [\dmax-l^+-1]$ and $v\in\{2,\ldots,n\}$,
and 
\begin{equation}
\label{eq:include_root}
\P(\td{1}\in\cm_1)\ge p_{l^+}-O(n^{-\eps/3}).
\end{equation}
\paragraph{Proof of~\eqref{eq:include_small}} 
Intuitively, the upward pass phase is monotone in the sense that moving to higher levels of the diffusion tree can only increase the amount of matching evidence available. A vertex at a higher level has access not only to its own local structural information, but also to the aggregated evidence propagated from all of its descendants. As a result, the probability of correctly matching a vertex does not decrease as one moves upward in the tree. Following this argument,~\eqref{eq:include_small} simply follows from~\eqref{eq:fraction_large_d}. In the following, we make this argument formal using a truncation argument for Algorithm~\ref{alg:mpmatch}.

For the propose of proof, we introduce a truncated version of Algorithm~\ref{alg:mpmatch}, which only include the upward pass phase of Algorithm~\ref{alg:mpmatch}. Moreover, instead of starting at depth $\dmax$ of the tree $\tic$, the algorithm starts at depth $d+l^+<\dmax$, and only attempt to match the vertices within depth $d+l^+$. We provide the pseudocode of the truncated algorithm in Algorithm~\ref{alg:truncate}. To differentiate from the notation in Algorithm~\ref{alg:mpmatch}, we denote the matching candidate set in Algorithm~\ref{alg:truncate} by $\cm_u''$ for each $u\in V_{\le d+l^+}$.

\begin{algorithm2e}[t]
\caption{Truncated version of Algorithm~\ref{alg:mpmatch}} 
\label{alg:truncate}
\SetAlgoNoEnd
\DontPrintSemicolon
\SetKwInOut{Input}{Input}
\SetKwInOut{Output}{Output}
\Input{Diffusion tree $\tic$, fully observed graph $\td{\cg}_2$, model parameters $\lambda,s,q$, integers $d_{\max},l$}
\Output{A matching candidate set $\cm''_u$ for each $u$ within depth $d+l^+$ of $\tic$ }
 $\cm_u'' \gets \emptyset$ for each $u\in V_{\le d+l^+}$\\
\For{$k=d+l^+:-1:0$}{
\For{$u\in V_{\tic,k}$}{
\For{$\td{v}\in V_{\td{\cg}_2}$}{
\If{$u$ and $\td{v}$ satisfy at least one of Criteria 1--3}{
 $\cm_u''\gets\cm_u''\cup\{\td{v}\}$
}
}
}
}

 \Return{$\cm_v$, for all $v\in V_{\le d+l^+}$}
\end{algorithm2e}

It follows by induction for any $u\in V_{\le d+l^+}$ and any $\td{v}\in V_{\tdcg_2}$, Algorithm~\ref{alg:truncate} adds $\td{v}$ to the set $\cm_{u}''$ only if Algorithm~\ref{alg:mpmatch} adds $\td{v}$ to $\cm_{u}$. 
First consider vertices at depth $d+l^+$. In Algorithm~\ref{alg:truncate}, a vertex $\td{v}$ is added to $\cm_u$ only if they satisfy Criterion 1, since no deeper candidate sets are available. In Algorithm~\ref{alg:mpmatch}, the same pair  $\td{v}$ is added to $\cm_u$ if any of Criteria 1--3 holds. Therefore, at depth $d+l^+$, we have $\cm_{u}''\subseteq\cm_u$.

Now assume inductively that for some $k\in\{0,\ldots,d+l^+-1\}$, we have $\cm_w''\subseteq\cm_w$ for all vertices  $w$ at depths $k+1,\ldots,d+l^+$. Consider a vertex $u$ at depth $k$. In Algorithm~\ref{alg:truncate}, a vertex $\td{v}$ is added to $\cm_{u}''$ if they satisfy any of Criteria 1--3. By the induction hypothesis, whenever $\td{v}_1\in \cm''_{u_1}$, we also have $\td{v}_1\in \cm_{u_1}$. Hence any candidate $\td{v}$ that Algorithm~\ref{alg:truncate} adds to $\cm_u''$ will also satisfy the corresponding criterion when evaluated by Algorithm~\ref{alg:mpmatch}. This shows that $\cm_u''\subseteq\cm_u$ at depth $k$, completing the induction.

In other words, we have
\[
\P(\td{v}\in\cm_v\cond v\in V_{d})\ge \P(\td{v}\in\cm''_v\cond v\in V_{d}),
\]
and it suffices to show that
\begin{equation}
\label{eq:truncate}
\P(\td{v}\in\cm''_v\cond v\in V_{d})\ge p_{l^+}-O(n^{-\eps/3}).
\end{equation}

The rest of the proof is then analogous to the proof of~\eqref{eq:fraction_large_d} for the case of $d=\dmax-l^+$ in the previous section.
We define $\Nup:=V_{\le d}\setminus\{v\}$,
and use $\ch$
(resp. $\td{\ch}$) as the induced subgraph on the $(l+l^+)$-neighborhood of $v$ (resp. $\td{v}$) in the graph $\cg_1'\setminus\Nup$ (resp. $\tdcg_2\setminus\tdNup$), with $v$ (resp. $\td{v}$) as its root.
Then we have $\td{v}\in \cm_v''$ if the following three conditions hold:
\begin{enumerate}
    \item $\ch$ and $\td{\ch}$ are both trees, and they satisfy the property $P_{l^+}$.
    \item In graph $\tdcg_2$, there exist no edges between the set $\tdNup$ and any non-root vertex of $\td{\ch}$.
    \item $\ch=\ct^\mathrm{IC}_{v,l+l^+}$, where $\ct^{\mathrm{IC}}_{i,l+l^+}$ is the subtree in $\tic$ rooted at $i$ with depth up to $\dmax-d+l$.
\end{enumerate}

Using the same argument as in the proof of equations~\eqref{eq:condition1},~\eqref{eq:condition2} and~\eqref{eq:condition3} in the previous section, we have
\[
\P(\text{Condition 1 holds}\cond v\in V_d)\ge p_{l^+}-O(n^{-\eps/3}),
\]
\[
\P(\text{Condition 2 holds}\cond v\in V_d)\ge 1-O(n^{-\eps/3}),
\]
and 
\[
\P(\text{Condition 3 holds}\cond v\in V_d)\ge 1-O(n^{-\eps/3}),
\]
for each $d\in [\dmax-l^+-1]$.
These together imply the desired inequality~\eqref{eq:truncate}, and complete the proof.

\paragraph{Proof of~\eqref{eq:include_root}} For the root node $1$, we consider the truncated algorithm that starts its first iteration at depth $l^+$. In this case, the set $\Nup$ become empty, and the two rooted graphs $\ch$ and $\tdh$ are the $(l+l^+)$-neighborhood of vertex $1$ and vertex $\td{1}$ in the two graphs $\cg_1'$ and $\tdcg_2$ respectively. Then by the previous argument, we know that a sufficient condition for $\td{1}\in \cm_1$ is that $\ch$ and $\tdh$ are both trees, and they satisfy property $P_{l^+}$. By Lemmas~\ref{prop:auxiliary} and Proposition~\ref{prop:tree_tv}, this conditions holds with probability at least $p_{l^+}-O(n^{-\eps/3})$, which completes the proof of~\eqref{eq:include_root}.

\section{Coupling Correlated \erdos--\renyi Neighborhood and Correlated Galton--Watson trees}\label{appd:proof_tv}
In this section, we prove Proposition~\ref{prop:tree_tv}, which is restated as follows.
\begingroup
\proptreetv*
\endgroup

\begin{proof}
    By the assumption $\dmax-l^+\le d\le \dmax$, we have $\dmax-d+l\le l+l^+$. To prove the proposition, it suffices to find a coupling $\sigma$ between $\hat{\mathbb{P}}_{l+l^+}$ and $\mathbb{P}_{l+l^+}$ such that the desired property holds. As is clear from the context, we omit the subscript $l+l^+$, and simply write $\hat{\mathbb{P}}$ and $\mathbb{P}$ in the rest of this proof. We also use short-hand notation $l':=l+l^+$.

    To construct the coupling, we first introduce notations regarding two trees $(\ct,\tdct)\sim\mathbb{P}$. We denote the root of both trees by $r$. Recall the definition of the correlated tree distribution $\mathbb{P}$. Let $\ct^*$ denote the intersection from the process of generating $\ct$ and $\tdt$, and let $T^*$ denote the set of all vertices in this intersection tree. With a slight abuse of notation, for a vertex $v\in T^*$, we also use $v$ to denote the copies of this vertex in $\ct$ and $\tdt$. For each $k\in\{0,\ldots,l'\}$, we define $\ct_k$ as the truncated version of $\ct$ at depth $k$, i.e., $\ct_k$ is the subgraph induced by the vertices within depth $k$ of $\ct$. We similarly define $\tdt_k$ as the truncated version of $\ct$ at depth $k$. We use $T^*_k$ to denote the set of vertices at depth $k$ of $\ct^*$, use $T_k$ to denote the set of vertices at depth $k$ of $\ct$ that are added to $\ch^*$ in the generation process, and use $\td{T}_k$ to denote the set of vertices at depth $k$ of $\tdt$ that are added to $\ct^*$ in the generation process. In other words, $T^*_k$ represents the set of vertices shared by $\ct$ and $\tdt$ at depth $k$, $T_k$ represents the set of vertices exclusive to $\ct$ at depth $k$, and $\td{T}_k$ represents the set of vertices exclusive to $\tdt$ at depth $k$. We also define $\bar{T}_k=T_k\cup T^*_k\cup \td{T}_k$. For each vertex $v\in T_k$ (resp. $v\in \td{T}_k$), let $C_v$ (resp. $\td{C}_v$) denote the set of its children in $\ct$ (resp. $\tdt$), and let $c_v$ (resp. $\td{c}_v$) denote the cardinality of this set. By the definition of $\mathbb{P}$, we have $c_v\sim\poi(\lambda q)$ and $\td{c}_v\sim\poi(\lambda )$. For a vertex $v\in T^*_k$, let $C_v^*$ denote the set of its children in $\ct^*$, $C_v$ denote the set of its exclusive children in $\ct$ and $\td{C}_v$ denote the set of its exclusive children in $\tdt$. We use $c_v^*$, $c_v$ and $\td{c}_v$ to denote the cardinality of these three sets respectively. By definition, we get $c_v^*\sim\poi(\lambda sq)$, $c_v\sim\poi(\lambda q(1-s))$ and $\td{c}_v\sim\poi(\lambda (1-qs))$. Moreover, all of these random variables are mutually independent.

    Now we move on to the pair of rooted graphs $\ch$ and $\tdh$. Recall the roots of the two graphs are $u$ and $\td{u}$ respectively. For each $k\in\{0,\ldots,l'\}$, we define $\ch_k$ as the subgraph of $\ch$ induced by the vertices within distance $k$ from $i$, and $\td{\ch}_k$ as the subgraph of $\td{\ch}$ induced by the vertices within distance $k$ from $\td{i}$. For $k\in\{0,\ldots,d\}$, we iteratively define a sequence of sets $H^*_k$. First, we define $H^*_0=\{i\}$. Then for each $k\ge 1$, we define $H^*_k$ as the set of indices $j\in [n]$ such that:
    \begin{enumerate}
        \item $j\in S_\ch(i,k)$ and $\td{j}\in S_\tdh(\td{i},k)$;
        \item $\exists j'\in H^*_{k-1}: j\stackrel{\ch}{\sim}j'$ and $\td{j}\stackrel{\tdh}{\sim}\td{j}'$.
    \end{enumerate}
    In other words, $H^*_k$ is the set of indices that appear at the depth $k$ of the common structure of $\ch$ and $\tdh$. We also define $H_k=\{j\in [n]:j\in S_\ch(i,k)\text{ and }j\notin H^*_k\}$, $\td{H}_k=\{j\in [n]:\td{j}\in S_\tdh(\td{i},k)\text{ and }\td{j}\notin H^*_k\}$ and $\bar{H}_k=H_k\cup H^*_k\cup\td{H}_k$. 
    We comment that the definitions of the sets $H_k,H_k^*,\td{H}_k$ are analogous to the definition of the sets $T_k.T_k^*,\td{T}_k$ in the context of correlated Galton--Watson trees. But the main difference is that the sets $T_k$ and $\td{T_k}$ are  disjoint by the definition of the distribution $\mathbb{P}$, while this is not the case for $H_k$ and $\td{H}_k$. For example, it is possible that there exists vertex $i_1\in H_1$ and $i_2\in \td{H}_1$ such that $i_1$ is connected to a vertex $i$ in $\cg_1'$ and $\td{i}_2$ is connected to the vertex $\td{i}$ in $\tdcg_2$. In this case, $i$ is not in set $H^*_2$ as it is not a neighbor of any vertex in $H_1^*$. Therefore, it is in both sets $H_2$ and $\td{H}_2$. This type of intersection between sets $H_k$ and $\td{H}_k$ is undesired, as it will make later growth of the neighborhood correlated. We will careful rule out these events in the construction of the coupling.
    
    For each $j\in H_k$, let $D_j=\{j'\in H_{k+1}:j'\stackrel{\ch}{\sim}j\}$. For each $j\in \td{H}_k$, let $\td{D}_j=\{j'\in \td{H}_{k+1}:\td{j}'\stackrel{\tdh}{\sim}\td{j}\}$. We use $d_j$ and $\td{d}_j$ to denote the cardinality of $D_j$ and $\td{D}_j$ respectively. 
    For each $j\in H^*_k$, let $$D^*_j=\{j'\in H^*_{k+1}: j'\stackrel{\ch}{\sim}j\text{ and }\td{j}'\stackrel{\tdh}{\sim}\td{j}\},$$
    \[
    D_j=\{j'\in H_{k+1}:j'\stackrel{\ch}{\sim}j \},
    \]
    and 
    \[
    \td{D}_j=\{j'\in \td{H}_{k+1}:\td{j}'\stackrel{\tdh}{\sim}\td{j}\}.
    \]
    The set $D^*_j$ represents the set of indices of the common neighbors of $j$ and $\td{j}$ in at depth $k-1$, and $D_j$ (resp. $\td{D}_j$) represents the exclusive neighbors of $j$ (resp. $\td{j}$) in $\ch$ (resp. $\tdh$) at depth $k+1$. We use $d^*_j$, $d_j$ and $\td{d}_j$ to denote the cardinality of these three sets respectively.

    We further define a few sequences of events based on which we iteratively provide sufficient conditions for the desired properties $\ch_k=\ct_k$ and $\td{\ch}_k=\td{\ct}_k$.
    For each $k\in[l']$, let 
\[
\ca_k=\{\text{there exists two distinct }j,j'\in \bar{H}_k\text{ s.t. }\bar{j}\stackrel{\bar{\cg}}{\sim} \bar{j}'\}.
\]
where $\bar{\cg}$ is the union graph of $\cg_1'$ and $\cg_2$. 
We also define 
\begin{align*}
\cb_k=&\{\text{there exists two distinct indices }j,j'\in\bar{H}_{k-1}\text{ and an index }j''\in[n]\setminus (\Nup\cup (\cup_{h=0}^{k-1}\bar{H}_h))\\
&\text{ s.t. } \bar{j}\stackrel{\bar{\cg}}{\sim} \bar{j}''\text{ and } \bar{j}'\stackrel{\bar{\cg}}{\sim} \bar{j}''\},
\end{align*}
\[
\cc_{k}=\{\text{there exists }j\in H_{k-1}\cup H^*_{k-1}\text{ and }j'\in \cup_{h=0}^{k-1}\td{H}_h\text{ s.t. }j\stackrel{\cg_1'}{\sim}j'\},
\]
\[
\cd_{k}=\{\text{there exists }j\in \td{H}_{k-1}\cup H^*_{k-1}\text{ and }j'\in \cup_{h=0}^{k-1}H_h\text{ s.t. }\td{j}\stackrel{\td{\cg}_2}{\sim}\td{j}'\}.
\]
We illustrate these four events in Figure~\ref{fig:events}.

\begin{figure}[htbp]
    \centering
    \includegraphics[width=0.5\linewidth]{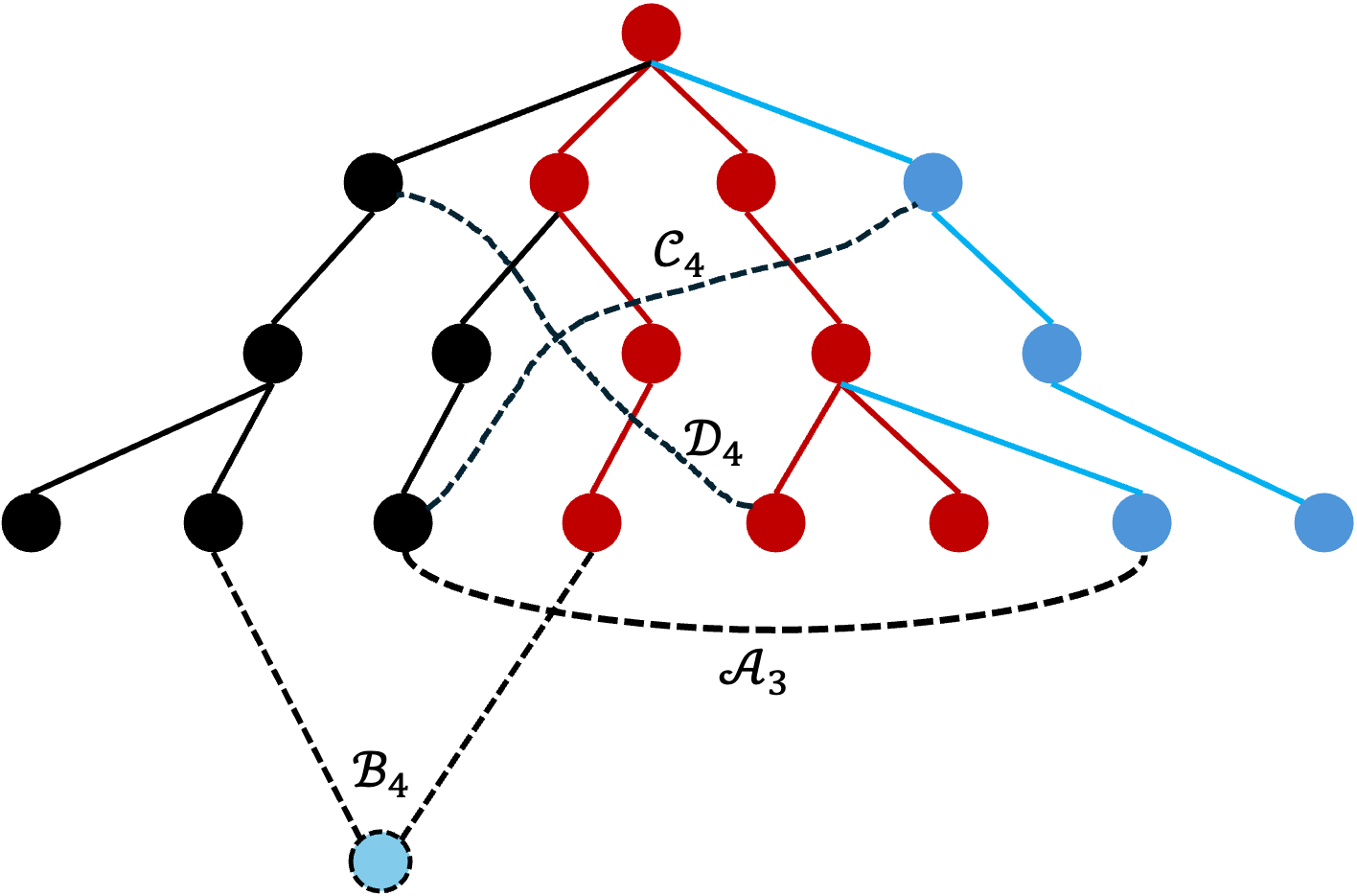}
    \caption{An illustration of the events we defined. In the figure, the black (resp. red and blue) vertices are from the set $\cup_{h=0}^3 H_h$ (resp. $\cup_{h=0}^3 H^*_h$ and $\cup_{h=0}^3 \td{H}_h$). In other words, the red nodes and edges represent the intersection of $\ch_{3}$ and $\td{\ch}_3$. The black (resp. blue) nodes and edges represent those exclusive to $\ch_3$ (resp. $\td{\ch}_3$). The existence of the dashed edges corresponds to the events annotated next to them.}
    \label{fig:events}
\end{figure}

We also define a sequence of events that guarantees the desired isomorphism $\ch_k\cong \ct_k$ and $\tdh_k\cong \tdt_k$. It also required addition regularity of the underlying permutation for the isomorphism and exclude undesired overlaps between $\ch$ and $\td{\ch}$, which are helpful for construct couplings for further depths of the trees.
For each $k\in\{0,\ldots,l'\}$, let $\ce_k$ denote the event that
\begin{enumerate}
    \item $\ch_k\cong \ct_k$ and $\tdh_k\cong \tdt_k$ under bijections $\phi_{k}:V_{\ch_k}\rightarrow V_{\ct_k}$ and $\td{\phi}_k:V_{\tdh_k}\rightarrow V_{\tdt_k}$ respectively;
    \item  The two bijections $\phi_{k}$ and $\td{\phi}_{k}$ satisfy $\phi_{k}(\cup_{h=0}^{k}H^*_h)=\cup_{h=0}^{k}T^*_h$ and $\phi_{k}(j)=\td{\phi}_{k}(\td{j})$ for every $j\in \cup_{h=0}^{k}H^*_h$;
    \item $(\cup_{h=0}^{k}H_h)\cap (\cup_{h=0}^{k}\td{H}_h)=\emptyset$.
\end{enumerate}
Since we want to find a coupling under which $\P(\ch\cong \ct,\tdh\cong \tdt)\ge 1-\td{O}(n^{-3\gamma/4})$. It suffices to lower bound the probability of the event $\cap_{k=0}^d\ce_k$ under the coupling. The following lemma utilizes the events $\ca_k,\cb_k,\cc_k,\cd_k$ to provide sufficient conditions for the desired event $\ce_k$.
\begin{lemma}
\label{lem:Ek_recursive}
    Let $k\in [l']$. If 
    \begin{enumerate}
        \item events $\ca_k^c$, $\cb_k^c$, $\cc_k^c$, $\cd_k^c$ and $\ce_{k-1}$ all hold,
        \item $d_j=c_{\phi_{k-1}(j)}$ for every $j\in H_{k-1}$,
        \item $\td{d}_{j}=\td{c}_{\td{\phi}_{k-1}(\td{j})}$ for every $j\in \td{H}_{k-1}$,
        \item $d_j=c_{\phi_{k-1}(j)}$, $\td{d}_j=\td{c}_{\phi_{k-1}(j)}$ and $d_j^*=c_{\phi_{k-1}(j)}^*$ for every $j\in H^*_{k-1}$,
    \end{enumerate}
    then event $\ce_k$ holds.
\end{lemma}
\begin{proof}[Proof of Lemma~\ref{lem:Ek_recursive}]
    Notice that under events $\ce_{k-1}$, $\ca_k^c$ and $\cb_k^c$, both $\ch_{k}$ and $\td{\ch}_k$ are trees. This is because $\ce_{k-1}$ guarantees that $\ch_{k-1}$ and $\td{\ch}_{k-1}$ are both trees, $\ca_k^c$ guarantees that there exists no edges between the depth-$k$ vertices of $\ch$ and $\tdh$, and $\cb_k^c$ guarantees that there exists no depth-$d$ vertices of $\ch$ or $\tdh$ that has more than one neighbor at depth $k-1$.

    Moreover, event $\cc^c_{k}$ guarantees that $H_k\cap (\cup_{h=0}^{k-1}\td{H}_h)=\emptyset$, event $\cd^c_{k}$ guarantees that $\td{H}_k\cap (\cup_{h=0}^{k-1}H_h)=\emptyset$ and event $\cb^c_k$ guarantees that $H_k\cap \td{H}_k=\emptyset$. Putting these together with the assumption $(\cup_{h=0}^{k-1}H_h)\cap (\cup_{h=0}^{k-1}\td{H}_h)=\emptyset$ in event $\ce_{k-1}$, we get $(\cup_{h=0}^{k}H_h)\cap (\cup_{h=0}^{k}\td{H}_h)=\emptyset$.

    By the second and fourth assumptions in the lemma, for each $j\in H_{k-1}$, we can arbitrarily map the vertices in $D_j$ to the vertices in $C_{\phi_{k-1}(j)}$ through a bijection. For each $j\in H^*_{k-1}$, we can arbitrarily map the vertices in $D_j$ to the vertices in $C_{\phi_{k-1}(j)}$, and the vertices in $D^*_j$ to the vertices in $C^*_{\phi_{k-1}(j)}$ both through bijective mappings. In this way, we can extend the bijection $\phi_{k-1}$ to a new bijection $\phi_{k}:V_{\ch_{k}}\rightarrow V_{\ct_k}$. It is easy to check that $\ch_k\cong \ct_k$ under $\phi_k$, and $\phi_{k}(\cup_{h=0}^{k}H^*_h)=\cup_{h=0}^{k}T^*_h$ is satisfied. By the third and fourth assumptions of the lemma, we can perform a similar extension on the mapping $\td{\phi}_{k-1}$ to a bijection $\td{\phi}_{k}:V_{\tdh_{k}}\rightarrow V_{\tdt_k}$, which satisfies that $\tdh_k\cong\tdt_k$ under this bijection. Moreover, this mapping can be set to satisfy that $\phi_{k}(j)=\td{\phi}_{k}(\td{j})$ for every $j\in \cup_{h=0}^{k}H^*_h$.
\end{proof}

For simplicity, we denote the events described in the second, third, and fourth assumptions in Lemma~\ref{lem:Ek_recursive} as $\cf_k$, $\td{\cf}_k$ and $\cf^*_k$ respectively.
    Notice that when $k=0$, $\ch_k$, $\tdh_k$, $\ct_k$, $\tdt_k$ are all trivial trees with only a root node and no edges. So $\ce_0$ is a trivial event that always holds. Then Lemma~\ref{lem:Ek_recursive} implies that
    \[
    \P(\cap_{k=0}^{l'} \ce_k)\ge \P(\cap_{k=1}^{l'}(\ca^c_k\cap\cb_k^c\cap\cc_k^c\cap\cd_k^c\cap\cf_k\cap\td{\cf}_k\cap\cf^*_k)).
    \]
    To bound this probability, we additionally introduce a sequence of events that controls the number of vertices at each level of the two graphs $\ch$ and $\td{\ch}$. For each $k\in[d]]$, let 
    \[
    \ci_k=\{|\bar{H}_k|\le (2\bl)^k\log n\},
    \]
    and we will instead bound $\P(\cap_{k=1}^{l'}\cj_k)$,
    where $\cj_k:=\ca^c_k\cap\cb_k^c\cap\cc_k^c\cap\cd_k^c\cap\cf_k\cap\td{\cf}_k\cap\cf^*_k\cap\ci_k$.
    By the chain rule, we have
    \begin{align}
        \P(\cap_{k=1}^{l'}\cj_k)
        &\ge \P(\cj_1)-\sum_{k=2}^{l'}\P(\cj_k^c|\cap_{h=1}^{k-1} \cj_h).\label{eq:Ed_chain}
    \end{align}
    
    So it suffices to bound $\P(\cj_1)$ and each term in the summation. 
    
    \emph{\textbf{Bounding $\P(\cj_1)$.}} Notice that 
    $\cb_1^c,\cc_1^c,\cd_1^c,\cf_1,\td{\cf}_1$ are all trivial events, since $\ch_0$ and $\td{\ch}_0$ are both trivial trees that only has a root. 
    The vertices in $\bar{H}_1$ are the neighbors of $u$ in the union graph $\bar{\cg}\setminus\bar{N}^\uparrow$, and therefore $|\bar{H}_1|\sim \bin(n-1-|\Nup|,\bl/n)$, where $\bl=\lambda+\lambda s-\lambda sq$. By the Chernoff--Hoeffding bound (see Lemma~\ref{lem:CB}) we have
    \begin{align*}
        \P(\ci_1^c)&\le \P(\bin(n,\bl/n)\ge 2\bl\log n)\le n^{-\bl}.
    \end{align*}
    Given event $\ci_1$, there are at most $4\bar{\lambda}\log^2 n$ vertex pairs within $\bar{H}_1$. The probability of having an edge between each pair is $O(n^{-1})$. Therefore, 
    \begin{equation}
        \P(\ca_1\cond \ci_1)\le 4\bar{\lambda}\log^2 n\cdot O(n^{-1})=\td{O}(n^{-1}). 
    \end{equation}
    To bound the probability of event $\cf_1^c$, we need to construct the coupling between $(d_u,\td{d}_u,d^*_u)$ and $(c_r,\td{c}_r,c^*_r)$.

    From the definition of correlated \erdos--\renyi graph pair model, we have
    $$(d_u,\td{d}_u,d^*_u,n-1-|\Nup|-d_u-\td{d}_u-d^*_u)\sim\mathrm{Multi}(n-1-|\Nup|,\lambda q(1-s)/n,\lambda(1-qs),\lambda qs/n,1-\bl/n).$$
    Meanwhile, the number of children of the root in $\ct$ and $\tdct$ has distribution
    \[
    (c_r,\td{c}_r,c^*_r)\sim \poi(\lambda q(1-s))\otimes\poi(\lambda(1-qs))\otimes\poi(\lambda sq).
    \]
    An upper bound for the total variation distance between multinomial distributions and independent Poisson distributions is provided in Lemma~\ref{lem:multi-poi} in Appendix~\ref{appd:technical}.
    By Lemma~\ref{lem:multi-poi}, we can couple $(d_u,\td{d}_u,d^*_u)$ with $(c_r,\td{c}_r,c^*_r)$ such that 
    \[
    \P((d_u,\td{d}_u,d^*_u)=(c_r,\td{c}_r,c^*_r))\ge 1-O\left(\frac{|\Nup|+\log^2 n}{n}\right)\ge 1-O(n^{-\gamma}).
    \]
Under this coupling, we have $\P(\cf^{*c}_1)\le O(n^{-\gamma})$. Putting these together yields
    \begin{align}
        \P(\cj_1)&\ge 1-\P(\ca_1\cup\cf^{*c}_1\cup\ci_1^c)\nonumber\\
        &\ge 1-\P(\ca_1\cup\ci_1^c)-\P(\cf^{*c}_1)\nonumber\\
        &\ge 1-\P(\ca_1|\ci_1)-\P(\ci_1^c)-\P(\cf^{*c}_1)\nonumber\\
        &\ge 1-O(n^{-\gamma}).\label{eq:J1}
    \end{align}

    \emph{\textbf{Bounding $\P(\cj_k^c|\cap_{h=1}^{k-1} \cj_h)$.}} Notice that event $\cap_{h=1}^{k-1} \cj_h$ implies that $|\bar{H}_h|\le (2\bl)^h\log n$ for each $h\in [k-1]$. Because $k\le l'=o(\log n)$, we have $\sum_{h=1}^{k-1}|H_h|\le \sum_{h=1}^{k-1} (2\bl)^h\log n=O(n^{\gamma/4})$. This further implies $|H_{k-1}|=O(n^{\gamma/4})$, $|\td{H}_{k-1}|=O(n^{\gamma/4})$, $|H^*_{k-1}|=O(n^{\gamma/4})$, $|\cup_{h=0}^{k-1}\td{H}_h|=O(n^{\gamma/4})$ and $|\cup_{h=0}^{k-1}H_h|=O(n^{\gamma/4})$, since all these sets are subsets of $\cup_{h=1}^{k-1}H_h$. Then by union bound, we have
    \begin{equation}
    \label{eq:treetv_ck}
    \P(\cc_k|\cap_{h=1}^{k-1} \cj_h)\le \frac{\bl}{n}\cdot(|H_{k-1}|+|H_{k-1}^*|)\cdot |\cup_{h=0}^{k-1}\td{H}_{h}|=O(n^{\gamma/2-1}),
    \end{equation}
    and 
    \begin{equation}
    \label{eq:treetv_dk}
    \P(\cd_k|\cap_{h=1}^{k-1} \cj_h)\le \frac{\bl}{n}\cdot(|\td{H}_{k-1}|+|H_{k-1}^*|)\cdot |\cup_{h=0}^{k-1}H_{h}|=O(n^{\gamma/2-1}).
    \end{equation}

    Events $\cc_k^c$, $\cd_k^c$ and $\cap_{h=1}^{k-1} \cj_h$ rule out the possible ability that any $j\in H^*_{k-1}$ is connected to any vertices include in $\cup_{h=0}^{k-1}\bar{H}_h$ when growing its neighborhood. Therefore, we have
    \begin{align*}
    &\left(d_j,\td{d}_j,d^*_j,n-|\Nup|-\sum_{h=0}^{k-1}|H_h|-d_j-\td{d}_j-d^*_j\right)\\
    &\sim\mathrm{Multi}\left(n-|\Nup|-\sum_{h=0}^{k-1}|H_h|,\lambda q(1-s)/n,\lambda(1-qs),\lambda qs/n,1-\bl/n\right).
    \end{align*}
    Then by Lemma~\ref{lem:multi-poi}, we can couple $(d_j,\td{d}_j,d^*_j)$ with $(c_{\phi_{k-1}(j)},\td{c}_{\phi_{k-1}(j)},c^*_{\phi_{k-1}(j)})$ such that
    \[
    \P((d_j,\td{d}_j,d^*_j)=(c_{\phi_{k-1}(j)},\td{c}_{\phi_{k-1}(j)},c^*_{\phi_{k-1}(j)}))\ge 1-O\left(\frac{|n-2n^{1-\gamma}|}{n}\right)\ge 1-O(n^{-\gamma}),
    \]
    where the penultimate inequality follows by the assumption that $|\Nup|\le n^{1-\gamma}$. Because $|H^*_{k-1}|=O(n^{-\gamma/4})$, we can use union bound to get 
    \begin{equation}
    \label{eq:treetv_fk*}
    \P(\cf^{*c}_k|(\cap_{h=1}^{k-1} \cj_h)\cap\cc^c_k\cap\cd^c_k )=O(n^{-3\gamma/4}).
    \end{equation}
    For each $j\in H_{k-1}$, we have
    \[
    d_j\sim\bin\left(n-1-|\Nup|-\sum_{h=1}^{k-1}|H_h|,\lambda q(1-s)\right).
    \]
    Recall that $c_{\phi_{k-1}(j)}\sim \poi(\lambda q(1-s))$.
    In Lemma~\ref{lem:binom-poi-tv} of Appendix~\ref{appd:technical}, we provide an upper bound of the total variation distance between binomial and Poisson random variables. By Lemma~\ref{lem:binom-poi-tv} 
    and the assumption that $|\Nup|\le n^{1-\gamma}$ and $\sum_{h=1}^{k-1}|H_h|=O(n^{\gamma/4})$, we can couple $d_j$ and $c_{\phi_{k-1}(j)}$ so that
    \[
    \P(d_j=c_{\phi_{k-1}(j)})\ge 1-O\left(\frac{n^{1-\gamma}}{n}\right)\ge 1-O(n^{-\gamma}),
    \]
    and it follows that
    \begin{equation}
    \label{eq:treetv_fk}
    \P(\cf^{c}_k|(\cap_{h=1}^{k-1} \cj_h)\cap\cc^c_k\cap\cd^c_k )=O(n^{-3\gamma/4}).
    \end{equation}
    Similarly, we can bound 
    \begin{equation}
    \label{eq:treetv_tdfk}
    \P(\td{\cf}^{c}_k|(\cap_{h=1}^{k-1} \cj_h)\cap\cc^c_k\cap\cd^c_k )=O(n^{-3\gamma/4}).
   \end{equation}

    Given $\cc_k^c$, $\cd_k^c$ and $\cap_{h=1}^{k-1} \cj_h$, we can apply union bound over pairs $j,j'\in \bar{H}_{k-1}$ and $j''\in [n]\setminus S\setminus (\cup_{h=0}^{k-1}\bar{H}_h)$ to get 
    \begin{equation}
    \label{eq:treetv_bk}
    \P(\cb_k|(\cap_{h=1}^{k-1} \cj_h)\cap\cc^c_k\cap\cd^c_k )\le O(n\cdot n^{\gamma/2}\cdot n^{-2})=O(n^{-1+\gamma/2}).
    \end{equation}
    Also, we know that 
    \[
    |\bar{H}_k|\sto\bin(n-|\Nup|-|\cup_{h=0}^{k-1}\bar{H}_h|,1-(1-\bl/n)^{(2\bl)^{k-1}\log n}).
    \]
    Therefore, by the Chernoff--Hoeffding bound,
    \begin{align}
        \P(\ci^c_k|(\cap_{h=1}^{k-1} \cj_h)\cap\cc^c_k\cap\cd^c_k )&\le \P\left(\bin\left(n,\frac{\bl(2\bl)^{k-1}\log n}{n}\right)\ge (2\bl)^k\log n\right)\nonumber\\
        &\le n^{-2\bl^2/3}.\label{eq:treetv_ik}
    \end{align}
    Given events $\ci_k$, $\cc_k^c$, $\cd_k^c$ and $\cap_{h=1}^{k-1} \cj_h$, we can use union bound to get 
    \begin{equation}
    \label{eq:treetv_ak}
    \P(\ca_k|(\cap_{h=1}^{k-1} \cj_h)\cap\cc^c_k\cap\cd^c_k\cap\ci_k)\le \frac{\bl}{n}\cdot (2\bl)^{2k}\log^2 n\le n^{\gamma-1},
    \end{equation}
    where the last inequality follows because $k=o(\log n)$. By equations~\eqref{eq:treetv_ck}-\eqref{eq:treetv_ak}, we have
    \begin{align}
        &\P(\cj_k^c|\cap_{h=1}^{k-1} \cj_h)\nonumber\\
        &=\P(\ca_k\cup\cb_k\cup\cc_k\cup\cd_k\cup\cf_k^c\cup\td{\cf}^c_k\cup\cf^{*c}_k\cup\ci^c_k|\cap_{h=1}^{k-1} \cj_h)\nonumber\\
        &\le \P(\cc_k|\cap_{h=1}^{k-1} \cj_h)+\P(\cd_k|\cap_{h=1}^{k-1} \cj_h)+\P(\cf^{*c}_k|(\cap_{h=1}^{k-1} \cj_h)\cap\cc^c_k\cap\cd^c_k )+\P(\cf^{c}_k|(\cap_{h=1}^{k-1} \cj_h)\cap\cc^c_k\cap\cd^c_k )\nonumber\\
        &+\P(\td{\cf}^{c}_k|(\cap_{h=1}^{k-1} \cj_h)\cap\cc^c_k\cap\cd^c_k )+\P(\cb_k|(\cap_{h=1}^{k-1} \cj_h)\cap\cc^c_k\cap\cd^c_k )+\P(\ci^c_k|(\cap_{h=1}^{k-1} \cj_h)\cap\cc^c_k\cap\cd^c_k )\nonumber\\
        &+\P(\ca_k|(\cap_{h=1}^{k-1} \cj_h)\cap\cc^c_k\cap\cd^c_k\cap\ci_k)\nonumber\\
        &=O(n^{-3\gamma/4}).\label{eq:diff}
    \end{align}
    Finally, substituting~\eqref{eq:J1} and~\eqref{eq:diff} into~\eqref{eq:Ed_chain} gives 
    \begin{align*}
        \P(\cap_{k=0}^{l'} \ce_k)&\ge 1-O(n^{-\gamma})-O(l'\cdot n^{-3\gamma/4})\ge 1-\td{O}(n^{-3\gamma/4}),
    \end{align*}
    which completes the proof.
\end{proof}

\section{Properties of sparse \erdos--\renyi Random Graphs}\label{appd:properties}
In this section, we prove a few properties in sparse \erdos--\renyi random graphs that are used in the proof of Propositions~\ref{prop:noerror} and~\ref{prop:inclusion}.
\subsection{Neighborhood size bounds}
\begin{restatable}{lemma}{lemsizeuncond}
    \label{lem:neighborhood_size}
    Consider an \erdos--\renyi random graph $\cg\sim\mathrm{ER}(n,\frac{\mu}{n})$, where $\mu>1$ is a constant independent of $n$.
    Let $d=c\log n$ for some $c\le\frac{1-\gamma}{\log\mu}$, where $\gamma$ is an arbitrarily small constant and $C$ be a constant satisfying $C\ge 1$. For an arbitrary vertex $i$ in $\cg$, we have
    \begin{equation}
    \label{eq:neighbor_size}
    \P\left(\exists k\in\{0,\ldots,d\}:|N_{\cg}(i,k)|\ge K\mu^k\log n\right)\le \td{O}(n^{-C/3}),
    \end{equation}
     where constant $K:=\frac{C\mu}{\mu-1}\prod_{h=0}^\infty(1+\mu^{-h/2})$.
\end{restatable}
\begin{proof}
    In this proof, we omit the subscript $\cg$ in the notations $N_{\cg}(i,k)$ and $S_\cg(i,k)$ as is clear from the context.
    For each $k\in\{0,\ldots,d\}$, define event $$\ca_k=\left\{|S(i,k)|<C(\log n)\mu^k\prod_{h=0}^k(1+\mu^{-h/2})\right\}.$$ 
    Since 
    \begin{align*}
        \sum_{k'=0}^k C(\log n)\mu^{k'}\prod_{h=0}^{k'}(1+\mu^{-h/2})&\le C(\log n)\prod_{h=0}^{\infty}(1+\mu^{-h/2}) \sum_{k'=0}^k \mu^{k'}\\
        &\le K\mu^{k}\log n,
    \end{align*}
    it suffices to show that $\P(\cap_{k=0}^d\ca_k)\ge 1-\td{O}(n^{-C/3})$. We can rewrite this probability as 
    \begin{align}
       \P(\cap_{k=0}^d\ca_k)&=\P(\ca_0)-\sum_{k=1}^{d}\P(\ca_k^c\cap(\cap_{h=0}^{k-1}\ca_{h}))\nonumber\\
       &\ge \P(\ca_0)-\sum_{k=1}^{d}\P(\ca_k^c|\cap_{h=0}^{k-1}\ca_{h}).\label{eq:chain}
    \end{align}

    Notice that $\ca_0$ trivially holds so $\P(\ca_0)=1$, and it suffices to bound each term in the summation in~\eqref{eq:chain}.

    Fix some $k\in\{0,\ldots,d\}$. Notice that given $|N(i,k-1)|=m$ and $|S(i,k-1)|=m'$, we have
    \[
    \cl(|S(i,k)|)=\bin(n-m,1-(1-\mu/n)^{m'}).
    \]
    Then conditioned on event $\cap_{h=0}^{k-1}\ca_{h}$, we have
    \[
    \cl(|S(i,k)|)\stackrel{\mathrm{sto.}}{\le}\bin\left(n,\frac{C(\log n)\mu^{k}\prod_{h=0}^{k-1}(1+\mu^{-h/2})}{n}\right),
    \]
    where $\sto$ is the stochastic dominance operator. By Bennett's inequality, we have
    \begin{align*}
        &\P(\ca_k^c|\cap_{h=0}^{k-1}\ca_{h})\\
        &\le \P\left(\bin\left(n,\frac{C(\log n)\mu^{k}\prod_{h=0}^{k-1}(1+\mu^{-h/2})}{n}\right)\ge C(\log n)\mu^{k}\prod_{h=0}^{k}(1+\mu^{-h/2})\right)\\
        &\le \exp\left(-C(\log n)\mu^{k}\left(\prod_{h=0}^{k-1}(1+\mu^{-h/2})\right)\cdot\left(1-\frac{C(\log n)\mu^{k}\prod_{h=0}^{k-1}(1+\mu^{-h/2})}{n}\right)\cdot \phi(\mu^{-k/2})\right),
    \end{align*}
    where $\phi(x):=(1+x)\log(1+x)-x$. 
    The function $\phi$ satisfies $\phi(x)\ge \frac{x^2}{3}$ for all $x\in [0,1]$.
    Because $\mu>1$ is a constant, we have $\prod_{h=0}^{k-1}(1+\mu^{-h/2})\ge 2$ and $\prod_{h=0}^{\infty}(1+\mu^{-h/2})<\infty$. 
    Because $k\le d=c\log n$, we know that $C(\log n)\mu^{k}=o(n)$.
    These further imply that 
    \[
    \prod_{h=0}^{k-1}(1+\mu^{-h/2})\cdot\left(1-\frac{C(\log n)\mu^{k}\prod_{h=0}^{k-1}(1+\mu^{-h/2})}{n}\right)>1
    \]
    for all large enough $n$.
    Therefore, we have 
    \begin{align*}
        \P(\ca_k^c|\cap_{h=0}^{k-1}\ca_{h})&\le\exp\left(-\frac{C\log n}{3}\right)=n^{-C/3}.
    \end{align*}
    Notice that the above inequality holds for every $k\in [d]$. Therefore, we have 
    \[
    \P(\cap_{k=0}^d\ca_k)\ge 1-dn^{-C/3}=1-\td{O}(n^{-C/3}),
    \]
    which completes the proof.
\end{proof}

\begin{restatable}{lemma}{lemsizecond}
    \label{lem:neighborhood_size_cond}
    Consider an \erdos--\renyi random graph $\cg\sim\mathrm{ER}(n,\frac{\mu}{n})$, where $\mu>1$ is a constant independent of $n$.
    Let $d=c\log n$ for some $c\le\frac{1-\gamma}{\log\mu}$, where $\gamma$ is an arbitrarily small constant and $C$ be a constant satisfying $C\ge 1$. For two distinct vertices $i$ and $j$ in $\cg$, and an integer $d'\le d$, we have
    \begin{align}
    &\P\left(\exists k\in\{0,\ldots,d\}:|N_{\cg}(i,k)|\ge K\mu^k\log n\;\bigg|\;j\in S_\cg(i,d')\right)
    \le O(n^{-\gamma/3}),\label{eq:neighbor_size_cond}
    \end{align}
     where constant $K:=\frac{C\mu}{\mu-1}\prod_{h=0}^\infty(1+\mu^{-h/2})(1+\mu^{-h})$.
\end{restatable}
\begin{proof}
    In this proof, we omit the subscript $\cg$ in the notations $N_{\cg}(i,k)$ and $S_\cg(i,k)$ as is clear from the context.
    Define event 
    \[
    \cb=\left\{\exists k\in\{0,\ldots,d\}:|N(i,k)|\ge \sum_{k'=0}^{k}C(\log n)\mu^{k'}\prod_{h=0}^{k'}(1+\mu^{-h/2})(1+\mu^{-h})\right\}.
    \]
    By the definition of constant $K$, we have
    \begin{align*}
        K\mu^k\log n\ge \sum_{k'=0}^{k}C(\log n)\mu^{k'}\prod_{h=0}^{k'}(1+\mu^{-h/2})(1+\mu^{-h}),
    \end{align*}
    for any $k$. Therefore, it suffices to show that $\P(\cb\cond j\in S(i,d'))\le O(n^{-\gamma/3})$. Notice that the conditioned event $j\in S(i,d')$ is equivalent as there exists a path of length $d'$ between $i$ and $j$ in graph $\cg$ \emph{and} there exists no path of length less than $d'$ between $i$ and $j$ in $\cg$. Let $\mathcal{P}^{d'}_{i,j}$ denote the collection of all possible path of length $d'$ between $i$ and $j$ in $\cg$. We define a total order for the paths in $\mathcal{P}^{d'}_{i,j}$ as follows: For two distinct paths $p=(v_0,v_1,\ldots,v_{d'})$ and $p'=(v_0',v_1',\ldots,v_{d'}')$ with $v_0=v_0'=i$ and $v_{d'}=v_{d'}'=j$. We say $p$ has a higher order than $p'$, denoted $p\prec p'$, if there exists an index $d''\in[d'-1]$ such that $v_k=v_k',\forall k< d''$ and $v_{d''}<v_{d''}'$. Then we can write
    \begin{align*}
        &\P(\cb\cond j\in S(i,d'))\\
        &=\sum_{p\in\mathcal{P}}\P(\cb,\text{the highest order $d'$-path between $i$ and $j$ in $\cg$ is $p$}\cond j\in S(i,d'))\\
        &\le \sum_{p\in\mathcal{P}}\P(\text{the highest order $d'$-path between $i$ and $j$ in $\cg$ is $p$}\cond j\in S(i,d'))\\
        &\quad\quad\quad\cdot\P(\cb\cond\text{the highest order $d'$-path between $i$ and $j$ in $\cg$ is $p$},j\in S(i,d'))\\
        &\le \max_{p\in\mathcal{P}}\P(\cb\cond\text{the highest order $d'$-path between $i$ and $j$ in $\cg$ is $p$},j\in S(i,d'))\\
        &=\max_{p\in\mathcal{P}}\P(\cb\cond \{\text{path $p$ exists in $\cg$}\}\cap \{\text{no path in $\mathcal{P}$ with higher order than $p$ exists in $\cg$}\}\\
        &\quad\quad\quad\quad\quad\quad \cap\{\text{no path with length less than $d'$ between $i$ and $j$ exists in $\cg$}\}).
    \end{align*}
    Notice that $\cb$ is an increasing graph property because if $\cb$ holds in $\cg$, then the property still holds after adding an arbitrary edge to $\cg$. On the other hand, both events  $$\{\text{no path in $\mathcal{P}$ with higher order than $p$ exists in $\cg$}\}$$ and $$\{\text{no path with length less than $d'$ between $i$ and $j$ exists in $\cg$}\}$$ are decreasing graph properties because if they hold in $\cg$, removing any edge from $\cg$ does not break the properties. Then by Harris inequality~\citep{harris1960lower}, we have
    \[
    \P(\cb\cond j\in S(i,d'))\le \max_{p\in\mathcal{P}}\P(\cb\cond \{\text{path $p$ exists in $\cg$}\}). 
    \]
    Fix an arbitrary $p\in \mathcal{P}$ with $p=(v_0,v_1,\ldots,v_{d'})$. For each $0\le k\le d$, define event
    \[
    \cb_k=\left\{|S(i,k)|\ge C(\log n)\mu^{k}\prod_{h=0}^{k}(1+\mu^{-h/2})(1+\mu^{-h})\right\}
    \]
    We also define
    \[
    \cc_k=\{S(i,k)\cap\{v_0,\ldots,v_{d'}\}=\{v_k\}\}
    \]
    for each $0\le k\le d'$, and
    \[
    \cc_k=\{S(i,k)\cap\{v_0,\ldots,v_{d'}\}=\emp\}
    \]
    for each $k>d'$.
    We then have
    \begin{align}
        &\P(\cb\cond \{\text{path $p$ exists in $\cg$}\})\nonumber\\
        &\le \P(\cup_{k=0}^d\cb_k\cond\{\text{path $p$ exists in $\cg$}\})\nonumber\\
        &\le \P(\cup_{k=0}^d(\cb_k\cup\cc_k^c)\cond\{\text{path $p$ exists in $\cg$}\})\nonumber\\
        &\le \P(\cb_0\cup\cc_0^c\cond \{\text{path $p$ exists in $\cg$}\})+\sum_{k=1}^{d}\P(\cb_k\cup\cc_k^c\cond\{\text{path $p$ exists in $\cg$}\}\cap(\cap_{h=0}^{k-1}\cb_h^c\cap\cc_h))\nonumber\\
        &=\sum_{k=1}^{d}\P(\cb_k\cup\cc_k^c\cond\{\text{path $p$ exists in $\cg$}\}\cap(\cap_{h=0}^{k-1}\cb_h^c\cap\cc_h)),\label{eq:B_cond}
    \end{align}
    where the equality holds because $\cb_0$ and $\cc_0^c$ are both trivial events that cannot happen. 

    Given event $\{\text{path $p$ exists in $\cg$}\}\cap(\cap_{h=0}^{k-1}\cb_h^c\cap\cc_h)$, we know that
    \[
    |S(i,k)| \sto 1+\bin\left(n,\frac{\mu}{n}\cdot  C(\log n)\mu^{k-1}\prod_{h=0}^{k-1}(1+\mu^{-h/2})(1+\mu^{-h})\right),
    \]
    where we add one to the right-hand side because when $k\le d'$, there is a known edge $(v_{k-1},v_k)$ on the given path $p$. By Bennett's inequality, we have
    \begin{align*}
        &\P\Bigg(\bin\left(n,\frac{\mu}{n}\cdot  C(\log n)\mu^{k-1}\prod_{h=0}^{k-1}(1+\mu^{-h/2})(1+\mu^{-h})\right)\\
        &\quad\quad\quad\quad\ge C(\log n)\mu^{k}\prod_{h=0}^k(1+\mu^{-h/2})\prod_{h'=0}^{k-1}(1+\mu^{-h'})\Bigg)\\
        &\le n^{-C/3}.
    \end{align*}
     Moreover, we have
    \[
    1+ C(\log n)\mu^{k}\prod_{h=0}^k(1+\mu^{-h/2})\prod_{h'=0}^{k-1}(1+\mu^{-h'})\le  C(\log n)\mu^{k}\prod_{h=0}^{k}(1+\mu^{-h/2})(1+\mu^{-h}).
    \]
    It then follows that
    \begin{equation}
    \label{eq:bk}
        \P(\cb_k\cond\{\text{path $p$ exists in $\cg$}\}\cap(\cap_{h=0}^{k-1}\cb_h^c\cap\cc_h))\le n^{-C/3}.
    \end{equation}

    On the other hand, given event $\{\text{path $p$ exists in $\cg$}\}\cap(\cap_{h=0}^{k-1}\cb_h^c\cap\cc_h)$, event $\cc^c_k$ essentially states that there exists an edge between $S(i,k-1)$ and $\{v_{k+1},\ldots,v_{d'}\}$. By the union bound over the vertex pair of these two sets, we can get
    \begin{align}
        &\P(\cc_k^c\cond\{\text{path $p$ exists in $\cg$}\}\cap(\cap_{h=0}^{k-1}\cb_h^c\cap\cc_h))\nonumber\\
        &\le \frac{\mu}{n}\cdot d\cdot C(\log n)\mu^{k-1}\prod_{h=0}^{k-1}(1+\mu^{-h/2})(1+\mu^{-h})\nonumber\\
        &\le O(n^{-\gamma/2}),\label{eq:ckc}
    \end{align}
    Equations~\eqref{eq:bk} and~\eqref{eq:ckc} together imply
    \begin{equation}
        \P(\cb_k\cup\cc_k^c\cond\{\text{path $p$ exists in $\cg$}\}\cap(\cap_{h=0}^{k-1}\cb_h^c\cap\cc_h))=O(n^{-\gamma/2}). 
    \end{equation}
    By~\eqref{eq:B_cond}, we finally get
    \[
    \P(\cb\cond \{\text{path $p$ exists in $\cg$}\})\le O(n^{-\gamma/3}).
    \]
    This holds for any path $p$, and completes the proof.
\end{proof}

    \subsection{Cycle-free neighborhoods conditioned on given structure}
    \begingroup
    \lemcyclefree*
    \endgroup
    \begin{proof}
        Throughout this proof, we omit the subscript $\cg$ in notations $S_\cg(i,k)$ and $N_{\cg}(i,k)$ as is clear from the context. We denote $l=\sqrt{\log n}$, and use $H$ to denote the set of vertices in $\ch$.
        For each $d\ge 0$, let $H_d$ denote the vertices in $\ch$ that are at distance $d$ from $i$. By the third assumption on $\ch$, we have $|H_d|\le C$ for any $d\ge 0$. For each $h\in [2l]$, we define events
        \begin{equation}
            \ca_h=\{\text{there exists an edge within $S(i,h)$}\}
        \end{equation}
        and
        \begin{equation}
            \cb_h=\{\text{there exists a vertex in $[n]\setminus N(i,h-1)$ that is connected to two vertices in $S(i,h-1)$}\}.
        \end{equation}
        It is easy to see that the event $\cap_{h=1}^{2l}(\ca_h^c\cap\cb_h^c)$ implies that the $2l$-neighborhood of $i$ is a tree.
        We further define two sequence of auxiliary events to control the neighborhood size of $i$ and its intersection with $H_h$: For each $h\in [2l]$, define
        \begin{equation}
        \label{eq:def_ch}
        \cc_{h}=\left\{|S(i,h)|\ge 4\log n\cdot\bar{\lambda}^h\prod_{t=0}^{h}(1+\mu^{-t/2})(1+\mu^{-t})\right\}
        \end{equation}
        and
        \begin{equation}
        \label{eq:def_dh}
        \cd_h=\{S(i,h)\cap H\neq H_h\}.
        \end{equation}
        We illustrate the defined events in Figure~\ref{fig:cond_neighbor}.

\begin{figure}[htbp]
    \centering
    \includegraphics[width=0.4\linewidth]{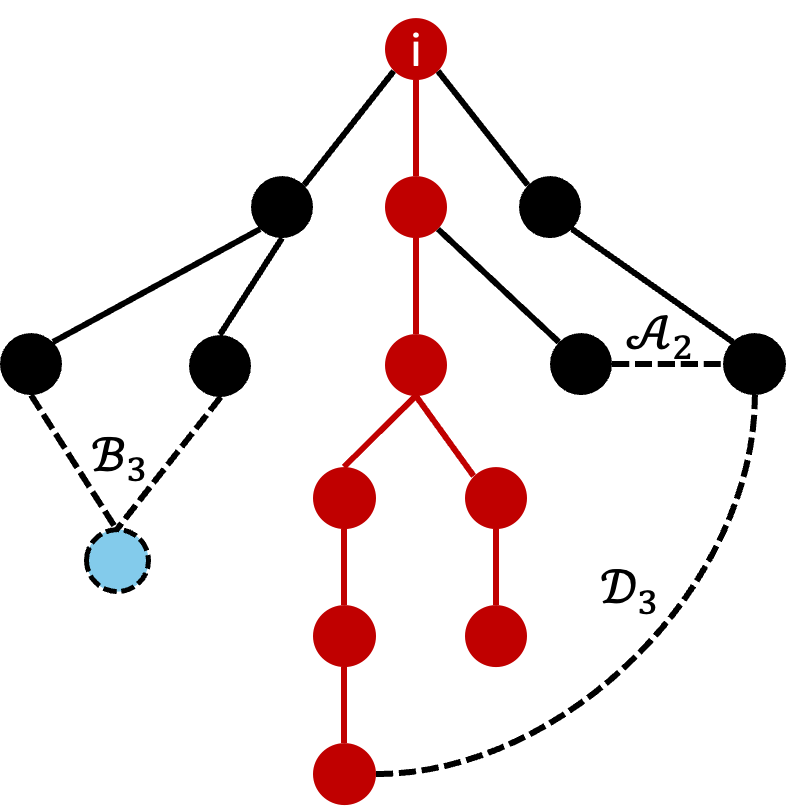}
    \caption{Illustration of events defined above. In the figure, the conditioned graph $\ch$ is represented by the red nodes and  edges. The black nodes and edges represent $i$'s neighborhood outside the conditioned graph. The existence of the dashed edges corresponds to the events annotated next to them.}
    \label{fig:cond_neighbor}
\end{figure}
        
        To prove the theorem, it suffices to show that 
        \begin{equation}
            \P(\cup_{h=1}^{2l}(\ca_h\cup\cb_h\cup\cc_h\cup\cd_h)\cond \ch\in\cg)=O(n^{-1+\gamma}).
        \end{equation}
        By the chain rule, we can write
        \begin{align}
            &\P(\cup_{h=1}^{2l}(\ca_h\cup\cb_h\cup\cc_h\cup\cd_h)\cond \ch\in\cg)\nonumber\\
            &\le \sum_{h=1}^{2l}\P(\ca_h\cup\cb_h\cup\cc_h\cup\cd_h\cond  \{\ch\subset\cg\}\cap (\cap_{h'=1}^{h-1}\ca^c_{h'}\cap\cb^c_{h'}\cap\cc^c_{h'}\cap\cd^c_{h'})),
        \end{align}
        and we will bound each term in this summation.

        First, notice that given event $ \{\ch\subset\cg\}\cap (\cap_{h'=1}^{h-1}\ca^c_{h'}\cap\cb^c_{h'}\cap\cc^c_{h'}\cap\cd^c_{h'})$, we know that
        \[
        |S(i,h)\cond\sto C+\bin\left(n,\frac{\mu}{n}\cdot 4\log n\cdot\mu^{h-1}\prod_{t=0}^{h-1}(1+\mu^{-t/2})\left(1+\mu^{-t}\right)\right).
        \]
        Here we add $C$ on the right-hand side because of there are at most edges from $H_{h-1}$ to $H_h$ in the given graph $\ch$. By Bennett's inequality, we have
        \begin{align*}
            &\P\Bigg(\bin\left(n,\frac{\mu}{n}\cdot 4\log n\cdot\mu^{h-1}\prod_{t=0}^{h-1}(1+\mu^{-t/2})\left(1+\mu^{-t}\right)\right)\\
            &\;\;\;\;\;\ge 4\log n\cdot\mu^{h}\prod_{t=0}^{h}(1+\mu^{-t/2})\prod_{t=0}^{h-1}\left(1+\mu^{-t}\right)\Bigg)\\
            &\le n^{-4/3}.
        \end{align*}
        Meanwhile, we have
        \[
        C\le 4\log n\prod_{t=0}^{h}(1+\mu^{-t/2})\prod_{t=0}^{h-1}\left(1+\mu^{-t}\right)
        \]
        for all large enough $n$. Therefore, we get 
        \begin{equation}
            \label{eq:ch}
            \P(\cc_h\cond\{\ch\subset\cg\}\cap (\cap_{h'=1}^{h-1}\ca^c_{h'}\cap\cb^c_{h'}\cap\cc^c_{h'}\cap\cd^c_{h'}))\le n^{-4/3}.
        \end{equation}
        Notice that event $\cd_h$ essentially means that there exists an edge between $S(i,h-1)$ and $H\setminus(\cup_{h'=0}^{h}H_{h'})$. We can apply a union bound over the vertex pairs in these two sets and get
        \begin{align}
        &\P(\cd_h\cond\{\ch\subset\cg\}\cap (\cap_{h'=1}^{h-1}\ca^c_{h'}\cap\cb^c_{h'}\cap\cc^c_{h'}\cap\cd^c_{h'}) )\nonumber\\
        &\le \frac{\mu}{n}\cdot K\log n\cdot 4\log n\cdot\mu^{h-1}\prod_{t=0}^{h-1}(1+\mu^{-t/2})(1+\mu^{-t})\nonumber\\
        &= O(n^{-1+\gamma/2}),\label{eq:dh}
    \end{align}
    where the last equality follows because $h=o(\log n)$, and therefore $\mu^h=O(n^{\gamma/8})$.
    For event $\cb_h$, we can again apply the union bound to get
    \begin{align}
        &\P(\cb_h\cond \{\ch\subset\cg\}\cap (\cap_{h'=1}^{h-1}\ca^c_{h'}\cap\cb^c_{h'}\cap\cc^c_{h'}\cap\cd^c_{h'}))\nonumber\\
        &\le \frac{\mu^2}{n^2}\cdot n\cdot  \left(4\log n\cdot\mu^{h-1}\prod_{t=0}^{h-1}(1+\mu^{-t/2})(1+\mu^{-t})\right)^2\nonumber\\
        &\;\;\;\;+\frac{\mu}{n}\cdot C\cdot 4\log n\cdot\mu^{h-1}\prod_{t=0}^{h-1}(1+\mu^{-t/2})(1+\mu^{-t})\nonumber\\
        &=O(n^{-1+\gamma/2}).\label{eq:bh}
    \end{align}
    When we are further given events $\cc_h^c$ and $\cd_h^c$, we know that the probability of having an edge between each vertex pair in $S(i,h)$ is $\mu/n$ and there are at most $\left(4\log n\cdot\mu^{h}\prod_{t=0}^{h}(1+\mu^{-t/2})(1+\mu^{-t})\right)^2$ vertex pairs. Then by union bound, we have
    \begin{align}
        &\P(\ca_h\cond  \{\ch\subset\cg\}\cap (\cap_{h'=1}^{h-1}\ca^c_{h'}\cap\cb^c_{h'}\cap\cc^c_{h'}\cap\cd^c_{h'})\cap \cc_h^c\cap\cd_h^c)\nonumber\\
        &\le \frac{\mu}{n}\cdot\left(4\log n\cdot\mu^{h}\prod_{t=0}^{h}(1+\mu^{-t/2})(1+\mu^{-t})\right)^2\nonumber\\
        &=O(n^{-1+\gamma/2}).\label{eq:ah}
    \end{align}
    Equations~\eqref{eq:ch}-\eqref{eq:ah} together yield that

    \begin{align*}
        &\P(\cup_{h=1}^{2l}(\ca_h\cup\cb_h\cup\cc_h\cup\cd_h)\cond \ch\in\cg)\\
        &\le 2l\cdot O(n^{-1+\gamma/2})\\
        &=O(n^{-1+\gamma}),
    \end{align*}
    which completes the proof.
    \end{proof}

\subsection{Probability of non-empty neighborhood}
    \nonemptyprob*
    \begin{proof}
        In the \erdos--\renyi random graph $\cg\sim\mathrm{ER}(n,\frac{\mu}{n})$, with high probability, there exists a unique giant component of size at least $Cn(1+o(1))$, where $c$ is a constant (see for example Chapter 8 in~\cite{blum2020foundations}). 
        It then follows symmetry that
        \begin{align*}
            &\P(\text{there exists a giant component of size $cn(1+o(1))$ in $\cg$ and vertex $i$ is in this component})\\
            &\ge C-o(1).
        \end{align*}
        Meanwhile, Lemma~\ref{lem:neighborhood_size} implies that 
        \[
        \P(|N_{\cg}(i,k-1)|\le n^{1-\gamma/2})\ge 1-o(1).
        \]
        Since $n^{1-\gamma/2}<Cn$ for all large enough $n$, we know that when $i$ is in the giant component of size $cn(1+o(1))$ and $|N_{\cg}(i,k-1)|\le n^{1-\gamma/2}$, the set $S_\cg(i,k)$ has to be non-empty. Therefore, we have
        \[
        \P(S_\cg(i,k)\neq \emptyset)\ge C-o(1)=\Theta(1).
        \]
    \end{proof}

\section{Concentration Inequalities and Technical Lemmas}\label{appd:technical}
\subsection{Concentration inequalities}
\begin{lemma}[Bennett's inequality~\citep{bennett1962}]
\label{lem:bennett}
    Suppose $X_1,\ldots,X_n$ are independent random variables with finite second moments. Assume $X_i\le B$ a.s. for each $i\in [n]$. Let $V=\sum_{i=1}^n\E[X_i^2]$. Then, for every $x\ge 0$, 
    \[
    \P\left(\sum_{i=1}^n(X_i-\E[X_i])\ge x\right)\le \exp\left(-\frac{V}{B^2}\cdot \phi(xB/V)\right),
    \]
    where $\phi(t):=(1+t)\log (1+t)-t$ for $t\ge 0$.
\end{lemma}

\begin{lemma}[Chernoff--Hoeffding bound~\citep{dubhashi2009}]
\label{lem:CB}
    Let $X=\sum_{i=1}^nX_i$ be a sum of independent Bernoulli random variables with $\mu=\E[X]$. Then for any $\delta>0$,
    \[
    \P(X\ge(1+\delta)\mu)\le \left(\frac{e^\delta}{(1+\delta)^{1+\delta}}\right)^\mu.
    \]
\end{lemma}
\subsection{Technical lemmas for the Poisson approximation}
\begin{restatable}{lemma}{binpoimultigap}
\label{lem:binom-poi-gap}
Let $n\in \mathbb{N}^+$, $\mu\in\mathbb{R}^+$ and $x,y\in\mathbb{N}$ be such that $x<n$ and $y\le n-x$.
    Then we have
    \[
    \frac{\P(\bin(n-x,\mu/n)=y)}{\P(\poi(\mu)=y)}\le \exp\left(\frac{\mu}{n}(x+y)\right).
    \]
\end{restatable}
\begin{proof}
    By the definition of binomial random variables, we have
    \[
    \P(\bin(n-x,\mu/n)=y)=\binom{n-x}{y}(\mu/n)^y(1-\mu/n)^{n-x-y},
    \]
    and by the definition of Poisson random variables, we have
    \[
    \P(\poi(\mu)=y)=\frac{e^{-\mu}\mu^y}{y!}.
    \]
    It then follows that
    \begin{align*}
        \frac{\P(\bin(n-x,\mu/n)=y)}{\P(\poi(\mu)=y)}&=\frac{(n-x)!}{(n-x-y)!}\left(\frac{1}{n}\right)^ye^\mu\left(1-\frac{\mu}{n}\right)^{n-x-y}\\
        &\le (n-x)^y \left(\frac{1}{n}\right)^y e^\mu\left(1-\frac{\mu}{n}\right)^{n-x-y}\\
        &\le e^\mu\left(1-\frac{\mu}{n}\right)^{n-x-y}\\
        &\le \exp\left(\mu-\frac{\mu}{n}(n-x-y)\right)\\
        &\le \exp\left(-\frac{\mu}{n}(x+y)\right).
    \end{align*}
\end{proof}
\begin{restatable}[Lemma 5 in~\citep{mossel2015}]{lemma}{lembinompoi}
    \label{lem:binom-poi-tv}
        If $m$ and $n$ are positive integers and $c$ is a positive real number independent of $m$ and $n$, then
        \[
        \dtv\left(\bin\left(m,\frac{c}{n}\right),\poi(c)\right)=O\left(\frac{\max\{1,|m-n|\}}{n}\right)
        \]
    \end{restatable}
\begin{proof}
    We only consider the case of $m\le  2n$ as the statement is trivial when $m>2n$. It follows by the classic Poisson approximation of binomial random variables (see for example~\citep{hodges1960}) that
    \begin{equation}
        \label{eq:poi-approx}
        \dtv\left(\bin\left(m,\frac{c}{n}\right),\poi\left(\frac{mc}{n}\right)\right)=O(n^{-1}).
    \end{equation}
    By the triangle inequality, we have
    \begin{align*}
        &\dtv\left(\bin\left(m,\frac{c}{n}\right),\poi\left(c\right)\right)\\
        &\le  \dtv\left(\bin\left(m,\frac{c}{n}\right),\poi\left(\frac{mc}{n}\right)\right)+\dtv\left(\poi(c),\poi\left(\frac{mc}{n}\right)\right).
    \end{align*}
    Therefore, it suffices to show that 
    \[
    \dtv\left(\poi(c),\poi\left(\frac{mc}{n}\right)\right)=O\left(\frac{|m-n|}{n}\right).
    \]
    Consider two real numbers $\lambda<\mu$. We have
    \begin{equation}
        \dtv(\poi(\lambda),\poi(\mu))\le \P(\poi(\mu-\lambda)>0)=1-\exp(-(\mu-\lambda))\le \mu-\lambda\label{eq:tv_poi}
    \end{equation}
    This implies that 
    \[
    \dtv\left(\poi(c),\poi\left(\frac{mc}{n}\right)\right)=O\left(\frac{|m-n|c}{n}\right)=O\left(\frac{|m-n|}{n}\right),
    \]
    and completes the proof.
\end{proof}
\begin{restatable}{lemma}{lemmultipoi}
        \label{lem:multi-poi}
        Let $m$ and $n$ be two positive integers, and $c_1,c_2,c_3$ are constants independent of $m$ and $n$. Suppose we have a sequence of random variables $$(X_1,X_2,X_3,X_4)\sim\mathrm{Multi}(m,c_1/n,c_2/n,c_3/n,1-(c_1+c_2+c_3)/n),$$ and three independent Poisson random variables $X_1'\sim \poi(c_1)$, $X_2'\sim \poi(c_2)$ and $X_3'\sim\poi(c_3)$. Then we have 
        \[
        \dtv(\cl(X_1,X_2,X_3),\cl(X_1',X_2',X_3'))=\co\left(\frac{|m-n|+\log^2 n}{n}\right).
        \]
        
    \end{restatable}
\begin{proof}
    By the definition of the multinomial distribution, random variables $X_1,X_2,X_3$ can be sampled in a sequential manner:
        \begin{enumerate}
            \item Sample $X_1\sim\bin(m,c_1/n)$.
            \item Given $X_1$, sample $X_2\sim\bin(m-X_1,\frac{c_2}{n-c_1})$.
            \item Given $X_1$ and $X_2$, sample $X_3\sim\bin(m-X_1-X_2,\frac{c_3}{n-c_1-c_2})$.
        \end{enumerate}

        By Lemma~\ref{lem:binom-poi-tv}, we can couple $X_1$ and $X_1'$ such that
        \begin{equation}
        \label{eq:X1}
        \P(X_1\neq X_1')\le O\left(\frac{1}{n}\right)+O\left(\left|\frac{mc_1}{n}-c_1\right|\right)=O\left(\frac{1+|m-n|}{n}\right).
        \end{equation}
        By the Chernoff--Hoeffding bound, we know that 
        \begin{equation}
        \label{eq:X1_bound}
        \P(X_1\ge \log^2 n)\le \exp(-C\log^2 n),
        \end{equation}
        for some constant $C>0$. Conditioned on events $\{X_1=X_1'\}$ and $\{X_1\le \log^2 n\}$, we can use the Poisson approximation and~\eqref{eq:tv_poi} to get
        \begin{align}
        &\P(X_2\neq X_2'\cond X_1=X_1',X_1\le \log^2 n)\nonumber\\
        &\le \dtv\left(\bin\left(m-\log^2n,\frac{c_2}{n-c_1}\right),\poi\left(\frac{c_2(m-\log^2n)}{n-c_1}\right)\right)\nonumber\\
        &
        \;\;\;\;+\dtv\left(\poi(c_2),\poi\left(\frac{c_2(m-\log^2n)}{n-c_1}\right)\right)\nonumber\\
        &\le O\left(\frac{1}{n}\right)+O\left(\left|\frac{(m-\log^2 n)c_2}{n-c_1}-c_2\right|\right)\nonumber\\
        &=O\left(\frac{|m-n|+\log^2 n}{n}\right).\label{eq:X2}
        \end{align}
        We can apply the Chernoff--Hoeffding bound to get
        \begin{equation}
        \label{eq:X2_bound}
        \P(X_2\ge \log^2 n\cond X_1=X_1',X_1\le \log^2 n)\le \exp(-C\log^2 n).
        \end{equation}
        For $X_3$, we can similarly get the bound
        \begin{align}
            &\P(X_3\neq X_3'|X_1=X_1',X_2=X_2',X_1\le \log^2 n,X_2\le \log^2 n)\nonumber\\
            &\le O\left(\frac1n\right)+ O\left(\left|\frac{(m-2\log^2 n)c_3}{n-c_1-c_2}-c_3\right|\right)\nonumber\\
            &=O\left(\frac{|m-n|+\log^2 n}{n}\right).\label{eq:X3}
        \end{align}
        Finally, putting~\eqref{eq:X1}-\eqref{eq:X3} together gives
        \[
        \P(X_1=X_1',X_2=X_2',X_3=X_3')\ge 1-O\left(\frac{|m-n|+\log^2 n}{n}\right),
        \]
        and completes the proof.
\end{proof} 

\section{More Details on the Likelihood Ratio Tests}\label{appd:tree_test}
Recall the definition of distributions $\mathbb{Q}_d^{(\lambda,s,\td{s})},\mathbb{P}_{d}^{(\lambda,s,\td{s})}$. The for two unlabeled trees $t$ and $\td{t}$, the likelihood ratio is given by
\[
L_{d}^{(\lambda,s,\td{s})}(t,\td{t})=\frac{\mathbb{P}_{d}^{(\lambda,s,\td{s})}(t,\td{t})}{\mathbb{Q}_d^{(\lambda,s,\td{s})}(t,\td{t})}.
\]

\subsection{Recursive formula of the likelihood ratio}\label{appd:LR_recursive}
The following recursive formula of $L_{d}^{(\lambda,s,\td{s})}(t,\td{t})$ is stated in Proposition 4.3 of~\cite{maier2025}.

First, for the case of $d=0$, 
\[
L_{d}^{(\lambda,s,\td{s})}(t,\td{t})=1.
\]
Now suppose $d\ge 1$. For two unlabeled trees $t$ and $\td{t}$, let $c$ and $\td{c}$ denote their respective root degrees. Denote by $t_1,\ldots,t_{c}$ (resp. $\td{t}_1,\ldots,\td{t}_{\td{c}}$) the subtrees attached to the rooted of $t$ (resp. $\td{t}$). Then the likelihood ratio $L_{d}^{(\lambda,s,\td{s})}$ can be iteratively expressed as 
\begin{equation}
    \label{eq:LR_recursive}
    L_{d}^{(\lambda,s,\td{s})}(t,\td{t})=\sum_{k=0}^{\min\{c,\td{c}\}}\Psi(k,c,\td{c})\sum_{\substack{\sigma\in \mathrm{Inj}([k],[c])\\\td{\sigma}\in \mathrm{Inj}([k],[\td{c}])}}\prod_{i=1}^kL^{(\lambda,s,\td{s})}_{d-1}(t_{\sigma(i)},\td{t}_{\td{\sigma}(i)}),
\end{equation}
where 
\[
\Psi(k,c,\td{c}):=e^{\lambda s\td{s}}\frac{(1-\td{s})^{c-k}(1-s)^{\td{c}-k}}{\lambda^k k!},
\]
and we use $\mathrm{Inj}([m],[m'])$ to denote the set of all injective mappings from $[m]$ to $[m']$ for any $m\le m'$.

\begin{remark}[Time complexity of computing the likelihood ratio]
\label{rem:LR_complexity}
    Let $c_\mathrm{max}$ denote the maximum vertex degree in the two trees $t$ and $\td{t}$. Let $|t|$ and $|\td{t}|$ denote the total number of vertices in $t$ and $\td{t}$ respectively. Because the cardinality of the sets $\mathrm{Inj}([k],[c])$ and $\mathrm{Inj}([k],[\td{c}])$ can both be upper bounded by $c_\mathrm{max}!$, the time complexity of each recursive computation using~\eqref{eq:LR_recursive} is at most $c_\mathrm{max}^2\cdot(c_\mathrm{max}!)^2$. Then the total complexity of computing $L_{d}^{(\lambda,s,\td{s})}(t,\td{t})$ is at most $|t|\cdot|\td{t}|.c_\mathrm{max}^2\cdot(c_\mathrm{max}!)^2$.

    Under the assumption that $c_\mathrm{max}\le \frac{(1+o(1))\log n}{\log\log n}$ and both $t$ and $\td t$ have depth at most $l=\lfloor\sqrt{\log n}\rfloor$, we have $|t|=n^{o(1)}$, $|\td t|=n^{o(1)}$ and $c_\mathrm{max}!=n^{1+o(1)}$. Therefore, the time complexity of each likelihood ratio computation in Algorithm~\ref{alg:mpmatch} is at most $n^{2+o(1)}$.
\end{remark}

\subsection{Feasibility of likelihood ratio tests}\label{appd:LR_feasibility}
The following theorem follows from Theorems 4.9 and 4.10 in~\cite{maier2025}, and it provide sufficient conditions for the likelihood ratio test to achieve a non-vanishing power.

\begin{restatable}{theorem}{LRtest}
\label{thm:tree_test}
    Let $\alpha\approx0.3383$ denote Otter's tree counting constant. Suppose $s\td{s}>\alpha$, $\lambda s\td{s}>1$ and $\lambda>\lambda'(s,\td{s})$, where $\lambda'$ is a constant depending only on $s$ and $\td{s}$. Then the following inequality holds:
    \begin{equation}
    \label{eq:power}
    \liminf_{d\rightarrow \infty} \mathbb{P}^{(\lambda,s,\td{s})}_d(L^{(\lambda,s,\td{s})}_d(\ct,\td{\ct})>\exp(C(\lambda s \td{s})^d))\ge 1-p^{\mathrm{ext}}_{\lambda s\td{s}},
    \end{equation}
    where $C$ is a constant depending only on $\lambda,s$ and $\td{s}$.
\end{restatable}

Recall that $l=\lfloor\sqrt{\log n}\rfloor$ is the an input parameter to Algorithm~\ref{alg:mpmatch} that decides the depth of trees used in the likelihood ratio tests, and $\eps'>0$ is an arbitrarily small constant used in the definition of the probability sequence $p_0,p_1,\ldots$. The following corollary provides the performance guarantee of the likelihood ratio test with the specific parameters chosen in Algorithm~\ref{alg:mpmatch}.

\begin{corollary}
\label{cor:tree_test_finite}
    Suppose $qs^2>\alpha$, $\lambda sq>1$ and  $\lambda>\lambda'(qs,s)$. Then
    \begin{equation}
    \label{eq:LR_power}
    \mathbb{P}^{(\lambda/s,qs,s)}_{l-1}\left(L_{l-1}^{(\lambda/s,qs,s)}(\ct,\td{\ct})>\exp\left(\frac{(\lambda sq)^{l-1}}{\log n}\right)\right)\ge 1-p^\mathrm{ext}_{\lambda sq}-\eps',
    \end{equation}
    for all $n>n_0$, where $n_0$ is a constant depending on $\lambda,s,q$ and $\eps'$. 
\end{corollary}
\begin{proof}
By Theorem~\ref{thm:tree_test}, we have
\[
    \liminf_{d\rightarrow \infty} \mathbb{P}^{(\lambda/s,qs,s)}_d(L^{(\lambda/s,qs,s)}_d(\ct,\td{\ct})>\exp(C(\lambda qs)^d))\ge 1-p^{\mathrm{ext}}_{\lambda qs},
    \]
    where $C$ is a constant depending only on $\lambda,s$ and $q$. Since $l\rightarrow\infty$ as $n\rightarrow\infty$, it follows by the definition of the limit infimum that there exists a constant $n_0'$ depending only on $\epsilon'$ such that 
    \[
    \mathbb{P}^{(\lambda/s,qs,s)}_{l-1}(L^{(\lambda/s,qs,s)}_{l-1}(\ct,\td{\ct})>\exp(C(\lambda qs)^{l-1}))\ge 1-p^{\mathrm{ext}}_{\lambda qs}-\eps',
    \]
    for all $n>n_0'$.
    Meanwhile, since constant $C$ depends only on $\lambda,s$ and $q$, there exists a constant $n_0''$ depending only on $\lambda,s$ and $q$ such that $\frac{1}{\log n}\le C$ for all $n>n_0''$.
    Finally, by setting $n_0=\max\{n_0',n_0''\}$, we get 
    \[
    \mathbb{P}^{(\lambda/s,qs,s)}_{l-1}\left(L_{l-1}^{(\lambda/s,qs,s)}(\ct,\td{\ct})>\exp\left(\frac{(\lambda sq)^{l-1}}{\log n}\right)\right)\ge 1-p^\mathrm{ext}_{\lambda sq}-\eps',
    \]
    for all $n>n_0$.

\end{proof}

\end{document}